\documentclass[11pt]{article}

\newif\ifjasa
\jasafalse

\newif\ifbka
\bkafalse

\newif\ifjmlr
\jmlrfalse

\newif\ifshortsim
\shortsimfalse

\ifjasa
\else
\ifjmlr
\else
\renewcommand{\baselinestretch}{1}
\fi
\fi

\ifjmlr
\usepackage{jmlr2e}
\usepackage{lastpage}

\jmlrheading{25}{2024}{1-\pageref{LastPage}}{7/24}{7/24}{21-0000}{Peter W. MacDonald, Elizaveta Levina and Ji Zhu}

\ShortHeadings{Latent process models for functional networks}{MacDonald, Levina and Zhu}
\firstpageno{1}
\else
\fi

\usepackage{amsmath}
\usepackage{amsfonts}
\usepackage{lscape}
\ifjasa
\usepackage{natbib}
\else
\ifjmlr
\else
\usepackage[square,numbers]{natbib}
\fi
\fi
\usepackage{graphicx}
\graphicspath{{figures/}}
\usepackage{amssymb}
\usepackage{color}
\usepackage{bm}
\ifjmlr
\else
\usepackage{amsthm}
\usepackage[hidelinks]{hyperref}
\fi
\usepackage{enumerate}
\usepackage[nodisplayskipstretch]{setspace}
\usepackage{changepage}
\usepackage[ruled]{algorithm2e}
\usepackage{comment}

\ifjmlr
\else
\theoremstyle{definition}
\newtheorem{theorem}{Theorem}

\newtheorem{proposition}{Proposition}
\newtheorem{corollary}{Corollary}
\newtheorem{lemma}{Lemma}

\fi
\newtheorem{assumption}{Assumption}

\newcommand{\tabby}{\hspace{10pt}}

\newcommand{\iprod}[2]{\left\langle #1 , #2 \right\rangle}

\newcommand{\iid}{\overset{\mathrm{iid}}{\sim}}

\ifbka
  \newcommand{\tp}{\mathrm{\scriptscriptstyle T}}
\else
  \newcommand{\tp}{\intercal}
\fi
\ifbka
  \newcommand{\expect}{E}
\else
  \newcommand{\expect}{\mathbb{E}}
\fi
\ifbka
  \newcommand{\prob}{\mathrm{pr}}
\else
  \newcommand{\prob}{\mathbb{P}}
\fi
\ifbka
  \newcommand{\neweps}{\epsilon}
\else
  \newcommand{\neweps}{\varepsilon}
\fi

\definecolor{purple}{rgb}{0.5,0,1}

\newcommand\numberthis{\addtocounter{equation}{1}\tag{\theequation}}

\newif\ifblind
\blindfalse

\newif\ifthesis
\thesisfalse

\begin{document}

\ifthesis
\title{{\bf CHAPTER 2:} Latent process models for functional network data}
\author{\empty}
\date{\empty}
\else
\ifjmlr
\title{Latent Process Models for Functional Network Data}
\else
\title{Latent process models for functional network data}
\fi
\ifblind
\author{\empty}
\else
\ifjmlr
\author{\name Peter W. MacDonald \email pwmacdon@uwaterloo.ca \\
\addr Department of Statistics \& Actuarial Science \\
University of Waterloo\\
Waterloo, ON N2L 3G1, Canada \\
\AND
\name Elizaveta Levina \email elevina@umich.edu \\
\name Ji Zhu \email jizhu@umich.edu \\
\addr Department of Statistics\\
University of Michigan\\
Ann Arbor, MI 48109-1107, USA}

\editor{My Editor}
\else
\author{Peter W. MacDonald \\ Department of Statistics and Actuarial Science, University of Waterloo \\ ~ \\
  \ifjasa
  \thanks{
  The authors gratefully acknowledge the U.S. National Sciences Foundation and the Rackham Graduate School for funding this research. There are no conflicts of interest to declare.}
  \else
  \fi
  Elizaveta Levina and Ji Zhu \\ Department of Statistics, University of Michigan}
\fi
\fi
\ifjasa
  \date{November 23, 2022}
\else
\ifjmlr
\else
  \date{July 15, 2024}
\fi
\fi
\fi

\ifjasa
\else
\ifthesis
\else
\ifjmlr
\else
\onehalfspacing
\fi
\fi
\fi

\maketitle

\bigskip
\begin{abstract}%
Network data are often sampled with auxiliary information or collected through the observation of a complex system over time, leading to multiple network snapshots indexed by a continuous variable.
Many methods in statistical network analysis are traditionally designed for a single network, and can be applied to an aggregated network in this setting, but that approach can miss important functional structure.    Here we develop an approach to estimating the expected network explicitly as a function of a continuous index, be it time or another indexing variable.
We parameterize the network expectation through low dimensional latent processes, whose components we represent with a fixed, finite-dimensional functional basis.
We derive a gradient descent estimation algorithm,  establish theoretical guarantees for recovery of the low dimensional structure, compare our method to competitors, and apply it to a data set of international political interactions over time, showing our proposed method to adapt well to data, outperform competitors, and provide interpretable and meaningful results.
\end{abstract}

\ifjasa
\else
\ifthesis
\else
\ifjmlr
\else
\fi
\fi
\fi

\section{Introduction}
Modern data are collected in a much greater variety of forms than classical statistics considered, and require novel models and methods.
Networks are one important example of a complex data structure which has received recent interest in many fields of application.
In general, network data on $n$ statistical units, or nodes, describe connections, or edges, between pairs of those units.
This information is stored in 
the $n \times n$ adjacency matrix $A$, where each entry
\ifjasa
\else
$\{A_{ij} : i, j=1,\ldots,n\}$
\fi
describes the connection from node $i$ to node $j$, which could be binary or real-valued.
Much of the statistical network analysis literature deals with a single network, sometimes with auxiliary information, but increasingly samples of networks are also studied.
Samples of networks arise in diverse applications: for instance, in neuroimaging we observe a brain connectivity network for each study subject, and in political science, we observe relations between countries over many years.
In this paper, we focus on functional network data, the setting where edges are indexed by a continuous auxiliary variable.
This variable is commonly time, but it can be anything else that has a natural ordering.

Often, functional network data are collected as repeated, indexed snapshots of a single evolving network.
In this regime, suppose that we observe a network on a common set of $n$ nodes, at $m$ distinct indices $\{x_k\}_{k=1}^m$ in a compact set $\mathcal{X} \subseteq \mathbb{R}$.
The complete data is made up of a collection of indexed adjacency matrix snapshots, $\{A_k\}_{k=1}^m \subseteq \mathbb{R}^{n \times n}$, where each matrix entry $\left[ A_k \right]_{ij}$ gives the value of edge $(i,j)$ from node $i$ to node $j$ for the snapshot corresponding to index $x_k$.
These adjacency matrices may have binary edges with values in $\{0,1\}$, or weighted edges taking any real value, and may be either undirected or directed.
Examples of functional network data collected as time-indexed snapshots include dynamic friendship networks \citep{snijders17stochastic}, or dynamic networks of international conflict (cf. Section~\ref{sec:real_data}).
However, networks may also be indexed in some other continuous covariate; for instance brain images for a collection of subjects, indexed by a single continuous task score \citep{shan22network}.

It is also common that functional network data are collected in the form of time-stamped records of interactions between pairs of nodes.
Examples of functional network data collected in this manner include e-mail networks, bike share networks \citep{matias18semiparametric}, or animal interaction networks \citep{mersch13tracking}.
We can construct snapshots from such event data by partitioning the index set into $m$ contiguous intervals, and constructing weighted adjacency matrices which count the number or total weight of events between nodes $i$ and $j$ in a given interval.

Another way to view adjacency matrix snapshots is to treat them as a multiplex network \citep{kivela14multilayer}, a multilayer network object with $m$ layers corresponding to the snapshots.
However, the ordering of the layers, inherited from the ordering of the indexing variable, gives them additional structure not usually present in a multiplex network.


Much of the existing literature focuses on the dynamic network setting, with the snapshots indexed by time, rather than the general functional setting.
We briefly review recent related work for dynamic networks; see \cite{kim18review} for a more detailed review.
Approaches based on the Stochastic Block Model (SBM) include
\cite{matias16statistical}, which assumes node community memberships follow a Markov chain and allows connection probabilities to vary; \cite{bhattacharyya18spectral}, which assumes community memberships are fixed but allows connection probabilities to vary; and \cite{pensky19spectral}, which allows both community memberships and connection probabilities to vary smoothly in time.
Latent space or latent position models for networks have also been popular since the seminal work of \cite{hoff02latent}.
Latent space approaches to dynamic networks go back to \cite{sarkar05dynamic}, which models latent positions in discrete time with independent Gaussian random walks.
A Bayesian latent space approach with similar random walk transitions on latent positions in discrete time was proposed by \cite{sewell15latent}, and
extended to continuous time by \cite{durante16locally}. 
Both \cite{lee11latent} and \cite{padilla22change} consider extensions of the random dot product graph (RDPG) to the dynamic setting, for the purpose of changepoint detection.
In all of these papers, the single network latent space framework is extended by mapping each node to a sequence or continuum of positions in the latent space, with network snapshots which are conditionally independent given the positions. 
More recently, \cite{athreya22discovering} and \cite{chen24euclidean} consider a similar functional RDPG model, and make a valuable contribution towards the statistical interpretation of a sequence of embeddings of functional network snapshots, in order to reveal the overall changes in network structure over time.
In contrast, our work contributes new, improved methodology for producing such a sequence.

Direct modeling of indexed dyadic connection events has been studied by \cite{perry13point} and \cite{kreiss19nonparametric}, among others, again in the dynamic case.
These two papers model edge event intensities using time-varying edge covariates, with constant and time-varying coefficients, respectively.
A variational Bayes approach to fitting edge event intensities according to a latent community structure was proposed by \cite{matias18semiparametric}.
A similar community-based model was proposed by \cite{arastuie20chip}, focusing on the effect of self-excitation on community detection.

In this paper, we propose a new model for functional network data: a continuously indexed inner product latent position model,  which we will call a {\em latent process model}.
In contrast to many previous approaches for functional latent position models, which treat the positions as random variables, we treat the latent processes as function-valued parameters.
As a result, our approach does not rely on any assumptions of an explicit, simple, and discrete transition model for the latent positions between snapshots \citep{matias16statistical,sarkar05dynamic,sewell15latent,padilla22change}.
Instead, it allows for arbitrary, continuously varying latent processes, with estimation efficency depending primarily on their functional smoothness.
This means that our approach can easily handle irregularly spaced snapshot indices and missing edge entries, with faster computation time than similarly flexible Bayesian methods \citep{durante16locally}.

Our key contribution is to make this function estimation problem tractable by modeling latent processes using a finite, prespecified function basis.
By adaptively selecting the basis based on the data, we are able to share information to efficiently estimate network structure which is shared locally between snapshots, but need not assume that that any part of that structure is common to all snapshots \citep{bhattacharyya18spectral,arroyo21inference}.
In contrast to approaches which first smooth network snapshots and then estimate latent network structure \citep{pensky19spectral}, our estimation approach performs both tasks simultaneously, hence it can adapt to smoothness in the latent structure that may not be directly detectable from the network edges.
Conversely, compared to approaches which estimate latent network structure for each snapshot, then summarize or smooth the output \citep{sannapassino21link,athreya22discovering,chen24euclidean}, our approach can accurately estimate the latent processes as the edge variance increases
\ifshortsim
\else
  (see Figure~\ref{fig:i_varym}),
\fi
as its performance does not rely directly on the signal-to-noise ratio of the individual network snapshots.

\ifjasa
\else
The rest of this paper is organized as follows.  In Section~\ref{sec:model}, we define the latent process network model.
In Section~\ref{sec:estimation}, we develop an estimation algorithm using gradient descent on coordinates in a function basis.
Section~\ref{sec:theory} provides theoretical guarantees for recovery of the latent processes.
Section~\ref{sec:simulations} investigates the proposed methods in simulation studies, and Section~\ref{sec:real_data} applies them to a data set of international political interactions.
Finally, Section~\ref{sec:conclusion} contains brief discussion and future research directions.
\ifjmlr
Mathematical proofs, technical details, and additional simulation studies have been provided as supplementary material.
\else
Proofs and supplemental information are provided in appendices \ref{app:proofs} to \ref{app:smoothing_splines}.
\pagebreak
\fi
\fi

\section{Latent Process Network Models} \label{sec:model}

\ifjasa
\else
\ifjmlr
  In this section, we will introduce latent process network models. To begin, we fix some notation that we will use in the remainder of this paper.
\else
\subsection{Notation} \label{subsec:notation}
\fi
\fi

\ifjmlr
\else
Throughout this paper,
\fi
$\lVert \cdot \rVert_F$ will denote the matrix Frobenius norm, $\lVert \cdot \rVert_2$ the matrix $\ell_2$ operator norm, and $\iprod{\cdot}{\cdot}$ the Frobenius inner product.
When applied to a vector, these coincide with the standard vector $\ell_2$ norm and inner product.
With some abuse of notation, we will analogously use $\lVert \cdot \rVert_F$ to denote the vector $\ell_2$ norm of the vectorization of a 3-mode tensor.
The notation $[\cdot]_{ij}$ denotes the $(i,j)$-th entry of a matrix.
We use $\lambda_{\max}(\cdot)$ and $\lambda_{\min}(\cdot)$ to denote the maximum and minimum eigenvalues of a symmetric matrix;
\ifjasa
\else
The condition number is defined as the ratio of the maximum and minimum singular values of a matrix,
it is bounded between $0$ and $1$;
\fi
and $\mathcal{O}_d$ will denote the set of $d \times d$ orthogonal transformation matrices, which includes rotations, reflections, and combinations of the two.

Capital calligraphic letters are generally reserved for 3-mode tensors, and notation for tensor operations will follow \cite{kolda09tensor}.
The slices of a 3-mode tensor are constructed by fixing the first, second or third index respectively, and varying the other two.
The fibers of a 3-mode tensor are constructed by fixing two indices and varying the third.
For $m=1,2,3$, tensor-matrix multiplication in the $m$th mode is defined by the operator $\times_m$ and matrix multiplies each $m$th mode slice of the first tensor argument by the second matrix argument.
Similarly, for $m=1,2,3$, tensor-vector multiplication in the $m$th mode is defined by the operator $\bar{\times}_{m}$ and matrix-vector multiplies $m$th mode slice of the first tensor argument by the second vector argument.
Note that tensor-vector multiplication in the $m$th mode
results in a matrix
with dimensions corresponding to the other two modes of the original tensor argument.


\subsection{Latent Functional Parameterization}

We parameterize functional networks through a matrix-valued network mean function $\Theta:\mathcal{X} \rightarrow \mathbb{R}^{n \times n}$, such that $\Theta_{ij}(x_k) = \expect( \left[ A_k \right]_{ij})$ for all $i,j$ and $k=1,\ldots,m$.
We assume independent edges, that is,  $[A_k]_{ij}$ are independent for all $i \le j$ and $k$.   Formally, suppose $q(\cdot;\theta,\phi)$ is the edge distribution, parameterized by its mean $\theta$ and some possible nuisance parameters $\phi$.
Then for $1 \leq i,j \leq n$ and $k=1,\ldots,m$,
\ifjasa
\else
\begin{equation*}
  \left[ A_k \right]_{ij} \overset{\mathrm{ind}}{\sim} q(\cdot;\Theta_{ij}(x_k),\phi).
\end{equation*}
\fi
For instance, we could model edges with $q(\cdot;\theta,\sigma) = \mathcal{N}(\theta,\sigma^2)$, in which case the network is fully parameterized by $\Theta$ and the nuisance edge variance $\sigma^2$.
For binary edge networks, we could model $q(\cdot;\theta,\phi) = \operatorname{Bernoulli}(\theta)$, in which case the model is fully parameterized by $\Theta$.
In making the independence assumption, we follow both the single network latent space literature, which typically assumes edge independence conditional on latent positions \citep{athreya18statistical}, and multilayer network latent space models that also make the assumption of independence across layers
\ifblind
 \citep{arroyo21inference}.
\else
 \citep{macdonald22latent}.
\fi
While the independence assumption is likely not exactly correct, it is a common and useful analysis tool for estimating the network structure, and, in this setting, the functional trends in this structure.

For a fixed node pair $(i,j)$, we may also view the sequence of random variables  $\left\{ [A_k]_{ij} \right\}_{k=1}^m$ as a univariate functional response with inputs $x_k$.
The independence of edges over $k$ implies that these responses have a mean which is a function of the continuous index $x$, and independent errors.
Hence the focus of this work is on modeling functional mean structure with no temporal dependence.
In general functional settings, for instance where network snapshots represent brain scans of different patients indexed by a continuous task performance score, independence across snapshots is a reasonable working assumption.
In dynamic networks, independence across snapshots is not guaranteed, but still commonly assumed conditional on latent structure \citep[e.g.,][]{sewell15latent}.
In future work, we may consider allowing within-edge autoregressive errors, or other forms of dependence.
However, accurate nonparametric estimation of the underlying trend component, which we develop in this paper, will be of primary interest in many applications, and such an estimator is necessary to make individual edge sequences stationary before applying time series modeling approaches to the residuals \citep{fan03nonlinear}.

Extending the latent space modeling approach, we assume that the parameter of interest $\Theta$ is determined by the trajectories of each node in a $d$-dimensional latent space, denoted by $Z^{(i)} : \mathcal{X} \rightarrow \mathbb{R}^d$ for $i=1,\ldots,n$. We denote the component functions by $Z^{(i)} = (z_{i,1},\ldots,z_{i,d})^{\tp}$.
Throughout this paper, we assume an inner product similarity function, namely that $\Theta_{ij}(x) = \{ Z^{(i)}(x) \}^{\tp}Z^{(j)}(x)$ for any $1 \leq i,j \leq n$ and $x \in \mathcal{X}$.
We refer to this as an inner product {\em latent process network model}.
Collecting all the latent processes $Z^{(i)}$ into rows of an $n \times d$ matrix-valued function $Z$, we can write  $\Theta(x) = Z(x)Z(x)^{\tp}$. For each $x$, $Z(x)Z(x)^{\tp}$ is a rank $d$ matrix.
The appeal of the inner product similarity is  a parameterization of the network mean function which is low rank for any $x \in \mathcal{X}$.

Evaluating the latent processes at the snapshot indices, we see that our function-valued parameterization produces a tensor decomposition of the $n \times n \times m$ tensor
\ifjasa
\else
  $\expect(\mathcal{A})$
\fi
with $n \times n$ slices $\expect(A_k)$ for $k=1,\ldots,m$, similar to the Tucker or canonical polyadic (CP) decompositions \citep{kolda09tensor}.
However, our formulation does not force the latent structure to factorize in the third mode.
\ifjasa
\else
Consider the simple case where $d=1$, and compare our parameterization
\ifjasa
  $\expect(A_k) = Z(x_k)Z(x_k)^{\tp}$
\else
\begin{equation*}
  \expect(A_k) = Z(x_k)Z(x_k)^{\tp}
\end{equation*}
\fi
to a representation of $\expect(\mathcal{A})$ by a CP decomposition which is symmetric in the first two modes.
The latent process representation cannot in general be reproduced by a rank $1$ CP decomposition, which would require that $\expect(A_k) = w_k \bm{z}\bm{z}^{\tp}$ for an $n$-vector $\bm{z}$ and scalars $w_k$ for $k=1,\ldots,m$.
Under the rank $1$ CP decomposition, the expected value of every network snapshot would share a common eigenvector.
To capture the structure of the latent process representation would require a rank $m$ CP decomposition with different eigenvectors,
and in general no dimension reduction or information sharing across snapshot means.

\fi
While we still assume there is a low rank representation of each snapshot mean, we treat them as functions of the index, and propose methodology with good theoretical and empirical properties when the representations are smooth in $x$.
In other words, we do not force our functional network models to rely on tensor-valued extensions of matrix algebra procedures, recognizing that node modes should be treated differently from the index mode.
In the two node modes with the third index mode fixed, the data are assumed to have low rank or latent position structure, a highly successful and popular approach for network models.
On the other hand, in the third index mode the data are assumed to have smooth structure as a function of the index, as in classical nonparametric regression.

\subsection{Identifiability} \label{sec:identify}

Latent space models with inner product link functions are well known to be non-identifiable due to their invariance to orthogonal transformations of the latent positions, since  $\expect(A) = XX^{\tp} = XO(XO)^{\tp}$ for any $n \times d$ matrix $X$ and $d$-dimensional orthogonal transformation $O \in \mathcal{O}_d$.
This non-identifiability also extends to continuous time:   for any orthogonal matrix-valued function $Q:\mathcal{X} \rightarrow \mathcal{O}_d$, we have
\begin{equation} \label{nonident_eqn}
  Z(x)Z(x)^{\tp} = Z(x)Q(x) \left\{ Z(x)Q(x) \right\}^{\tp}.
\end{equation}
Thus, $Z$ is identifiable only up to uncountably many orthogonal transformations.
In particular, for a given $Z$, we define the unidentified class of latent processes $\mathcal{T}(Z)$ by
\begin{equation} \label{nonident_cont}
  \mathcal{T}(Z) = \left\{ Z(x)Q(x) : Q:\mathcal{X} \rightarrow \mathcal{O}_d \right\}.
\end{equation}
Our goal is to take advantage of smoothness in $Z$ to share information across network snapshots.
The non-identifiability could mask the smoothness, because even if $Z(x)$ is a smooth function of $x$, $Z(x)Q(x)$ may not be.
Although all elements of $\mathcal{T}(Z)$ lead to identically distributed network snapshots, an estimation algorithm which targets the class representative which is ``maximally smooth'' will have the greatest potential to share information across snapshots and improve estimation efficiency.
In general, our theory will instead consider the distance between an estimate $\widehat{Z}$ and an unknown representative of $\mathcal{T}(Z)$.

\section{Estimation With Gradient Descent} \label{sec:estimation}

\ifjmlr
\else
\subsection{Approximation With Finite Dimensional Function Bases} \label{subsec:basis_approx}
\fi

The latent process assumption reduces the parameter space from $n^2$ to $nd$ function-valued parameters.
To further simplify estimation, we will restrict to a finite dimensional parameter space by assuming that each component function of each latent process, $z_{i,r}$,  $i=1,\ldots,n$, $r=1,\ldots,d$ is well approximated by the span of a common $q$-dimensional function basis.
That is, suppose that $B = (B_1,\ldots,B_q)^{\tp}$ is a $q$-dimensional basis of functions each mapping from $\mathcal{X}$ to $\mathbb{R}$.
We assume that every component function $z_{i,r}(x)$ is close to $\bm{w}_{i,r}^{\tp}B(x)$ for some $q$-dimensional coordinate vector $\bm{w}_{i,r}$.
For each $r$, we collect $\bm{w}_{i,r}$ as the rows of an $n \times q$ matrix $\bm{W}_r$. Let $\mathcal{W} = \{\bm{W}_r\}_{r=1}^d$ denote the $n \times q \times d$ tensor containing all the basis coordinates for all nodes in all latent dimensions.
For such a coordinate tensor, the first mode corresponds to the nodes, the second mode to the index space after summarizing from $m$ snapshots into $q$ basis coordinates, and the third mode to the latent space dimensions.

The basis $B$ could be, for instance, a $B$-spline basis on $\mathcal{X}$, or any other similar functional basis.   A $B$-spline basis of order $D \geq 0$ is defined by an increasing sequence of $K$ internal knots, as well as boundary knots.
The dimension of a $B$-spline basis is $q = K + D + 1$. The span of a given $B$-spline basis is a collection of piecewise polynomial functions which are $(D-1)$-times differentiable at the internal knots and smooth elsewhere.
For additional mathematical properties, we refer readers to \cite{schumaker07spline}.
In this paper, we will use order $3$, or cubic $B$-spline bases, although the algorithms to follow would proceed similarly for any function basis $B$, including orthonormal function bases such as the Fourier basis.
The simulations and real data analysis in Sections~\ref{sec:simulations} and \ref{sec:real_data} explore the performance of $B$-splines in detail, but we also compare performance using penalized smoothing splines in Appendix E of the supplementary materials.

For simplicity, we assume a common basis $B$ for all $i=1,\ldots,n$ and all $r=1,\ldots,d$, although this could also be relaxed.
In practice, the underlying latent processes may not exactly belong to $\operatorname{span}(B)$.
However, if the latent processes are smooth in $x$, we will be able to approximate them effectively with functions in $\operatorname{span}(B)$.
In Section~\ref{subsec:tuning}, we develop an approach to choose the dimension, and thus the approximation power of $B$ adaptively from data.

We begin by defining a nonconvex least squares optimization problem equivalent to maximizing a Gaussian likelihood.
In examples in Section~\ref{sec:simulations}, we minimize the same objective function for other edge distributions, in particular the binary edge model with $q(\theta;\phi) = \operatorname{Bernoulli}(\theta)$.
Denote
\begin{equation} \label{snap_opt}
  \ell(\mathcal{W}) = \sum_{k=1}^m \lVert A_k - \sum_{r=1}^d \bm{W}_r B(x_k) B(x_k)^{\tp} \bm{W}_r^{\tp} \rVert_F^2.
\end{equation}
Fixing the latent space dimension $d$, and a $q$-dimensional function basis $B$, we can apply gradient descent over the $n \times q \times d$ tensor-valued argument $\mathcal{W}$
\ifjasa
\else
with $n \times q$ slices $\bm{W}_r \in \mathbb{R}^{n \times q}$ for $r=1,\ldots,d$.
Throughout the paper, we will store the basis coordinates in this way as 3-mode tensors, where the first mode corresponds to the nodes, the second to the index space, and the third to the dimensions of the latent space.
\fi
The following proposition derives the gradient of $\ell$ with respect to each $n \times q$ slice of $\mathcal{W}$. The proof is provided in
\ifjasa
\else
\ifjmlr
  Appendix A of the supplementary materials.
\else
  Appendix~\ref{app:proofs}.
\fi
\fi
\begin{proposition} \label{prop:snap_gradient}
  Define $\ell(\mathcal{W})$ as in \eqref{snap_opt}. Then
  \begin{equation*}
    \frac{\partial \ell}{\partial \bm{W}_r}(\mathcal{W}) \propto  - \sum_{k=1}^m \left\{ A_k - \sum_{r'=1}^d \bm{W}_{r'} B(x_k)B(x_k)^{\tp} \bm{W}_{r'}^{\tp} \right\} \bm{W}_r B(x_k)B(x_k)^{\tp} \in \mathbb{R}^{n \times q}
  \end{equation*}
  for $r=1,\ldots,d$.
\end{proposition}
Since \eqref{snap_opt} is a nonconvex objective, gradient descent will not necessarily converge to the global optimum, and the result depends on the starting value.  In Section~\ref{sec:theory}, we will directly prove results for the output of the gradient descent algorithm proposed below, rather than for the global minimizer of \eqref{snap_opt}.
We will see in Section~\ref{sec:theory} that the starting value for gradient descent may affect the target of estimation among the unidentified class $\mathcal{T}(Z)$ of latent processes defined by \eqref{nonident_cont}.



We propose a gradient descent algorithm which
estimates the $d$ latent dimensions concurrently.
Our concurrent gradient descent algorithm takes initial coordinates
\ifjasa
\else
\begin{equation*}
  \widehat{\mathcal{W}}^0 = \{\widehat{\bm{W}}_r^0\}_{r=1}^d,
\end{equation*}
\fi
step sizes $\eta_{h} > 0$, and a maximum number of iterations $H$ as inputs.
In general, we shall allow the step size to depend on the iteration number $h \geq 0$.
\begin{algorithm} \label{alg:gd_concur}
  \ifshortsim
\else
  \doublespacing
  \fi
  \caption{Concurrent gradient descent algorithm.}
  \begin{tabbing}
  \enspace For $h=1$ to $h=H$ \\
  \qquad For $r=1$ to $r=d$ \\
  \qquad\qquad $\widehat{\bm{W}}_r^{h} \leftarrow \widehat{\bm{W}}_r^{h-1} - \eta_h \frac{\partial \ell}{\partial \bm{W}_r}(\widehat{\mathcal{W}}^{h-1})$ \\
  \enspace Output $\widehat{\mathcal{W}}^{H} = \{\widehat{\bm{W}}_r^H\}_{r=1}^d$
  \end{tabbing}
\end{algorithm}
The output of Algorithm~\ref{alg:gd_concur} is an $n \times q \times d$ tensor-valued coordinate estimator $\widehat{\mathcal{W}}^{H}$.
In practice, we choose $\eta_h$ according to a backtracking search scheme.
We start from fixing the maximum step size, typically $\bar{\eta} = 1/nm$, and try one gradient descent step with size $\bar{\eta}$.
If the objective decreases we accept it and continue, otherwise we try the original step again with size $\bar{\eta}/2$.
This is repeated until we find a step size that decreases the objective.
In the next iteration, we begin with the step size accepted in the previous iteration, and repeat.


Based on an estimator $\widehat{\mathcal{W}}^{H}$ found using our gradient descent scheme, we can use the function basis $B$ to convert back to an estimate of the unknown latent processes. For $x \in \mathcal{X}$, define the $n \times d$ matrix
\ifjasa
\else
\begin{equation*}
  \widehat{Z}^{H}(x) = \widehat{\mathcal{W}}^{H} \bar{\times}_2 B(x),
\end{equation*}
\fi
where we recall that the tensor-vector product operator $\bar{\times}_2$ takes a weighted sum of the $n \times d$ slices of $\widehat{\mathcal{W}}^{H}$ and the elements of $B(x)$ along the second mode, corresponding to the index space.
We will refer to the estimator $\widehat{Z}^{H}$, or simply $\widehat{Z}$, as a functional adjacency spectral embedding (FASE).

The FASE estimator, unlike many other approaches to dynamic networks,  can easily produce an estimate of the unknown latent processes at a value of $x$ where no snapshot is observed.
While the function basis modeling approach taken here is not well suited to forecasting (extrapolation), it should do well for predicting a snapshot at a new index $x$ in the interior of the index space (interpolation).
We evaluate the performance of FASE on this task and compare it to some competing approaches in Appendix C.2 of the supplementary materials.
FASE can also be easily adapted to the case where some edge variables are missing from some snapshots, similar to \cite{macdonald22latent}, by ignoring those triples in the calculation of the least squares objective \eqref{snap_opt}.
\ifblind
\else
FASE is implemented in an R package \texttt{fase}
available on CRAN.
\fi

\subsection{Initializing Gradient Descent} \label{subsec:initializers_new}

As most iterative methods, FASE relies on a suitable initial value.
In this section, we develop a principled initialization approach based on local averaging of network snapshots.

We start with the formal definition of adjacency spectral embedding \citep[ASE][]{tang13universally}. For an $n \times n$ symmetric matrix $M$ with eigendecomposition $Y \Lambda Y^{\tp}$, define
\ifjasa
\else
\begin{equation*}
  \operatorname{ASE}_d(M) = Y_d \Lambda_d^{1/2} \in \mathbb{R}^{n \times d},
\end{equation*}
\fi
where $Y_d \in \mathbb{R}^{n \times d}$ and $\Lambda_d \in \mathbb{R}^{d \times d}$ correspond to the first $d$ eigenvectors and eigenvalues of $M$, ordered according to the
absolute values of the diagonal entries of $\Lambda$.
If the diagonal entries of $\Lambda$ are distinct, then the ASE is uniquely defined up to sign flips of each column.

Under the conditions of the FASE model, suppose the sets $T_{\ell}$ for $\ell=1,\ldots,L$ form a partition of $\{1,\ldots,m\}$ into $L$ contiguous groups of approximately equal size.
We will use this decomposition of the index set to construct our initializer for the unknown coordinates.
First we embed a local mean adjacency matrix using the adjacency spectral embedding, and define, for each $\ell=1,\ldots,L$,
\begin{equation*}
  \widehat{Z}^0_{\ell} = \operatorname{ASE}_d \left( \frac{1}{\lvert T_{\ell} \rvert} \sum_{k \in T_{\ell}} A_k \right) .
\end{equation*}
Then, for each $\ell = 2,\ldots,L$, we perform an additional alignment step which orthogonally transforms the columns of each embedding to minimize the discrepancy with the previous embedding, as measured by the Frobenius norm.
This helps produce an initial smooth set of processes, targeting a smooth representative of the unidentified class $\mathcal{T}(Z)$ and increasing the potential to share information across network snapshots.

Then we set these embeddings as a piecewise constant estimator of the true processes, according to the partition $\{T_{\ell}\}_{\ell=1}^L$.
For each $k=1,\ldots,m$, define
\begin{equation*}
  \widehat{Z}^0(x_k) = \{ \widehat{Z}_{\ell} : k \in T_{\ell} \},
\end{equation*}
and denote them together as an $n \times m \times d$ tensor $\widehat{\mathcal{Z}}^0$.
As for the coordinate tensors, the first mode corresponds to nodes, the second to the index space now exactly corresponding to the snapshots, and the third mode to the latent space dimensions.

The initial coordinates for each $i=1,\ldots,n$ and $r=1,\ldots,d$ are given by the coordinates of the least squares solution in the corresponding spline basis,
\begin{equation} \label{init_zb}
  \widehat{\mathcal{W}}^0 = \widehat{\mathcal{Z}}^0 \times_2 (\bm{B}^{\tp}\bm{B})^{-1} \bm{B}^{\tp},
\end{equation}
where the tensor-matrix product $\times_2$ along the second index space mode maps each of the $m$-dimensional fibers of $\widehat{\mathcal{Z}}^0$ to their best fitting $q$-dimensional basis coordinates.
This procedure is written formally in Algorithm~\ref{alg:init_alg}.

\begin{algorithm} \label{alg:init_alg}
  \ifshortsim
\else
  \doublespacing
  \fi
  \caption{Local average embedding initialization algorithm.}
  \begin{tabbing}
  \enspace Partition $\{1,\ldots,m\}$ into $L$ contiguous subsets $\{T_{\ell}\}_{\ell}^L$ \\ 
  \enspace For $\ell=1$ to $\ell=L$ \\
  \qquad Set $\widehat{Z}^0_{\ell} = \operatorname{ASE}_d \left( \frac{1}{\lvert T_{\ell} \rvert} \sum_{k=1}^m A_k \right)$ \\
  \qquad {\bf Optional:} If $\ell>1$ then \\
  \qquad\qquad Set $Q_{\ell} = \operatorname{argmin}_{Q \in \mathcal{O}_d} \lVert \widehat{Z}^0_{\ell}Q - \widehat{Z}^0_{\ell-1} \rVert^2_F$ \\
  \qquad\qquad Update $\widehat{Z}^0_{\ell} \leftarrow \widehat{Z}^0_{\ell}Q_{\ell}$ \\
  \enspace For $k=1,\ldots,m$ \\
  \qquad Set $\widehat{Z}^0(x_k) = \{ \widehat{Z}^0_{\ell} : k \in T_{\ell}\}$ \\
  \enspace Set $\widehat{\mathcal{Z}}^0 = \left\{ \widehat{Z}^0(x_k) \right\}_{k=1}^m \in \mathbb{R}^{n \times m \times d}$ \\
  \enspace Set $\widehat{\mathcal{W}}^0 = \widehat{\mathcal{Z}}^0 \times_2 (\bm{B}^{\tp}\bm{B})^{-1} \bm{B}^{\tp}$ \\
  \enspace Output $\widehat{\mathcal{W}}^0$
\end{tabbing}
\end{algorithm}

In Section~\ref{subsec:thm_init}, we will identify conditions under which this estimator $\mathcal{W}^0$ is close enough to a corresponding target to provide a good initializer for gradient descent for FASE.
We also use this result to formally motivate a choice of $L$, the number of sets in the partition.

\subsection{Parameter Tuning} \label{subsec:tuning}

Up to this point, we have treated both the latent space dimension $d$, and function basis $B$ as fixed.
In practice these will have to be selected based on data.
To simplify this tuning problem, we consider only cubic $B$-spline 
bases with knots in $\mathcal{X}$ placed at equally spaced quantiles of the snapshot indices $\{x_1,\ldots,x_m\}$.
Thus, basis selection is fully determined by an integer $q$ which is a function of the number of knots.
The parameter $d$ controls static complexity: flexibility of the network mean structure for each fixed $x \in \mathcal{X}$,
while the parameter $q$ controls functional complexity: flexibility of each latent process as a function of the index.

To select these two tuning parameters, we derive a penalized least squares criterion based on the total squared error \eqref{snap_opt} for the observed network snapshots, and the number of parameters ($nqd$) used to define the latent process model.
The details of this derivation are given in
\ifjmlr
\else
\ifjasa
  online Appendix C.
\else
  Appendix~\ref{app:tuning}.
\fi
\fi

In short, the objective \eqref{snap_opt} can be decomposed as the sum of squared residuals for $2n$ linear regression problems, comprising the incoming and outgoing edge values for each node.
Each of these problems is based on $nm/2$ independent observations, and $qd$ unknown coefficients.
We evaluate the generalized cross validation \citep[GCV][]{golub79generalized} criterion for each problem and take the mean, resulting in an overall {\em network GCV} (NGCV) criterion equivalent to
\begin{equation} \label{ngcv}
  \operatorname{NGCV}(q,d) = \log\left\{\frac{\ell(\widehat{\mathcal{W}})}{mn^2} \right\} - 2 \log\left( 1 - \frac{2qd}{nm} \right)
\end{equation}
for an estimated $n \times q \times d$ tensor $\widehat{\mathcal{W}}$ of basis coordinates.
Then, parameter tuning will proceed by minimizing $\operatorname{NGCV}$ over a grid $\{(q,d) : q_{\min} \leq q \leq q_{\max}, \tabby 1 \leq d \leq d_{\max}\}$, with user specified lower and upper bounds on each parameter.

While minimization of NGCV only requires the model to be fit once for each $(q,d)$ pair, this can still be computationally costly for large $m$ and $n$. For comparison, we also propose a more efficient heuristic approach to optimizing over the grid using coordinate descent.
In the coordinate descent scheme, we initialize $d=1$, and perform alternating minimization for $q$ and $d$ by evaluating NGCV on a grid and treating the other as fixed.

We evaluate our tuning approach on synthetic data in
\ifjasa
\else
  Section~\ref{subsec:tuning_sims}.
\fi
It appears to consistently recover the ground truth parameters when the signal-to-noise ratio is sufficiently high.
In other cases, it still tends to select parameters with good performance in terms of recovery of the latent processes.

\section{Theoretical Guarantees} \label{sec:theory}

\ifjasa
\else
\ifjmlr
\else
\subsection{Preliminaries} \label{subsec:prelims}
\fi
\fi

In this section, we establish theoretical guarantees for the error of the gradient descent estimators of the true latent processes, averaged over nodes and snapshots, .
All proofs are provided in
\ifjasa
\else
\ifjmlr
  Appendix A of the supplementary materials.
\else
  Appendix~\ref{app:proofs}.
\fi
\fi
As in Section~\ref{sec:model}, denote the true latent processes by $Z$, an $n \times d$ matrix-valued function of $\mathcal{X}$.
The latent space dimension $d$ will be treated as fixed.
Suppose that after centering, the edge distribution 
is sub-Gaussian with parameter $\sigma$,
and for simplicity assume that the networks are undirected and self loops are allowed.

We also define some notation for the important spectral quantities related to the true latent processes.
Define
\begin{equation*}
   \gamma_Z^2 = \min_{k=1,\ldots,m} \left[ \lambda_{\min}\left\{ Z(x_k)^{\tp}Z(x_k) \right\} \right],
\end{equation*}
and
\begin{equation} \label{Z_conditioning}
  \kappa = \gamma_Z^{-2} \left\{ \max_{k=1,\ldots,m} \left( \lambda_{\max}\left\{ Z(x_k)^{\tp}Z(x_k) \right\} \right) \right\} \geq 1.
\end{equation}
Throughout this section, constants may depend on $d$ or $\kappa$, which are treated as fixed, but will always be free of $n$, $m$, $q$, $\sigma$, and $\gamma_Z$ which will be tracked in the final bounds.
Estimation will proceed by
Algorithm~\ref{alg:gd_concur} on $\ell(\mathcal{W})$ with fixed $q$-dimensional $B$-spline basis $B$.

Let
\ifjasa
\else
\begin{equation*}
  \bm{B} = \begin{pmatrix} B(x_1) & \cdots & B(x_m) \end{pmatrix}^{\tp},
\end{equation*}
\fi
an $m \times q$ matrix which we will refer to as the $B$-spline {\em design matrix}.
Throughout the section, we will make the following assumptions on the basis and associated design matrix.
\begin{assumption} \label{assump:bs}
  \tabby 
  \begin{enumerate}[{(A)}]
    \item Assume that for each $k$, $B(x_k) \geq 0$ element-wise, $\lVert B(x_k) \rVert_1 = 1$, and $B(x_k)$ has at most $2D+1$ nonzero entries.
    \item Assume that $\bm{B}$ satisfies
    \ifjasa
\else
    \begin{equation*}
      \frac{c_Bm}{q} \leq \lambda_{\min}(\bm{B}^{\tp}\bm{B}) \leq \lambda_{\max}(\bm{B}^{\tp}\bm{B}) \leq \frac{C_Bm}{q}
    \end{equation*}
    \fi
    for constants $C_B > c_B > 0$.
  \end{enumerate}
\end{assumption}
Part (A) is satisfied by $B$-spline bases of fixed order $D \geq 0$.
Moreover, because $B$-splines are locally supported, if the snapshot indices and basis knots are evenly spread across $\mathcal{X}$, it follows that $\operatorname{tr}(\bm{B}^{\tp}\bm{B}) \sim m$.
Thus part (B) simply requires that $\bm{B}^{\tp}\bm{B}$ has a condition number of a constant order, a condition which holds if the support of every basis function contains on the order of $m/q$ snapshot indices, for growing $m$ and $q$.

\subsection{Results for Algorithm~\ref{alg:gd_concur}} \label{subsec:thm_concur}

In this subsection, we will establish an asymptotic bound on the error of the FASE estimator as computed by Algorithm~\ref{alg:gd_concur}, up to an unknown orthogonal transformation.
Recall that in general we can only identify an unknown representative of $\mathcal{T}(Z)$.
Thus, as in single layer latent space approaches \citep{ma20universal}, for a given iteration $h$ of gradient descent, there may be a different representative of $\mathcal{T}(Z)$ which minimizes the Frobenius norm error over the unidentified class.

In the following, we will only need the true processes evaluated at the snapshot indices, so we store these in an $n \times m \times d$ tensor $\mathcal{Z}$ with $n \times d$ slices given by $Z(x_k)$ for $k=1,\ldots,m$.
We also define the unidentified class of process snapshots by
\begin{equation*}
  \mathcal{T}^m(\mathcal{Z}) = \{ Z(x_k)Q_k : Q_k \in \mathcal{O}_d, k=1,\ldots,m \} \subset \mathbb{R}^{n \times m \times d}.
\end{equation*}
These process snapshots produce the same expected adjacency matrix snapshots as the original $\mathcal{Z}$ and are thus indistinguishable with respect to the least squares objective \eqref{snap_opt}.
As in Section~\ref{subsec:initializers_new}, the first mode corresponds to nodes, the second to the index space, and the third mode to the latent space dimensions.

To prove a result about Algorithm~\ref{alg:gd_concur}, we must control the intrinsic approximation error introduced by restricting our estimates to a finite-dimensional function basis.
This error will depend on $\mathcal{T}^m(\mathcal{Z})$ and the functions in $\operatorname{span}(B)$, as well as, due to nonidentifiability of $\mathcal{Z}$, on the current iterate of gradient descent.

Define a map from a set of coordinates to the space of snapshots,
\begin{align*}
  \mathcal{R}_{\bm{B},\mathcal{Z}}(W) &= \operatorname{argmin}_{Z \in \mathcal{T}^m(\mathcal{Z})} \lVert W - Z \times_2 (\bm{B}^{\tp}\bm{B})^{-1} \bm{B}^{\tp} \rVert_F^2 .
\end{align*}
In words, $\mathcal{R}_{\bm{B},\mathcal{Z}}$ takes a set of coordinates and finds the best-aligned orthogonal transformation of the true processes.
The tensor product in the second term inside the norm can be interpreted as finding a least squares solution with respect to the $B$-spline design matrix, as in \eqref{init_zb}.


Then $\mathcal{R}_{\bm{B},\mathcal{Z}}(\widehat{\mathcal{W}}^{h})$ is an $n \times m \times d$ tensor which we call the {\em snapshot target} at iteration $h$; it is an element of $\mathcal{T}^m(\mathcal{Z})$, so for each $k=1,\ldots,m$ and $h \geq 0$, we can find $Q_k^{*,h} \in \mathcal{O}_d$ such that $Z(x_k) Q_k^{*,h}$ is the $k$th $n \times d$ slice of $\mathcal{R}_{\bm{B},\mathcal{Z}}(\widehat{W}^{h})$.
We also define
\begin{equation} \label{def_target}
  \mathcal{W}^{*,h} = \mathcal{R}_{\bm{B},\mathcal{Z}} (\widehat{\mathcal{W}}^{h}) \times_2 (\bm{B}^{\tp}\bm{B})^{-1} \bm{B}^{\tp} \in \mathbb{R}^{n \times q \times d},
\end{equation}
which we call the {\em coordinate target} at iteration $h$.
Control of the approximation error at each iteration will naturally depend on how well the coordinate target approximates the snapshot target, through
\begin{align}
  \neweps^{(h)}_{\mathrm{approx},2} &= \frac{1}{m} \sum_{k=1}^m \lVert \mathcal{W}^{*,h} \bar{\times}_2 B(x_k) - Z(x_k) Q_k^{*,h} \rVert_F^2 \nonumber \\ 
  \neweps^{(h)}_{\mathrm{approx},\infty} &= \max_{k=1,\ldots,m} \lVert \mathcal{W}^{*,h} \bar{\times}_2 B(x_k) - Z(x_k) Q_k^{*,h} \rVert_F^2 \nonumber 
\end{align}
for $h \geq 0$, which describe average and maximum approximation errors over the snapshot times, respectively.
In the case of $B$-splines, the approximation errors will be small if the target in snapshot space is made up of latent processes which are smooth in $x$ \citep{schumaker07spline}.

Our main result will require that the approximation errors are controlled uniformly over $h$ with high probability.
While this condition cannot be verified, empirically we find these quantities do not tend to increase in $h$, and thus if the initializer induces a sufficiently smooth snapshot target, so too will the gradient descent iterates.


With these definitions in hand, we are now ready to state the assumptions required for our main result.
First, a condition on a properly scaled signal-to-noise ratio.

\begin{assumption} \label{assump:snr}
  \begin{equation*}
    \frac{\sigma^2 q^5n\log q}{m \gamma_Z^4} = o(1) .
  \end{equation*}
\end{assumption}


Additionally, we require asymptotic conditions on the approximation and initialization errors.
\begin{assumption} \label{assump:approxes}
  \begin{align*}
    \operatorname{sup}_{h \geq 0} \neweps^{(h)}_{\mathrm{approx},2} &= o_{\mathbb{P}}(\gamma_Z^2 / q) \\
    \operatorname{sup}_{h \geq 0} \neweps^{(h)}_{\mathrm{approx},\infty} &= o_{\mathbb{P}}(\gamma_Z^2)
  \end{align*}
\end{assumption}

\begin{assumption} \label{assump:init}
  \begin{equation*}
    \lVert \widehat{\mathcal{W}}^{0} - \mathcal{W}^{*,0} \rVert_F^2 = o_{\mathbb{P}}(\gamma_Z^2)
  \end{equation*}
\end{assumption}


All three of these assumptions appear in the proof to ensure a contraction property of the iterates in coordinate space, in particular that the discrepancy
\begin{equation*}
  \lVert \widehat{\mathcal{W}}^{h} - \mathcal{W}^{*,h} \rVert_F^2
\end{equation*}
between the current estimate and the current coordinate target shrinks as $h$ increases.
Empirically, we see that gradient descent does indeed converge as $h$ increases, but in theory this contraction does not guarantee convergence of $\widehat{\mathcal{W}}^{h}$.
Hence, our result is written in terms of a $\limsup$, which exists even if gradient descent does not converge.

Assumption~\ref{assump:approxes} puts a requirement on the approximation errors, and Assumption~\ref{assump:init} puts a requirement on the initialization error.
In the simplest case for a basis of dimension $q=1$, the latent processes are modeled as constant in $x$.
Then, $\neweps_{\mathrm{approx},2}^{(h)}$ is the total sample variance of the snapshot target processes around their means, and we require that this is asymptotically smaller than the squared magnitude of the process.
For larger $q$, we need additional parameters to proportionally reduce this variation.
Similarly, we require that the maximum squared discrepancy over $k=1,\ldots,m$ is smaller than the squared magnitude of the process.
Finally, Assumption~\ref{assump:init} requires that as $n$ increases, the initialization error is small compared to the magnitude of the processes.
Validity of this assumption is addressed in Section~\ref{subsec:thm_init}.




With these assumptions, we are ready to state the main result of this section.
Recall that in this asymptotic regime, we have $n \rightarrow \infty$, and allow $m$, $q$, $\gamma_Z$, and $\sigma$ to possibly grow as functions of $n$.


\begin{theorem} \label{thm:main_concur}
  Suppose $\{A_k\}_{k=1}^m$ are generated from a latent process network model, with independent sub-Gaussian edges with parameter $\sigma$.
  Suppose we compute a FASE estimator using Algorithm~\ref{alg:gd_concur} with $q$-dimensional $B$-spline basis $B$ and step size $\eta_{h} \equiv \eta'q/m \gamma_Z^2$ for a constant $\eta'$.
  Suppose Assumptions \ref{assump:bs}, \ref{assump:snr}, \ref{assump:approxes}, and \ref{assump:init} hold.
  Then if the sequence of average approximation errors $\neweps^{(h)}_{\mathrm{approx},2}$ satisfies
  \begin{equation} \label{main_approx}
    \limsup_{h \rightarrow \infty} \neweps^{(h)}_{\mathrm{approx},2} = O_{\mathbb{P}}(\alpha_n),
  \end{equation}
  it follows that the estimation errors satisfy
  \begin{equation} \label{main_concur_concl}
  \limsup_{h \rightarrow \infty} \frac{1}{mn} \sum_{k=1}^m \lVert \widehat{Z}^{h}(x_k) - Z(x_k)Q_k^{*,h} \rVert_F^2 = O_{\mathbb{P}} \left(  \frac{\sigma^2 q^4 \log q}{\gamma_Z^2 m} + \frac{\alpha_n}{n} \right) ,
  \end{equation}
  where $Z(x_k)Q_k^{*,h}$ is the $k$th $n \times d$ slice of the snapshot target at iteration $h$.
\end{theorem}

The left hand side of \eqref{main_concur_concl} is an upper bound on the error of the limiting FASE estimator output by Algorithm~\ref{alg:gd_concur}, averaged over the snapshot indices and nodes.
The right hand side of \eqref{main_concur_concl} can be interpreted as a statistical error term and an approximation bias term.
The constant $\eta'$ is derived explicitly in the proof, and depends on the problem dimension, basis design, and spectral properties of $Z$.
In practice we do not attempt to use this step size, instead selecting it adaptively as described in Section~\ref{sec:estimation}.

In this setting, assuming that each latent process is in $\operatorname{span}(B)$ will not help establish consistency, as in general this assumption will no longer hold after applying the unknown orthogonal transformation, leading to nonzero approximation bias.
However, in the limit the average approximation bias may be bounded above by $\limsup_{h \rightarrow \infty} \neweps^{(h)}_{\mathrm{approx},2}$.
Thus \eqref{main_concur_concl} shows that the error in estimating the unknown latent processes inherits the rate $\alpha_n$ from the intrisic basis approximation error.
Based on the normalization of $\neweps^{(h)}_{\mathrm{approx},2}$, we expect $\alpha_n/n$ to be approximately constant in $n$, $m$, and $\sigma$; but decrease in $q$ as the function basis becomes more flexible.
Hence, there should be an optimal choice of $q$ for which the right hand side of \eqref{main_concur_concl} goes to zero asymptotically, and the FASE estimator is consistent.
To better understand this tradeoff, we state a more interpretable result in a special case where $d=1$.

When $d=1$, the unknown rotations at each $x_k$ reduce to sign flips, and we can show that when the initialization error is small, gradient descent will contract towards a special, deterministic coordinate target.
Define
\begin{equation*}
  \bm{W}_1^* = \operatorname{argmin}_{W \in \mathbb{R}^{n \times q}} \sum_{k=1}^m \lVert WB(x_k) - Z_1(x_k) \rVert_2^2
\end{equation*}
and state slightly altered assumptions on the approximation and initialization error in terms of this new $\bm{W}_1^*$.
\begin{assumption} \label{assump:approxes1}
  \begin{align*}
    \frac{1}{m}\sum_{k=1}^m \lVert \bm{W}_1^* B(x_k) - Z_1(x_k) \rVert_2^2 &= o(\gamma_Z^2 / q)  , \\
    \max_{k=1,\ldots,m} \lVert \bm{W}_1^* B(x_k) - Z_1(x_k) \rVert_2^2 &= o(\gamma_Z^2) .
  \end{align*}
\end{assumption}
\begin{assumption} \label{assump:init1}
  \begin{equation*}
    \lVert \widehat{\bm{W}}_1^{0} - \bm{W}_1^{*} \rVert_F^2 = o_{\mathbb{P}}(\gamma_Z^2) .
  \end{equation*}
\end{assumption}
Then we have the following stronger result, stated as a corollary of Theorem~\ref{thm:main_concur}.

\begin{corollary} \label{cor:main_concur1}
  Suppose $\{A_k\}_{k=1}^m$ are generated from a latent process network model with $d=1$, and independent sub-Gaussian edges with parameter $\sigma$.
  Suppose we compute a FASE estimator using Algorithm~\ref{alg:gd_concur} with $q$-dimensional $B$-spline basis $B$ and step size $\eta_{h} \equiv \eta''q/m \gamma_Z^2$ for a constant $\eta''$.
  Suppose Assumptions \ref{assump:bs}, \ref{assump:snr}, \ref{assump:approxes1}, and \ref{assump:init1} hold.
  Then if the approximation error satisfies
  \begin{equation} \label{main_approx1}
     \frac{1}{m}\sum_{k=1}^m \lVert \bm{W}_1^* B(x_k) - Z_1(x_k) \rVert_2^2 = O(\alpha'_n),
  \end{equation}
  it follows that the estimation error satisfies
  \begin{equation} \label{main_concur_concl1}
  \limsup_{h \rightarrow \infty} \frac{1}{mn} \sum_{k=1}^m \lVert \widehat{Z}_1^{h}(x_k) - Z_1(x_k) \rVert_F^2 = O_{\mathbb{P}} \left(  \frac{\sigma^2 q^4 \log q}{\gamma_Z^2 m} + \frac{\alpha'_n}{n} \right).
  \end{equation}
\end{corollary}
If we make a parametric assumption that each orthogonalized latent process is in $\operatorname{span}(B)$, we have $\bm{W}_1^*B(x_k) = Z_1(x_k)$ for all $k=1,\ldots,m$, so that $\alpha'_n=0$, and get a consistent estimator on average over the nodes and snapshot indices.
In the corresponding nonparametric setting, suppose that each component function $z_{i,1}$ is twice differentiable.
Then by approximation results for cubic $B$-splines \citep{schumaker07spline}, there exists $w \in \mathbb{R}^q$ such that
\begin{equation*}
   \sup_{x \in \mathcal{X}} \left\lvert z_{i,1}(x) - w^{\tp}B(x) \right\rvert \tabby \lesssim \frac{1}{q^{2}} \cdot \sup_{x \in \mathcal{X}} \left\lvert \frac{\partial^2 z_{i,1}(x)}{\partial x^2} \right\rvert.
\end{equation*}
Thus, if the true processes have uniformly bounded second derivatives, we have $\alpha'_n \lesssim n/q^2$, and there exists a theoretically optimal $q$ which grows with $m$ and $n$ such that the FASE estimator is consistent for $Z_1$ on average over the nodes and snapshot indices.

This guarantee on the alignment of the true and estimated latent processes for $d=1$ motivates a sequential estimator of FASE for higher-dimensional models, which estimates one latent dimension at a time.
We describe this approach in Appendix F of the supplementary materials, and implement it in the R package \texttt{fase}.
However, the sequential estimator relies on strong assumptions on separation of the singular values of each $Z(x)$ (uniformly in $x$), so in general we recommend the concurrent gradient descent estimator for FASE presented in Section~\ref{sec:estimation}, which we use for all the simulation and real data results to follow.

\subsection{Results for initialization} \label{subsec:thm_init}

In this section we state the main result for our proposed initializer, a high probability upper bound for recovery of the initial coordinate target, as defined in Section~\ref{subsec:thm_concur}.
Then, we derive an optimal choice for $L$, the number of sets in the partition of the index set, as defined in Section~\ref{subsec:initializers_new}.

For simplicity, in this section we assume the snapshot indices are equally spaced on $\mathcal{X} = [0,1]$, that is, $x_k = k/m$ for $k=1,\ldots,m$, and $T_{\ell}$ splits the indices into equal sized contiguous subsets.  Supposing $m/L$ is an integer, $\lvert T_{\ell} \rvert = m/L$ for all $\ell=1,\ldots,L$, and two indices in the same $T_{\ell}$ are separated by at most $1/L$ in the index space.

We also assume a Lipschitz condition on the latent processes $z_{i,r}(x)$ as a function of $x$.
\begin{assumption} \label{assump:init_lipschitz}
  Suppose that for $i=1,\ldots,n$ and $r=1,\ldots,d$, each latent process $z_{i,r}(x)$ satisfies
  \begin{equation*}
    \sup_{x,y \in \mathcal{X}} \lvert z_{i,r}(x) - z_{i,r}(y) \rvert \leq K_1 \lvert x - y \rvert
  \end{equation*}
  for a uniform a constant $K_1 > 0$.
\end{assumption}
Note that this condition need only hold for a particular orthogonal transformation of the latent processes, not for any element of the unidentified class $\mathcal{T}(Z)$.

With these assumptions, we can state the main result of this section.
\begin{proposition} \label{prop:init_concur}
  Suppose $\{A_k\}_{k=1}^m$ are generated from a latent process network model, with independent sub-Gaussian edges with parameter $\sigma$.
  Define the initializer $\widehat{\mathcal{W}}^0$ as in \eqref{init_zb}, with parameter $L$.  
  Suppose Assumption \ref{assump:init_lipschitz} holds.
  Then the initializer $\widehat{\mathcal{W}}^0$ satisfies
  \begin{equation*}
    \lVert \widehat{\mathcal{W}}^0 - \mathcal{W}^{*,0} \rVert_F^2 \leq \left\{ \frac{C_B q}{c_B^2} \left( \frac{(2\sqrt{10d} + 1)K_1 \sqrt{d n}}{\gamma_Z L} + \frac{c_{\mathrm{prob}} \sigma \sqrt{10 d L n}}{\gamma_Z^2 \sqrt{m}} \right)^2 \right\} \gamma_Z^2
  \end{equation*}
  with probability at least $1 - 4L \exp(-n)$, where $\mathcal{W}^{*,0}$ is the associated coordinate target for $\widehat{\mathcal{W}}^0$, as defined in \eqref{def_target}, and $c_{\mathrm{prob}}$ is a universal constant.
\end{proposition}

Towards justifying the asymptotic rate in Assumption~\ref{assump:init}, 
we derive an optimal choice of $L$ (up to constant factors) and analyze the resulting effect on initialization error.
Treating $C_B$, $c_B$, $K_1$, and $d$ as constants,
we have, with high probability,
\begin{equation*}
  \frac{1}{\gamma_Z} \lVert \widehat{\mathcal{W}}^0 - \mathcal{W}^{*,0} \rVert_F \lesssim \frac{\sqrt{qn}}{L \gamma_Z} + \frac{\sigma \sqrt{qLn}}{\gamma_Z^2 \sqrt{m}} .
\end{equation*}

Some algebra shows that the first two terms of the upper bound are minimized for
\begin{equation*}
  \hat{L} \sim \left( \frac{\gamma_Z \sqrt{m}}{\sigma} \right)^{2/3},
\end{equation*}
and plugging this in gives optimized upper bound
\begin{equation*} 
  \frac{1}{\gamma_Z} \lVert \widehat{\mathcal{W}}^0 - \mathcal{W}^{*,0} \rVert_F \lesssim \frac{\sigma^{2/3} (qn)^{1/2}}{m^{1/3}\gamma_Z^{5/3}}.
\end{equation*}
If this upper bound on relative error goes to zero, then the initialization assumption will hold asymptotically with high probability.

\section{Evaluation on Synthetic Networks} \label{sec:simulations}


\ifjasa
\else
\subsection{Latent Process Recovery} \label{subsec:error_sims}
\fi


\ifjmlr
\else
We compare our FASE estimator to existing methods for similar inner product latent space models for both weighted and binary edge networks.
\fi
We will implement FASE (Algorithm~\ref{alg:gd_concur}) with both an oracle and an adaptive NGCV tuning scheme, as well as three
\ifjasa
\else
  ASE-based approaches
\fi
that have been applied in the past to functional, in particular dynamic, network data \citep{sannapassino21link}.
First, we apply the usual $d$-dimensional ASE to each of the $m$ adjacency matrix snapshots.
Second, we apply the omnibus ASE (OMNI) \citep{levin17central}, which finds a $d$-dimensional embedding at each snapshot index based on the ASE of
\ifshortsim
\else
  the so-called {\em omnibus} matrix given by
  \begin{equation*}
    \begin{pmatrix}
      A_1 & \frac{A_1 + A_2}{2} & \cdots & \frac{A_1 + A_m}{2} \\
      \frac{A_1 + A_2}{2} & A_2 & & \vdots \\
      \vdots & & \ddots & & \\
      \frac{A_1 + A_m}{2} & \cdots & & A_m \\
    \end{pmatrix}.
  \end{equation*}
\fi
Third, we apply the multiple ASE (COSIE) \citep{arroyo21inference}, which assumes that the expectations of the adjacency matrix snapshots share a common invariant subspace.
For these baseline estimators, we assume oracle knowledge of $d$.

We generate instances of the latent process network model under three scenarios.
\begin{enumerate}[{(i)}]
  \item Parametric Gaussian network with $B$-spline processes. In this scenario we generate each component process $z_{i,r}$ for $i=1,\ldots,n$ and $r=1,\ldots,d$ from a cubic $B$-spline basis $B$ on $[0,1]$, with equally spaced knots and dimension 10.
  In particular we generate $\bm{w}_{i,r} \sim \mathcal{N}_{10}(0,I_{10})$, and then define $z_{i,r}(x) = \bm{w}_{i,r}^{\tp}B(x)$ for $x \in [0,1]$.
  Then for equally spaced snapshot indices $0=x_1 < \cdots < x_m=1$, we set
  \ifshortsim
\else
  \begin{equation*}
    A_k = Z(x_k)Z(x_k)^{\tp} + E_k,
  \end{equation*}
  \fi
  where each $E_k$ is a symmetric matrix of independent Gaussian random variables with variance $\sigma^2$.
  \item Nonparametric Gaussian network with sinusoidal processes. In this scenario we generate each component process as
  \begin{equation*}
    z_{i,r}(x) = \frac{3 \sin [2\pi (2x - U_{i,r})]}{1 + 5[x + B_{i,r}(1-2x)]} + G_{i,r}
  \end{equation*}
  where $U_{i,r} \iid \mathrm{Unif}[0,1]$, $B_{i,r} \iid \mathrm{Bernoulli}(1/2)$, and $G_{i,r} \iid \mathcal{N}(0,1/4)$.
  \ifshortsim
  \else
  Each process is a shifted (according to $G$) sine function which goes through 2 full cycles from a random starting point (controlled by $U$), with amplitude either increasing or decreasing (controlled by $B$) from $3$ to $1/2$ or $1/2$ to $3$.
  \fi
  Then for equally spaced snapshot indices $0=x_1 < \cdots < x_m=1$, we set
  \ifshortsim
\else
  \begin{equation*}
    A_k = Z(x_k)Z(x_k)^{\tp} + E_k,
  \end{equation*}
  \fi
  where each $E_k$ is a symmetric matrix of independent Gaussian random variables with variance $\sigma^2$.
  \item Parametric RDPG network with $B$-spline processes.
  In this scenario we generate each component process $z_{i,r}$ for $i=1,\ldots,n$ and $r=1,\ldots,d$ from a cubic $B$-spline basis $B$ on $[0,1]$, with equally spaced knots and dimension $10$.
  We generate each $d$-dimensional fiber of the full $n \times 10 \times d$ coordinate tensor $\mathcal{W}$ as an independent Dirichlet random variable with parameter $[0.1 \cdots 0.1]^{\tp}$.
  The coordinates are then rescaled to control the overall network density.
  Then for $i \leq j$, we generate
  \begin{equation*}
    \left[ A_k \right]_{ij} \sim \operatorname{Bernoulli} \left\{ \sum_{r=1}^d z_{i,r}(x_k)z_{j,r}(x_k) \right\}
  \end{equation*}
  and set $\left[ A_k \right]_{ji}$ to make each $A_k$ symmetric.
\end{enumerate}

To compare the performance of FASE against the baseline estimators, we evaluate error for recovery of the latent processes up to orthogonal transformation, averaged over the snapshot indices:
\begin{equation*}
  \mathrm{Err}_Z(\widehat{Z}) = \left\{
  \frac{1}{ndm} \sum_{k=1}^m \min_{Q_k \in \mathcal{O}_d} \lVert \widehat{Z}(x_k) -  Z(x_k)Q_k \rVert_F^2 \right\}^{1/2}.
\end{equation*}
This error metric is similar to the error bounded in the conclusion of Theorem~\ref{thm:main_concur}.
If our adaptive implementation of FASE selects $d$ incorrectly, then either the true or estimated latent processes are given additional columns of all zeros so that the dimensions match.
We report two errors for the FASE estimator found using Algorithm~\ref{alg:gd_concur}, one for the adaptive version which selects $d$ and $q$ using a grid search with candidates $q=6,8,\ldots,16$ and $d=1,2,\ldots,6$, and the NGCV criterion defined in Section~\ref{subsec:tuning} (FASE (NGCV)); and another oracle FASE estimator which fits the model with ground truth $d$, and $q$ selected to minimize $\mathrm{Err}_Z$ (FASE (ORC)). 
In all of the following plots, vertical lines at each point denote plus and minus $2$ standard errors over the independent replications.


In
\ifshortsim
\else
  Figures~\ref{fig:i_varym} and \ref{fig:i_varyn},
\fi
we report results for scenario (i) generated with
\ifshortsim
\else
  $\sigma \in \{2,4,6,8\}$.
\fi
In
\ifshortsim
\else
  Figure~\ref{fig:i_varym}
\fi
we vary the number of snapshots $m \in \{20,40,\ldots,200\}$ for fixed $n=100$ and $d=2$, and in
\ifshortsim
\else
  Figure~\ref{fig:i_varyn}
\fi
we vary the number of nodes $n \in \{80,120,\ldots,400\}$ for fixed $m=80$ and $d=2$.
In all settings, $\mathrm{Err}_Z$ is averaged over 50 independent replications.
In
\ifshortsim
\else
  Figure~\ref{fig:i_varym}, in all four panels, none
\fi
of the baseline ASE estimators show an improvement with increasing $m$, while FASE does.
Note that for $\sigma=4$ the plotted points for COSIE are not visible, as the performance coincides with that of OMNI.
In
\ifshortsim
\else
  Figure~\ref{fig:i_varyn},
\fi
only FASE and ASE show improvement with increasing $n$.
While ASE improves at a faster rate than FASE, it never outperforms it, even for the largest values of $n$ considered.
In almost all settings, FASE performs the best of all methods.
This is true regardless of whether tune $d$ and $q$ adaptively, as the errors for FASE (NGCV) and FASE (ORC) are nearly indistinguishable in these plots.
Among the baselines, the errors for ASE, which is unbiased, are by far the most sensitive to $\sigma$.
On the other hand, COSIE and OMNI, which share information globally, incur a lot of bias in this setting, where the latent processes are only similar locally in the index variable.
\ifshortsim
\else
  In Figure~\ref{fig:i_varym} for $\sigma=8$ and $m=20$,
\fi
the adaptive FASE estimator is able to outperform the oracle on average.
This phenomenon, which we discuss in more detail in
\ifjasa
\else
  Section~\ref{subsec:tuning_sims},
\fi
can occur when the signal is low and the adaptive estimator selects a value of $d$ which achieves better error than the ground truth $d$.
\ifshortsim
\else
\begin{figure}
  \centering
  \includegraphics[width=\textwidth]{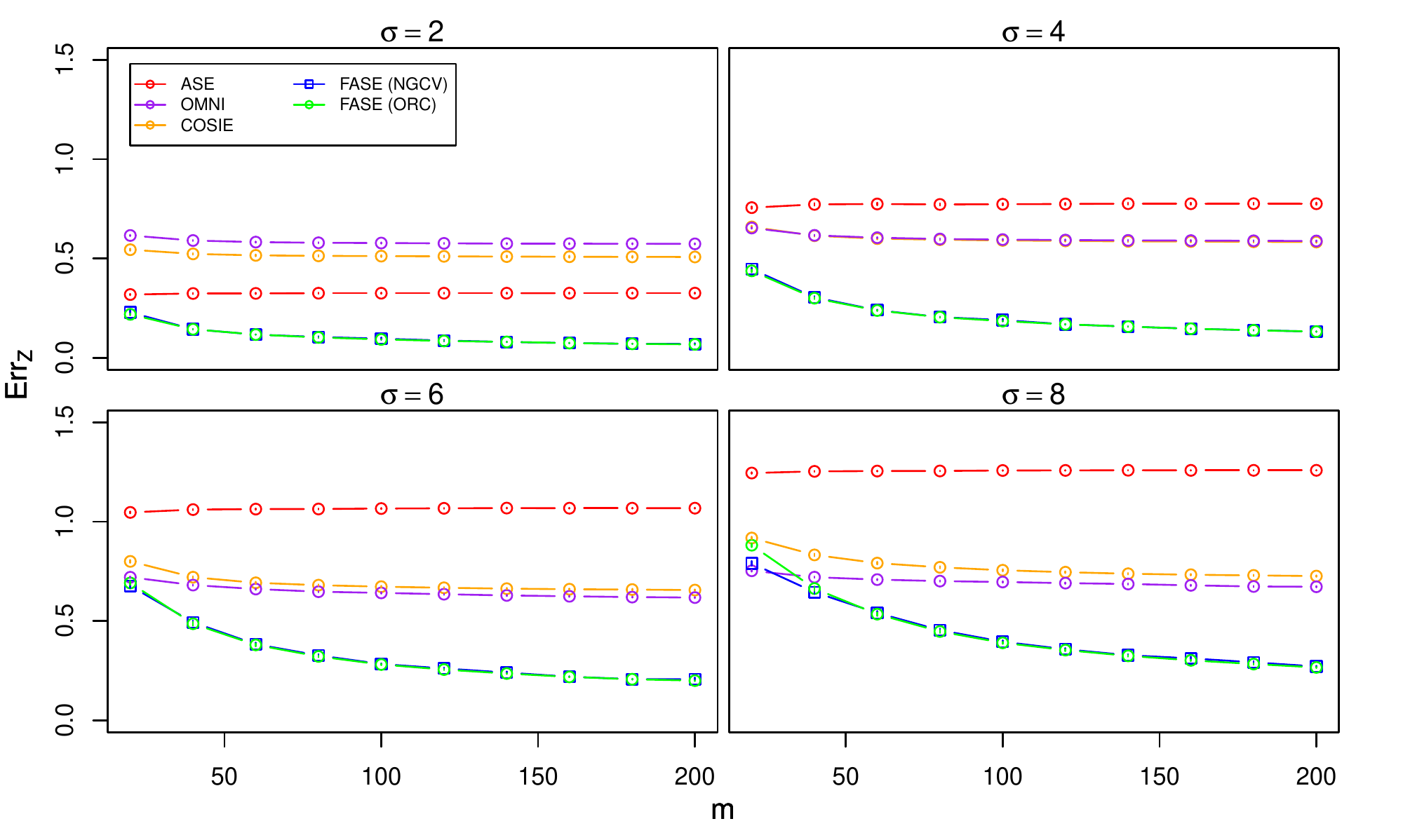}
  \caption{Mean of $\mathrm{Err}_Z$, varying $m$, the number of snapshots. Scenario (i), parametric Gaussian networks. Plots are labeled by edge standard deviation $\sigma$.}
  \label{fig:i_varym}
\end{figure}
\begin{figure}
  \centering
  \includegraphics[width=\textwidth]{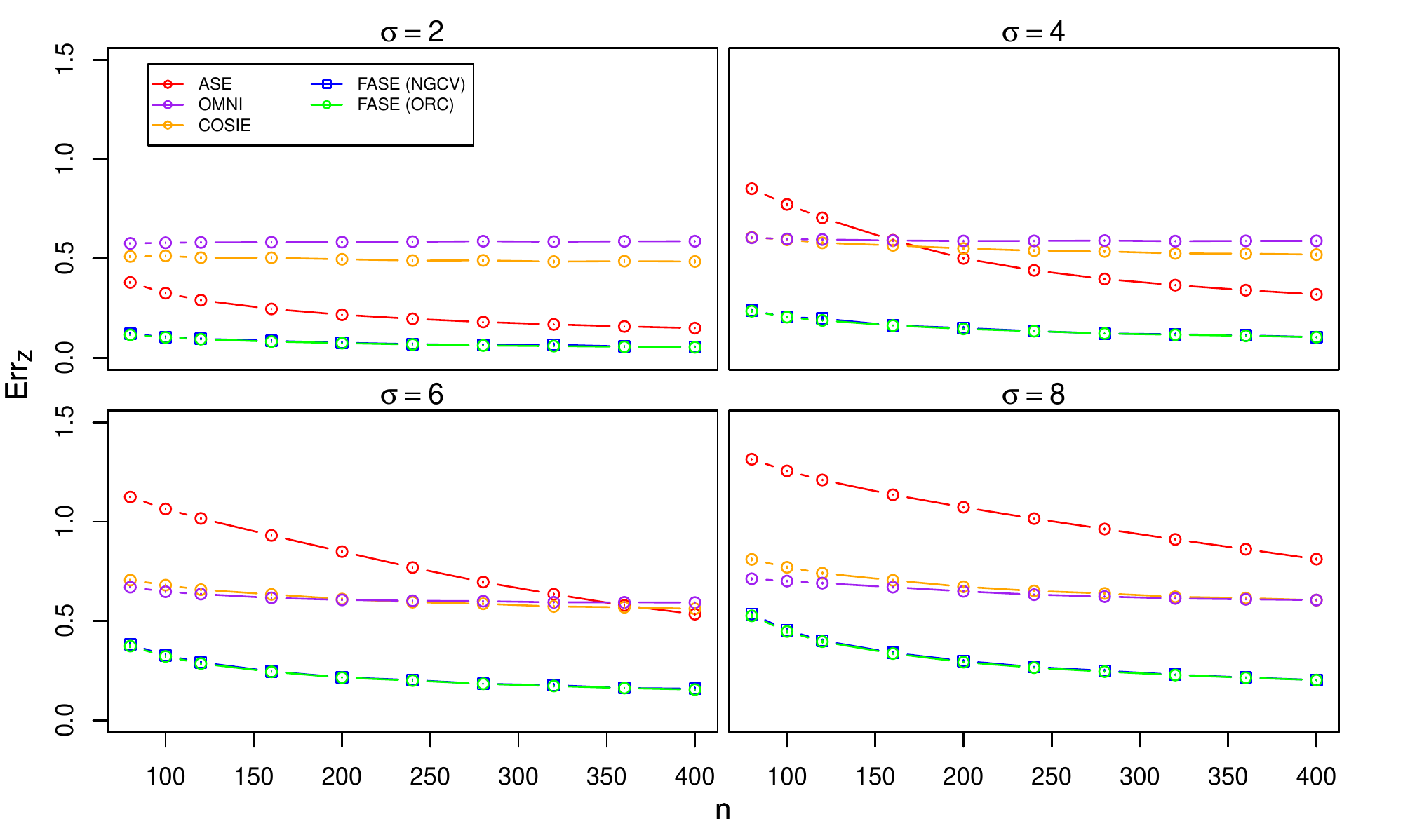}
  \caption{Mean of $\mathrm{Err}_Z$, varying $n$, the number of nodes. Scenario (i), parametric Gaussian networks. Plots are labeled by edge standard deviation $\sigma$.}
  \label{fig:i_varyn}
\end{figure}
\fi

In
\ifshortsim
\else
  Figures~\ref{fig:ii_varym} and \ref{fig:ii_varyn},
\fi
we report results for scenario (ii) generated with
\ifshortsim
\else
  $\sigma \in \{2,4,6,8\}$.
\fi
In
\ifshortsim
\else
  Figure~\ref{fig:ii_varym}
\fi
we vary the number of snapshots $m \in \{20,40,\ldots,200\}$ for fixed $n=100$ and $d=2$, and in
\ifshortsim
\else
  Figure~\ref{fig:ii_varyn}
\fi
we vary the number of nodes $n \in \{80,120,\ldots,400\}$ for fixed $m=80$ and $d=2$.
In all settings, $\mathrm{Err}_Z$ is averaged over 50 independent replications.
Results in these plots look similar to those in scenario (i), confirming that FASE is not relying on any parametric assumptions made on the true latent processes.
In fact, even in the low signal to noise parameter settings for scenario (i) where OMNI outperformed FASE,
FASE now outperforms all of its competitors.
In
\ifshortsim
\else
  Figure~\ref{fig:ii_varyn}, for $\sigma=2$
\fi
the errors for ASE are very close to those for FASE, but this is only because of the very high signal to noise ratio, and even in relative terms, the performance of ASE is comparatively worse as $\sigma$ increases.
\ifshortsim
\else
\begin{figure}
  \centering
  \includegraphics[width=\textwidth]{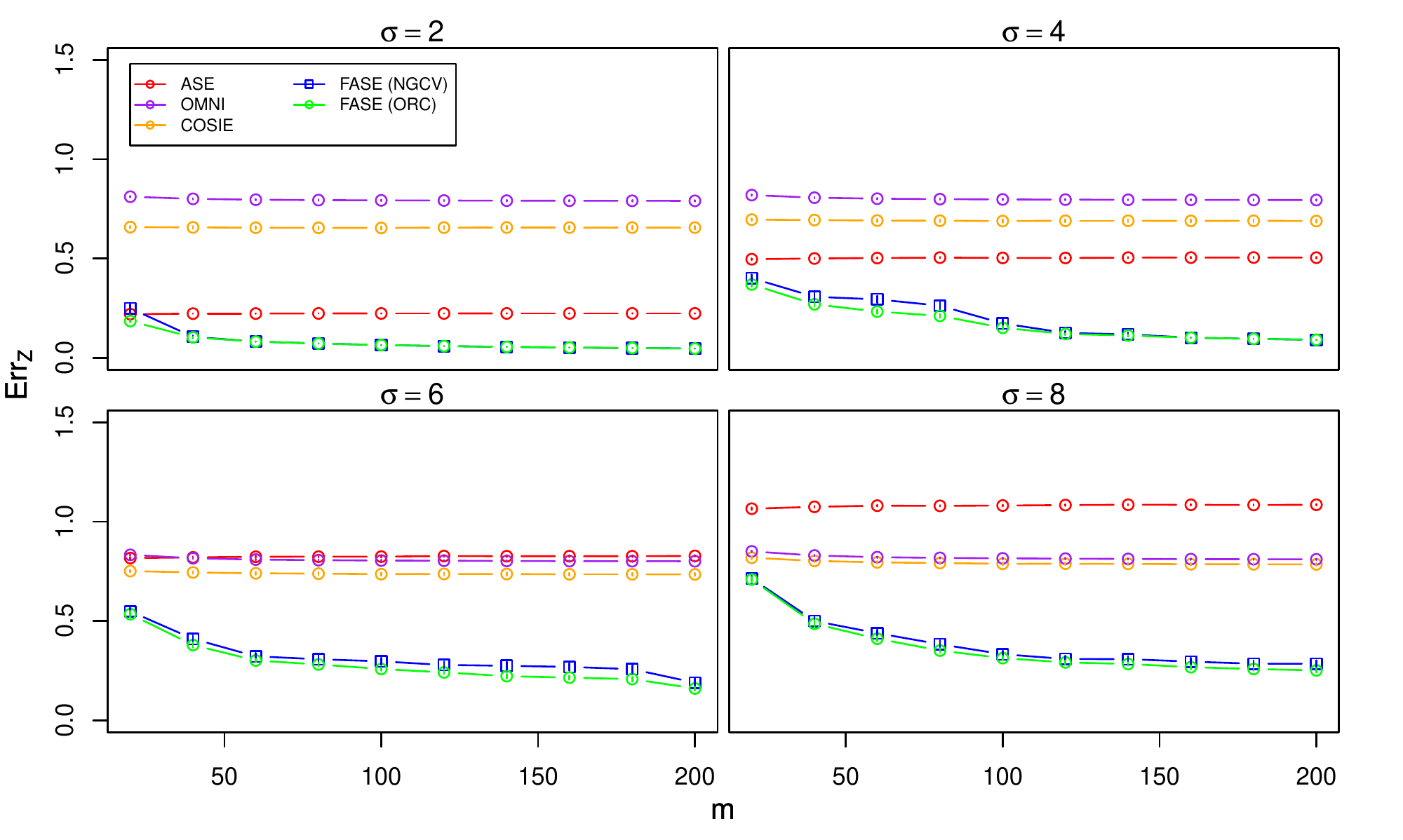}
  \caption{Mean of $\mathrm{Err}_Z$, varying $m$, the number of snapshots. Scenario (ii), nonparametric Gaussian networks. Plots are labeled by edge standard deviation $\sigma$.}
  \label{fig:ii_varym}
\end{figure}
\begin{figure}
  \centering
  \includegraphics[width=\textwidth]{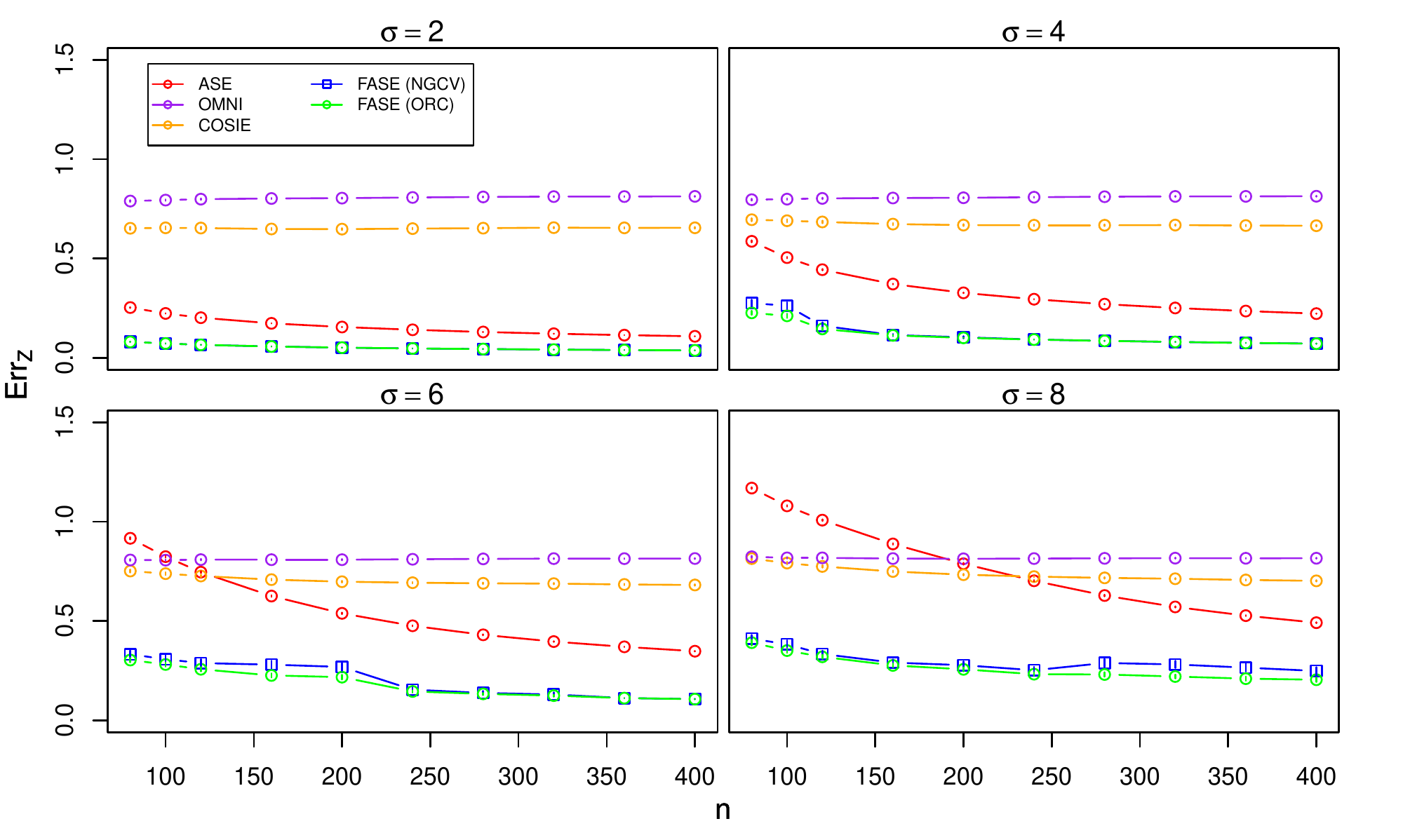}
  \caption{Mean of $\mathrm{Err}_Z$, varying $n$, the number of nodes. Scenario (ii), nonparametric Gaussian networks. Plots are labeled by edge standard deviation $\sigma$.}
  \label{fig:ii_varyn}
\end{figure}
\fi

In
\ifshortsim
\else
  Figures~\ref{fig:iii_varym} and \ref{fig:iii_varyn},
\fi
we report results for scenario (iii) generated with edge densities $0.1$, $0.25$ and $0.5$.
In
\ifshortsim
\else
  Figure~\ref{fig:iii_varym}
\fi
we vary the number of snapshots $m \in \{20,40,\ldots,200\}$ for fixed $n=100$ and $d=2$, and in
\ifshortsim
\else
  Figure~\ref{fig:iii_varyn}
\fi
we vary the number of nodes $n \in \{80,120,\ldots,400\}$ for fixed $m=80$ and $d=2$.
In all settings, $\mathrm{Err}_Z$ is averaged over 50 independent replications.
Once again, none of the baseline ASE estimators show an improvement with increasing $m$, while FASE does, and only ASE and FASE improve with increasing $n$.
In all settings, FASE performs the best of all methods.
For these RDPG networks, we see that while the more conservative COSIE and OMNI approaches improve as the density decreases, the unbiased ASE approach gets substantially worse, as the signal is decreasing.
Similarly, FASE gets slightly worse for decreasing density, but still always outperforms COSIE and OMNI.
\ifshortsim
\else
\begin{figure}
  \centering
  \includegraphics[width=\textwidth]{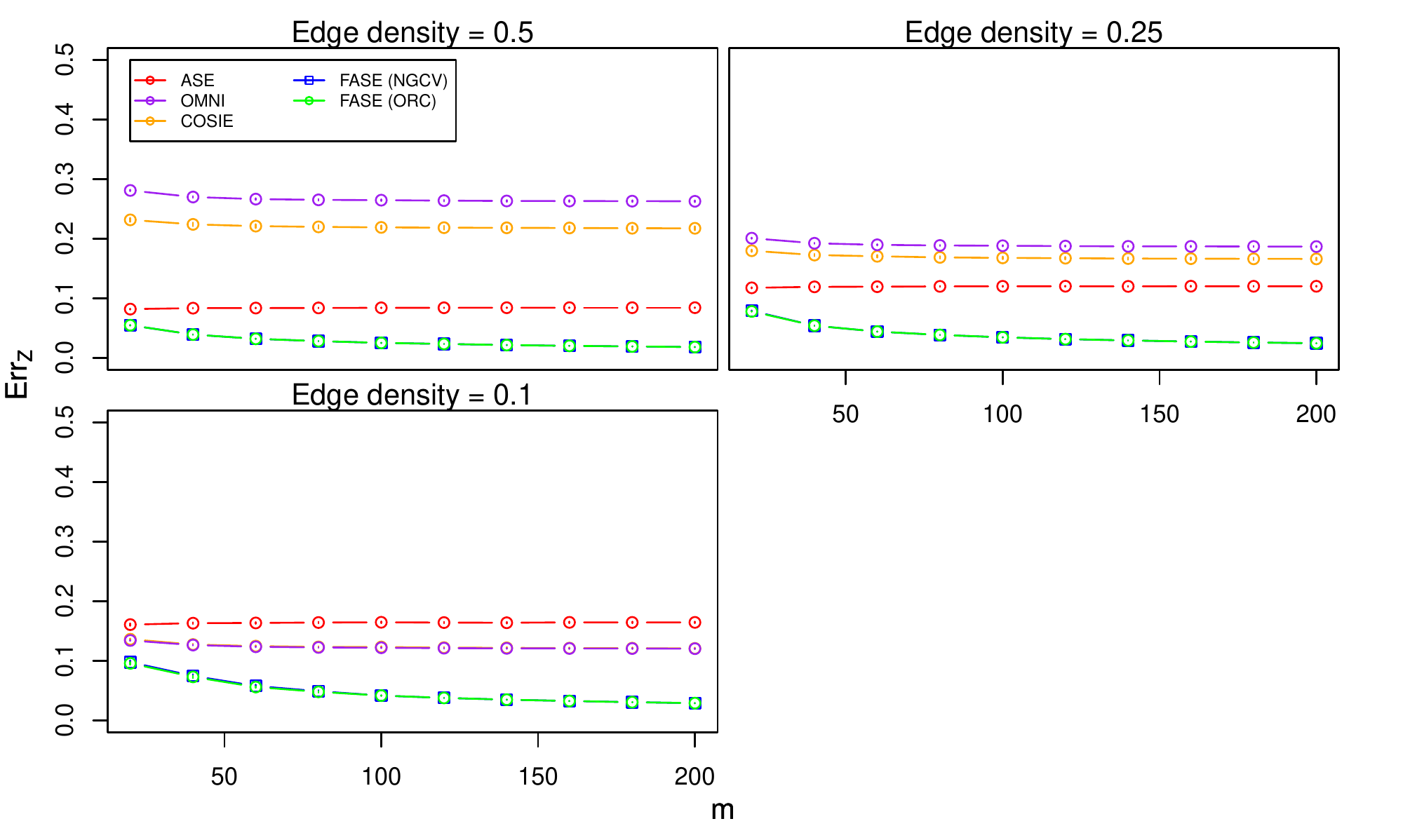}
  \caption{Mean of $\mathrm{Err}_Z$, varying $m$, the number of snapshots. Scenario (iii), parametric RDPG networks. Plots are labeled by edge density.}
  \label{fig:iii_varym}
\end{figure}
\begin{figure}
  \centering
  \includegraphics[width=\textwidth]{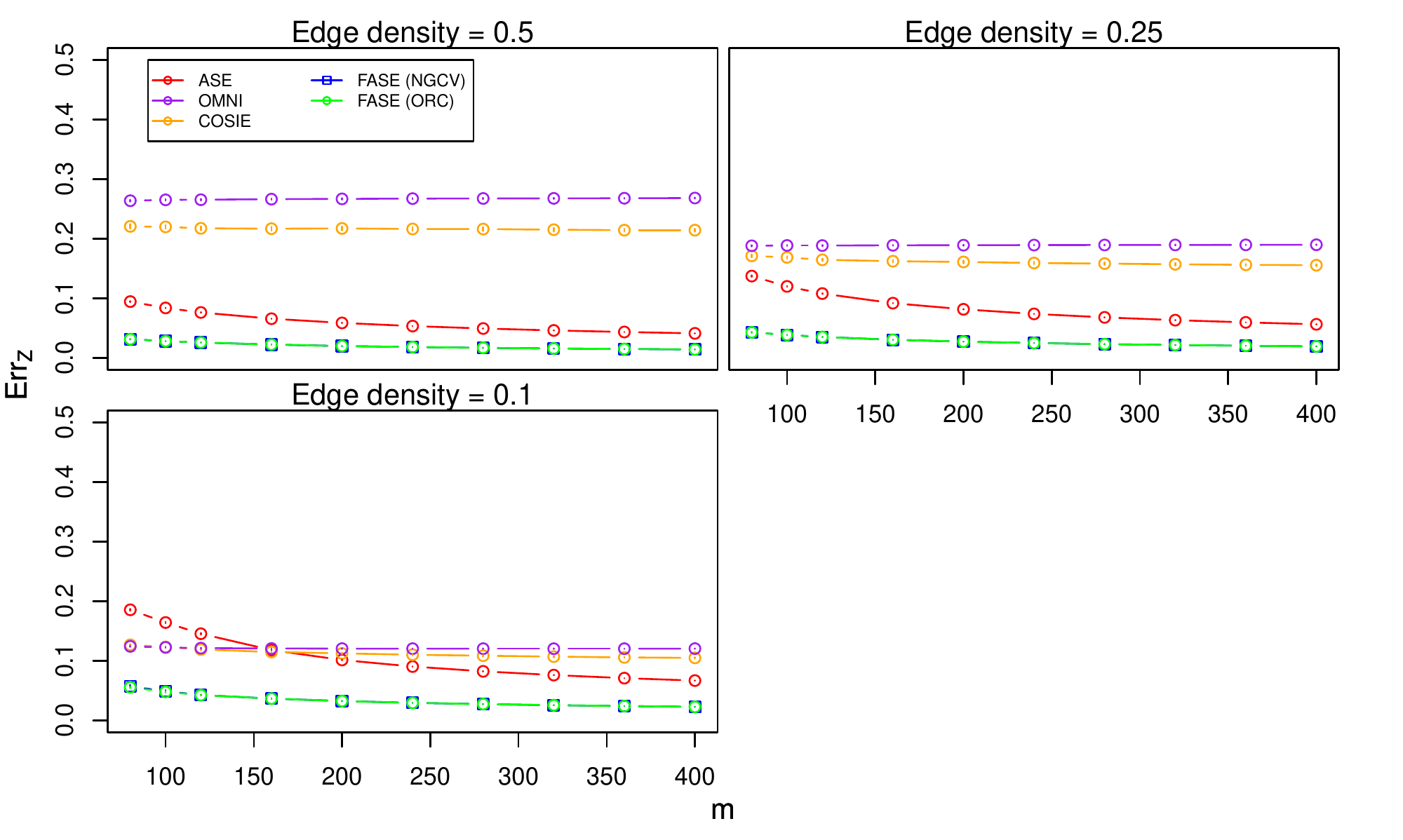}
  \caption{Mean of $\mathrm{Err}_Z$, varying $n$, the number of nodes. Scenario (iii), parametric RDPG networks. Plots are labeled by edge density.}
  \label{fig:iii_varyn}
\end{figure}
\fi

Finally, we report some brief results on convergence and runtime of our FASE estimator.
In these simulations, we set the convergence criterion for Algorithm~\ref{alg:gd_concur} to the relative decrease in the objective function dropping below $10^{-5}$.
In 99\% of replications this occurs in less than 200 iterations, and it always occurs in less than 600 iterations.
As an example benchmark, fitting FASE with $n=m=100$, $d=2$ and $q=10$ to data generated as in scenario (i) takes about 4 seconds on a computer with an Apple M2 Pro Chip and 16GB RAM.

Taken together, we see that among spectral embedding approaches for functional network data, FASE shows state of the art performance for recovery of the underlying latent process structure, up to unknown rotations.
As suggested by Theorem~\ref{thm:main_concur}, in both RDPG and Gaussian settings, we do not see the errors for FASE vanishing to zero, instead they appear to be bounded below by an intrinsic approximation error term.
Even in scenarios (i) and (iii), when we generate the true latent processes from a $B$-spline basis, since we cannot necessarily remove the unknown orthogonal transformation, we essentially revert to a nonparametric regime, and benefit from the fact that FASE only requires smoothness of the latent processes for efficient recovery.

\subsection{Tuning With the NGCV Criterion} \label{subsec:tuning_sims}

We evaluate the performance of our new NGCV model selection criterion on latent process network models generated from scenarios (i), (ii), and (iii) as defined in Section~\ref{subsec:error_sims}.
We compare the quality of selection and the quality of the eventual fitted model compared to an oracle.

To evaluate quality of selection of $d$, we see if our selected model matches the ground truth used to generate the model.
To evaluate the quality of selection of $q$, as there is not always a ground truth parameter, we see if our selected model matches the oracle $q$, denoted by $q_{\mathrm{ORC}}$, and which minimizes the process recovery error up to orthogonal transformation:
\begin{equation*}
  q_{\mathrm{ORC}} = \underset{q_{\min} \leq q' \leq q_{\max}}{\operatorname{argmin}} \mathrm{Err}_Z \left(\widehat{Z}^{(q')}
  \right),
\end{equation*}
where $\widehat{Z}^{(q')}$ is a FASE estimator fit with latent space dimension $d$ and basis dimension $q'$.
In Tables~\ref{tab:tuning_i}-\ref{tab:tuning_iii}, we report the proportion of replications in which the NGCV selection matches the ground truth value for $d$ (d-Prop) and the proportion of replications in which it matches both values in the best pair $(q_{\mathrm{ORC}},d)$ (Prop).
To evaluate the quality of the fitted model, Tables~\ref{tab:tuning_i}-\ref{tab:tuning_iii} also display the ratio between $\mathrm{Err_Z}$ for the FASE estimator fit with $(\hat{q},\hat{d})$ selected according to NGCV, and the FASE estimator fit with the best pair $(q_{\mathrm{ORC}},d)$ (Ratio).

As in Section~\ref{subsec:error_sims}, we consider three scenarios: (i) parametric Gaussian networks with $B$-spline latent processes, (ii) nonparametric Gaussian networks with sinusoidal latent processes, and (iii) parametric RDPG networks with $B$-spline latent processes. In all scenarios, we search for $(q,d)$ pairs over a $6 \times 6$ grid with $q=6,8,\ldots,16$ and $d=1,2,\ldots,6$. We will perform selection either by fitting models over the entire grid, or by coordinate descent (CD), as described in Section~\ref{subsec:tuning}.
Typically, coordinate descent converges in $3$ or $4$ univariate searches, meaning that it fits around $1/2$ to $2/3$ as many models compared to the full grid search over the $6 \times 6$ grid.
The computational improvement of coordinate descent over a full grid search will be more pronounced for larger grids.

For scenario (i), we fix $n=100$, $m=80$, $q=10$, and vary $\sigma=2,4,6,8$ and $d=2,4$. The results, averaged over $50$ replications, are given in Table~\ref{tab:tuning_i}.
\begin{table}
  \centering
    \begin{tabular}{|l|l|l|l|l|l|l|l|}
      \hline
      $d$ & $\sigma$ & $d$-Prop & Prop & Ratio & $d$-Prop & Prop & Ratio \\
      & & (grid) & (grid) & (grid) & (CD) & (CD) & (CD) \\ \hline
      2 & 2 & 0.98 & 0.98 & 1.020 & 0.98 & 0.98 & 1.020 \\ \hline
      2 & 4 & 1.00 & 0.96 & 1.003 & 1.00 & 0.96 & 1.003 \\ \hline
      2 & 6 & 0.96 & 0.58 & 1.015 & 0.94 & 0.58 & 1.021 \\ \hline
      2 & 8 & 1.00 & 0.60 & 1.016 & 0.86 & 0.58 & 1.039 \\ \hline
      4 & 2 & 0.84 & 0.84 & 1.094 & 0.88 & 0.88 & 1.072 \\ \hline
      4 & 4 & 0.92 & 0.92 & 1.024 & 0.96 & 0.92 & 1.013 \\ \hline
      4 & 6 & 0.96 & 0.48 & 1.012 & 0.96 & 0.48 & 1.015 \\ \hline
      4 & 8 & 0.66 & 0.54 & 1.017 & 0.58 & 0.46 & 1.026 \\ \hline
    \end{tabular}
  \caption{Parameter tuning results for scenario (i).}
  \label{tab:tuning_i}
\end{table}
With both grid selection and coordinate descent, even when the selected parameters do not match the oracle, the average error for the selected model is at most 10\% greater than the average oracle error.
Moreover, despite fitting fewer models, the coordinate descent approach almost always agrees with the full grid selection.
In both settings of $d$, selection becomes more challenging for large values of $\sigma$.
For $\sigma=6$, this affects selection of $q$, while for $\sigma=8$ it affects selection of $d$.
As the noise level increases, both the oracle and NGCV tend to select smaller values of $q$.
However, NGCV is more conservative in this respect: its choice of $q$ has already decreased for $\sigma=6$, while the oracle choice does not decrease until $\sigma=8$, which explains why selection of $q$ is better for $\sigma=8$ compared to $\sigma=6$.

For scenario (ii) we fix $n=100$, $m=80$, and vary $\sigma=2,4,6,8$ and $d=2,4$. The results, averaged over $50$ replications, are given in Table~\ref{tab:tuning_ii}.
\begin{table}
  \centering
  \begin{tabular}{|l|l|l|l|l|l|l|l|}
      \hline
      $d$ & $\sigma$ & $d$-Prop & Prop & Ratio & $d$-Prop & Prop & Ratio \\
      & & (grid) & (grid) & (grid) & (CD) & (CD) & (CD) \\ \hline
      2 & 2 & 1.00 & 0.92 & 1.003 & 1.00 & 0.92 & 1.003 \\ \hline
      2 & 4 & 1.00 & 0.82 & 1.005 & 1.00 & 0.82 & 1.005 \\ \hline
      2 & 6 & 0.34 & 0.26 & 1.197 & 0.46 & 0.28 & 1.157 \\ \hline
      2 & 8 & 0.54 & 0.48 & 1.085 & 0.56 & 0.44 & 1.093 \\ \hline
      4 & 2 & 0.98 & 0.94 & 1.024 & 1.00 & 0.94 & 1.002 \\ \hline
      4 & 4 & 0.94 & 0.90 & 1.026 & 0.94 & 0.90 & 1.026 \\ \hline
      4 & 6 & 0.24 & 0.24 & 1.158 & 0.28 & 0.24 & 1.150 \\ \hline
      4 & 8 & 0.54 & 0.48 & 1.057 & 0.64 & 0.52 & 1.043 \\ \hline
    \end{tabular}
  \caption{Parameter tuning results for scenario (ii).}
  \label{tab:tuning_ii}
\end{table}
With both grid selection and coordinate descent, even when the selected parameters do not match the oracle, the average error for the selected model is at most 20\% greater than the average oracle error.
Compared to scenario (i), selection is more difficult in this scenario.
When incorrectly selected, $d$ is typically chosen to be larger than the true value, likely due to the unknown orthogonal rotations.
Wrong selection of $d$ has a large relative effect on the error, especially with lower $\sigma$, as the true $Z$ must be padded with zeros to match the dimensions of the two objects.
However, we can see that especially for grid selection, if we restrict to cases where the ground truth $d$ is selected according to NGCV with grid selection, it is likely that it will also correctly select $q_{\mathrm{ORC}}$.
Conversely, when $d$ is chosen to be larger than the truth, the NGCV criterion tends to compensate by choosing $q$ smaller than $q_{\mathrm{ORC}}$.

For scenario (iii) we fix $n=100$, $m=80$, $q=10$, and vary the edge density in $0.5, 0.25, 0.1$ and $d=2,4$.
Under the Dirichlet-based simulation scheme described in the previous section, we cannot generate RDPG networks with $d=4$ and density $0.5$, so this combination is omitted.
In brief, note that the coordinates for different nodes are generated independently from a Dirichlet distribution on the $d$-dimensional probability simplex, centered at $(1/d \hspace{5pt} \cdots \hspace{5pt} 1/d)^{\tp}$.
For two such variables $X$ and $Y$, $\expect(X^{\tp}Y) = 1/d$.
To reduce the overall network density, we can rescale the positions by a constant $0 < \rho \leq 1$, however $\rho > 1$ will produce many pairs of positions with inner product greater than $1$, outside the parameter space of the Bernoulli edge distribution.
The results, averaged over $50$ replications, are given in Table~\ref{tab:tuning_iii}.
\begin{table}
  \centering
    \begin{tabular}{|l|l|l|l|l|l|l|l|}
      \hline
      $d$ & Density & $d$-Prop & Prop & Ratio & $d$-Prop & Prop & Ratio \\
      & & (grid) & (grid) & (grid) & (CD) & (CD) & (CD) \\ \hline
      2 & 1/2 & 1.00 & 1.00 & 1.000 & 1.00 & 1.00 & 1.000 \\ \hline
      2 & 1/4 & 1.00 & 1.00 & 1.000 & 1.00 & 1.00 & 1.000 \\ \hline
      2 & 1/10 & 0.96 & 0.74 & 1.020 & 0.96 & 0.74 & 1.020 \\ \hline
      4 & 1/4 & 1.00 & 1.00 & 1.000 & 1.00 & 1.00 & 1.000 \\ \hline
      4 & 1/10 & 0.84 & 0.82 & 1.008 & 0.74 & 0.72 & 1.020 \\ \hline
    \end{tabular}
  \caption{Parameter tuning results for scenario (iii).}
  \label{tab:tuning_iii}
\end{table}
In this scenario, both grid selection and coordinate descent give good selection performance, even for smaller edge densities with weaker signal.
Although NGCV typically chooses $q$ smaller than $q_{\mathrm{ORC}}$ for $d=2$ and density $1/10$, we see that this has a small relative effect on the error, as the average selected model error is at most about 2\% greater compared to the average oracle error.

\section{Analysis of International Political Interactions} \label{sec:real_data}

As an application to real functional network data, we apply FASE to data collected by the Integrated Crisis Early Warning System (ICEWS) \citep{lautenschlager15icews}.
In this aggregated data set of international political interactions, we have $m=108$ monthly snapshots of interaction networks on the $n=50$ most active countries from January 2005 to December 2013 in terms of total absolute edge weight.
An undirected edge $[A_k]_{ij}$ describes the total ``weight'' of bilateral interaction between country $i$ and country $j$ in month $k$.
``Weight'' is a signed measure of the intensity and nature of interactions, calculated by the ICEWS.
Weights can be both positive, corresponding to cooperative interactions such as giving aid; or negative, corresponding to hostile interactions such as military action.
To calculate a weight, the ICEWS automatically scrapes and assigns signed weights to news articles, with edge weights calculated by summing all the news articles for a given month.

As the distribution of edge weights is highly skewed, we apply FASE after a log transformation given by
\ifjasa
\else
\begin{equation*}
  \operatorname{sign}([A_k]_{ij}) \log (1 + \lvert [A_k]_{ij} \rvert)
\end{equation*}
\fi
for $k=1,\ldots,m$ and $1 \leq i < j \leq n$.
We use a cubic $B$-spline basis with equally spaced knots, and select $\hat{q}=5$ and $\hat{d}=8$ using NGCV.
The details of this tuning procedure, including a plot of the grid of NGCV criteria are provided in Appendix D of the supplementary materials.

For interpretability of plots, as a post-processing step we perform a Procrustes alignment of each embedded snapshot to the previous snapshot's embedding.
The resulting plotted latent processes are still in the unidentified class $\mathcal{T}(\widehat{Z})$.
In Figures~\ref{fig:embed_d12} and \ref{fig:embed_d34}, we show an exploratory plot of the FASE at four time points for a subset of the latent dimensions.
The remaining dimensions are plotted in
\ifjmlr
\else
\ifjasa
  online Appendix E.
\else
  Appendix~\ref{app:real_data}.
\fi
\fi
Figure~\ref{fig:embed_d12} plots the first latent dimension against the second latent dimension, and Figure~\ref{fig:embed_d34} plots the third latent dimension against the fourth latent dimension.
These plots show the FASE estimates at four distinct time snapshots; a detailed view of the estimated latent processes as they evolve in continuous time can be seen in videos available online at \url{github.com/peterwmacd/fase/tree/main/videos}.

\ifjasa
\else
\begin{figure}
  \centering
  \includegraphics[width=\textwidth]{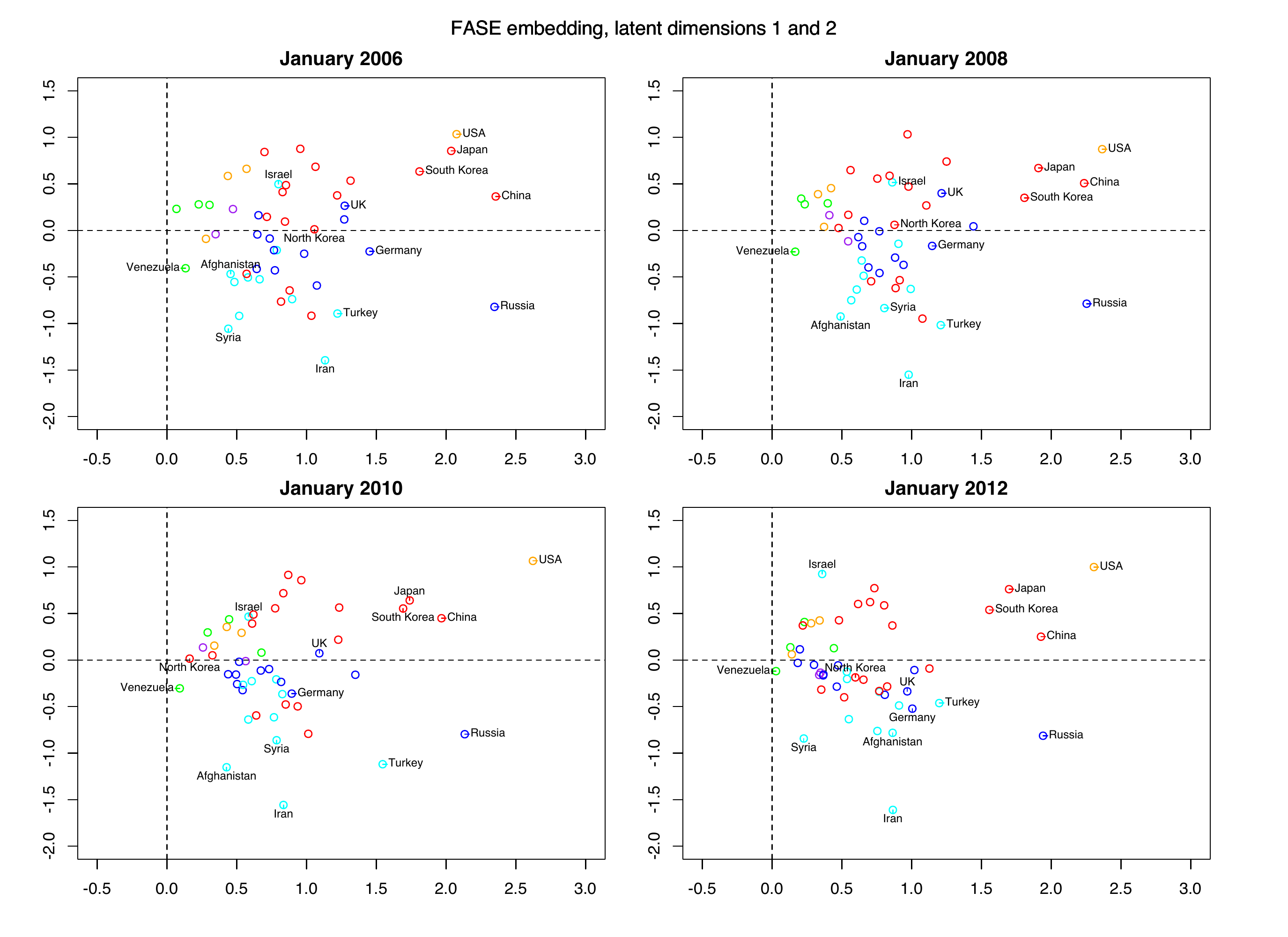}
  \caption{First (horizontal axis) and second (vertical axis) dimensions of FASE evaluated at four times: January 2006, January 2008, January 2010, and January 2012. Points are colored by geographical region. Purple: Africa, Red: Asia-Pacific, Blue: Europe, Cyan: Middle East, Orange: North America, Green: South America.}
  \label{fig:embed_d12}
\end{figure}
\begin{figure}
  \centering
  \includegraphics[width=\textwidth]{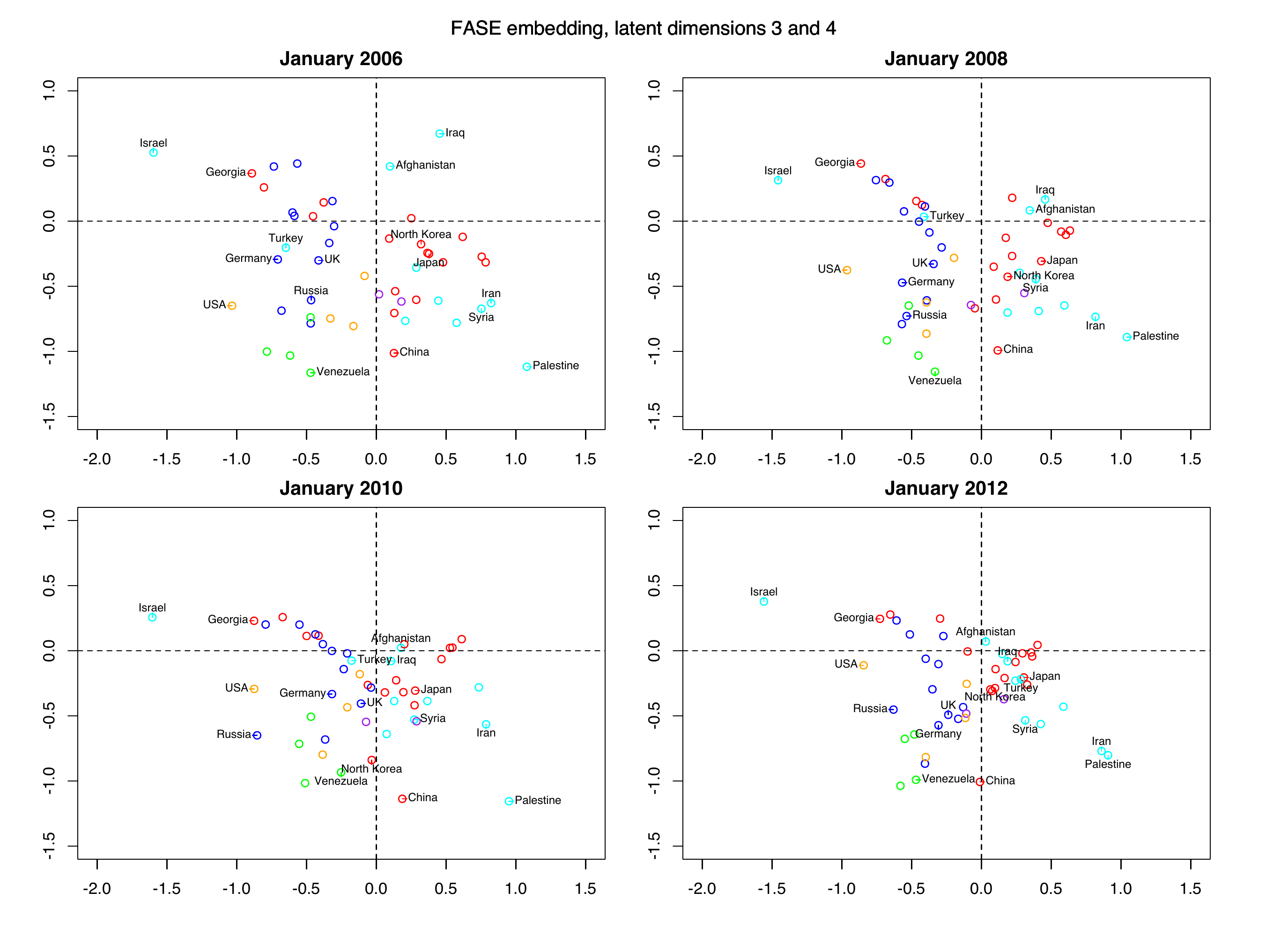}
  \caption{Third (horizontal axis) and fourth (vertical axis) dimensions of FASE evaluated at four times: January 2006, January 2008, January 2010, and January 2012. Points are colored by geographical region. Purple: Africa, Red: Asia-Pacific, Blue: Europe, Cyan: Middle East, Orange: North America, Green: South America.}
  \label{fig:embed_d34}
\end{figure}
\fi
In Figure~\ref{fig:embed_d12}, most countries have positive coordinates in the first latent dimension, corresponding to the total weight and sign of interactions.
Countries like the USA, China and Russia, have large positive values in both dimensions at all four of the plotted times.
Most Asian countries, plotted in red, have positive coordinates in the second latent dimension, while most European countries, plotted in blue, have negative coordinates.
The four Asian nations with large negative coordinates in January 2006, January 2008, and January 2010 are Armenia, Azerbaijan, Georgia and Kazakhstan, four former Soviet republics which are geographically Asian but have more political ties with Europe than with East Asia \citep{engvall17kazakhstan}.
We see some dynamic behavior in these plots as well.
The first latent coordinate for Syria moves substantially between January 2010 and January 2012, possibly a consequence of the Syrian civil war, which began in late 2011
\ifjasa
\else
\ifjmlr
  \citeyearpar[BBC News,][]{11us}.
\else
  \citep{11us}.
\fi
\fi
\ifjasa
\else
There is also consistent movement in the first latent coordinate among countries in the European Union.
In January 2006, their mean first latent coordinate is about $0.92$, while in January 2012 it is $0.53$, reflecting an overall decrease in cooperative relationships during this time period.
\fi

In Figure~\ref{fig:embed_d34}, we again see a regional split between Asian countries with mostly positive third latent coordinates; and European countries with mostly negative coordinates.
\noindent The top left quadrant and the bottom right quadrant separate countries with respect to the Israel-Palestine conflict, which accounts for the largest magnitude negative edges in this network.
We see that this conflict appears to pit Israel against most other Middle Eastern nations, while Europe and the USA tend towards the Israeli side of the conflict.
Again, there are key dynamic shifts in these plots.
In January 2006, the top right and bottom left quadrants appear to separate countries with respect to conflicts between the USA and Iraq, and between the USA and Afghanistan.
However, by January 2012 all three countries' latent coordinates are again much closer to the bulk of the cloud.

To further evaluate the dynamic behavior in this network, for each node we calculate the total distance traversed by its latent process in the $8$-dimensional latent space. In Figure~\ref{fig:distance_bar}, we show the $20$ countries with the greatest distance traversed.
Due to boundary effects around the beginning and end of the time interval, we restrict to distance traversed between January 2006 and December 2012.

\begin{figure}
  \centering
  \ifshortsim
\else
    \includegraphics[width=\textwidth]{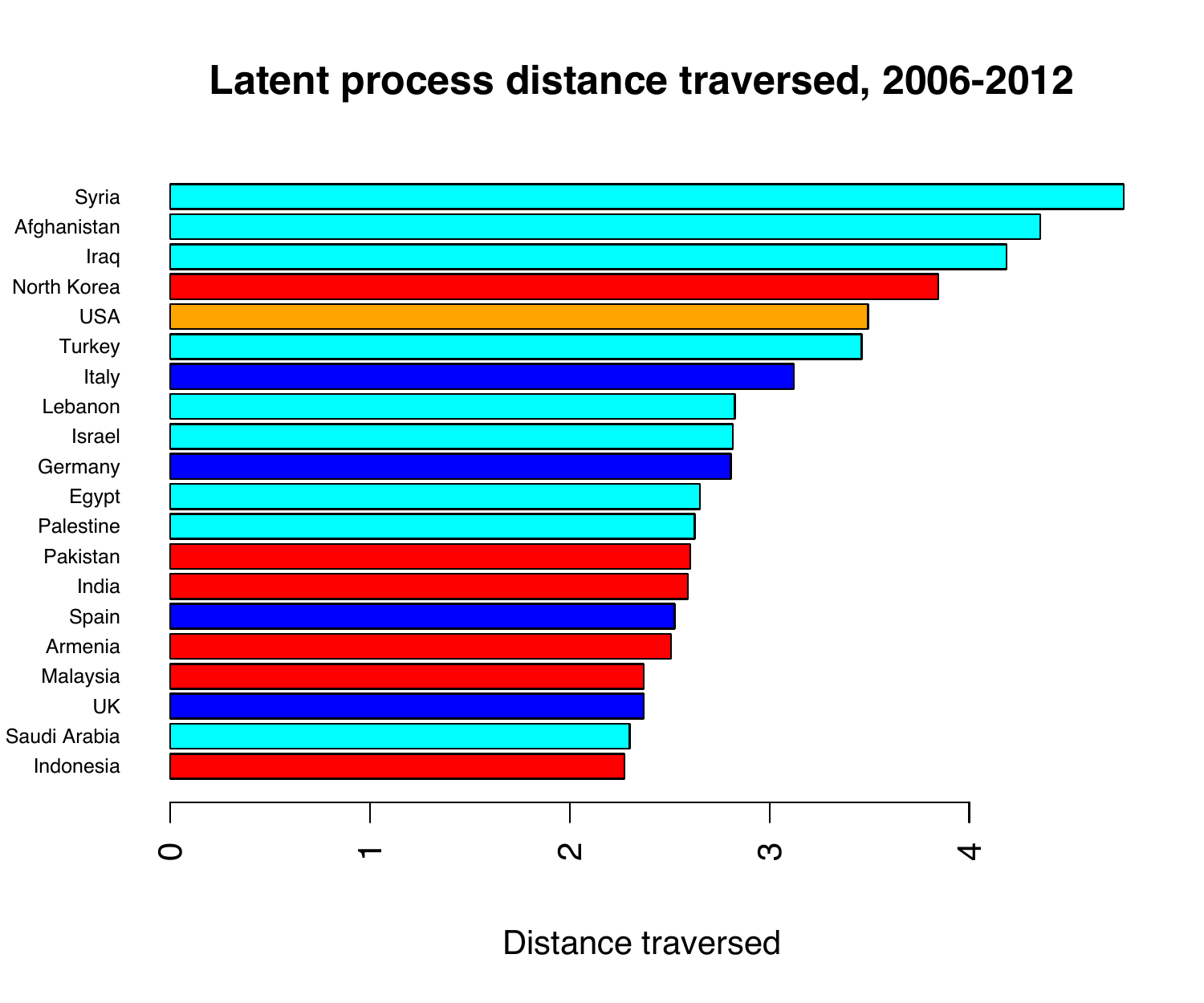}
  \fi
  \caption{Distance traversed by the estimated latent processes, restricted to the 20 countries with the greatest distance traversed. Bars are colored by geographical region. Red: Asia-Pacific, Blue: Europe, Cyan: Middle East, Orange: North America.}
  \label{fig:distance_bar}
\end{figure}

Taking a closer look at countries with the most dynamic behavior, we see that the civil war in Syria appeared to have implications for its relations with many countries, and its latent position changes substantially in the first, as well as the fifth and seventh latent dimensions during this time period.
Afghanistan and Iraq's coordinates both move in the third latent dimension, as well as the seventh, apparently as a result of improving relations with the USA during this time period.
\ifjasa
\else
Afghanistan's position also moves in the sixth latent dimension, possibly in response to border skirmishes with both Iran
\ifjasa
  \citeyearpar[Reuters,][]{08afghan}
\else
\ifjmlr
  \citeyearpar[Reuters,][]{08afghan}
\else
  \citep{08afghan}
\fi
\fi
and Pakistan
\ifjasa
  \citeyearpar[Reuters,][]{13tensions}
\else
\ifjmlr
  \citeyearpar[Reuters,][]{13tensions}
\else
  \citep{13tensions}
\fi
\fi
during this time period.
\fi
The latent position for North Korea moves substantially in the first and third latent dimensions, in both cases reaching a local minimum around January 2010.
As a result, the latent processes for Syria, Afghanistan, Iraq, and North Korea move the most of the countries in the network, despite not being extremely active in terms of total absolute edge weight.
They are the 21st, 15th, 22nd, and 14th most active countries, respectively.
These findings are consistent with the major world events of that time period.

\section{Discussion} \label{sec:conclusion}

\ifjasa
\else
In this paper, we have introduced a new latent process network model for functional network data collected as either adjacency matrix snapshots or aggregated indexed events.
We provide a fitting algorithm using $B$-spline approximation and gradient descent, leading to the FASE estimator.
We give theoretical guarantees and demonstrate the efficacy of our method on simulated and real data with both weighted and binary edges, comparing it to existing ASE-based approaches from the literature.
\fi

Identifiability remains a challenge for latent process network models of this type.
Without strong conditions on eigenvalue separation, even if the orthogonalized latent processes truly belong to $\operatorname{span}(B)$, because of the unknown orthogonal transformation, we cannot take advantage of their parametric form in estimation.
Despite this, we have provided theoretical guarantees up to orthogonal transformation, and demonstrated that for smooth latent processes, sharing of information across network snapshots can still lead to more efficient recovery of the underlying network structure.

Future directions include extending the model to accommodate some dependence, in particular autoregressive edge variables, and theory for more general basis functions, including periodic bases which can be used to model seasonality in dynamic networks.   Another important direction is developing inference for $Z$ or $\mathcal{W}$, with a view towards finding confidence bands for the latent processes, or testing whether a latent dimension is homogeneous across index values.

\ifjasa
\else
\bibliographystyle{plain}
\ifjmlr
\else
\bibliography{dyn_nets0}
\fi
\fi

\ifjasa
\else
\ifjmlr
\else

\appendix

\section{Technical proofs} \label{app:proofs}

\subsection{Proof of Proposition~\ref{prop:snap_gradient}}

\begin{proof}[Proof of Proposition~\ref{prop:snap_gradient}]
  To find the gradient for a given latent dimension $r=1,\ldots,d$, we can rewrite the objective as
  \begin{equation} \label{mle_expanded}
    \sum_{k=1}^m \lVert \left\{ A_k - \sum_{r' \neq r} \bm{W}_{r'} B(x_k) B(x_k)^{\tp} \bm{W}_{r'}^{\tp} \right\} - \bm{W}_{r} B(x_k) B(x_k)^{\tp} \bm{W}_{r}^{\tp} \rVert_F^2.
  \end{equation}
  where the matrix in braces is free of $\bm{W}_{r}$. Thus it is sufficient to analyze \eqref{mle_expanded} in the special case $d=1$, where the objective can be written as
  \begin{align*}
      & \min_{\bm{W}} \left\{ \sum_{k=1}^m \lVert A_k - \bm{W} B(x_k) B(x_k)^{\tp} \bm{W}^{\tp} \rVert_F^2 \right\} \\
    =& \sum_{k=1}^m \operatorname{tr}\left( \left[ A_k - \bm{W} B(x_k)B(x_k)^{\tp} \bm{W}^{\tp} \right]^2 \right) \\
    =& \sum_{k=1}^m \operatorname{tr}\big( A_k^2 - A_k \bm{W} B(x_k)B(x_k)^{\tp} \bm{W}^{\tp} - \bm{W} B(x_k)B(x_k)^{\tp} \bm{W}^{\tp} A_k \\
    &\tabby + \bm{W} B(x_k)B(x_k)^{\tp} \bm{W}^{\tp}\bm{W} B(x_k)B(x_k)^{\tp} \bm{W}^{\tp} \big) \\
    \propto& \sum_{k=1}^m \bigg\{ -\operatorname{tr}\left( A_k \bm{W} B(x_k)B(x_k)^{\tp} \bm{W}^{\tp} \right) - \operatorname{tr} \left( \bm{W} B(x_k)B(x_k)^{\tp} \bm{W}^{\tp} A_k \right) \\
    &\tabby + \operatorname{tr} \left( \bm{W} B(x_k)B(x_k)^{\tp} \bm{W}^{\tp}\bm{W} B(x_k)B(x_k)^{\tp} \bm{W}^{\tp} \right) \bigg\},
  \end{align*}
  where in the final expression we drop the term not depending on $\bm{W}$. Now take a derivative of each term with respect to $\bm{W}$. First,
  \begin{equation*}
    \frac{\partial}{\partial \bm{W}}\operatorname{tr}\left( A_k \bm{W} B(x_k)B(x_k)^{\tp} \bm{W}^{\tp} \right) = 2 A_k \bm{W} B(x_k)B(x_k)^{\tp}
  \end{equation*}
  and the other cross term is the same.
  Then,
  \begin{equation*}
    \frac{\partial}{\partial \bm{W}} \operatorname{tr} \left( \bm{W}B(x_k)B(x_k)^{\tp} \bm{W}^{\tp}\bm{W} B(x_k)B(x_k)^{\tp} \bm{W}^{\tp} \right) = 4 \bm{W} B(x_k)B(x_k)^{\tp} \bm{W}^{\tp}\bm{W} B(x_k)B(x_k)^{\tp}.
  \end{equation*}
  The entire gradient with respect to $\bm{W}$ is
  \begin{equation*}
    - 4 \sum_{k=1}^m \left( A_k - \bm{W} B(x_k)B(x_k)^{\tp} \bm{W}^{\tp} \right) \bm{W} B(x_k)B(x_k)^{\tp}.
  \end{equation*}
  Thus for the general case with $d > 1$ we see that the gradient with respect to $\bm{W}_r$ is the desired
  \begin{equation*}
    - 4 \sum_{k=1}^m \left\{ A_k - \sum_{r'=1}^d \bm{W}_{r'} B(x_k)B(x_k)^{\tp} \bm{W}_{r'}^{\tp} \right\} \bm{W}_r B(x_k)B(x_k)^{\tp}.
  \end{equation*}
\end{proof}

\subsection{Preliminaries for proofs of Theorems}

In this section we will introduce notation as well as some preliminary results which we will use in the proof of Theorems~\ref{thm:main_concur}.

We continue to use matrix and tensor notation introduced in Section~\ref{sec:model}, as well as the matrix nuclear norm denoted by $\lVert \cdot \rVert_*$, and the Frobenius inner product denoted by $\iprod{\cdot}{\cdot}$.
The Frobenius inner product is given by
\begin{equation*}
  \iprod{M}{R} = \operatorname{tr}(M^{\tp}R)
\end{equation*}
and satisfies $\iprod{M}{M} = \lVert M \rVert_F^2$ for matrices $M$ and $R$ of the same dimensions.
Recall that with some abuse of notation, we will use $\lVert \cdot \rVert_F$ and $\iprod{\cdot}{\cdot}$ to denote the vector $\ell_2$ norm and Euclidean inner product of the vectorization of a 3-mode tensor.

Auxiliary results will typically be referenced below mathematical displays in which they are used.
We also use some well known matrix algebra results, including the submultiplicative property of matrix norms, and Cauchy-Schwarz inequality for both the Euclidean and Frobenius inner products.
The matrix norms we consider satisfy
\begin{equation*}
  \frac{1}{\operatorname{rank}(M)} \lVert M \rVert_* \leq \frac{1}{\operatorname{rank}^{1/2}(M)} \lVert M \rVert_F \leq \lVert M \rVert_2 \leq \lVert M \rVert_F \leq \lVert M \rVert_*
\end{equation*}
for any matrix $M$.
We also use a special combined submultiplicative property
\begin{equation*}
  \lVert MR \rVert_F \leq \lVert M \rVert_2 \lVert R \rVert_F
\end{equation*}
for matrices $M$ and $R$ of suitable dimensions.
Next we prove three basic matrix algebra lemmas, and one probability lemma, both of which will be used in the proofs to follow.

\begin{lemma} \label{lemma:unfold}
  Suppose $M$ is an $n \times n$ symmetric matrix, and $X, Y$ are $n \times d$ matrices. Then
  \begin{equation*}
    \iprod{MX}{X-Y} = \frac{1}{2}\iprod{M}{XX^{\tp}-YY^{\tp}} + \frac{1}{2}\iprod{M}{(X-Y)(X-Y)^{\tp}} .
  \end{equation*}
\end{lemma}

\begin{proof}[Proof of Lemma~\ref{lemma:unfold}]
  \begin{align*}
    \iprod{MX}{X-Y} &= \iprod{M}{XX^{\tp}-YX^{\tp}} \\
    &= \iprod{M}{\frac{1}{2}(XX^{\tp} - YY^{\tp}) + \frac{1}{2}(XX^{\tp} + YY^{\tp}) - YX^{\tp}} \\
    &= \frac{1}{2}\iprod{M}{XX^{\tp}-YY^{\tp}} + \frac{1}{2}\iprod{M}{XX^{\tp} + YY^{\tp} - YX^{\tp} - XY^{\tp}} \\
    &= \frac{1}{2}\iprod{M}{XX^{\tp}-YY^{\tp}} + \frac{1}{2}\iprod{M}{(X-Y)(X-Y)^{\tp}} ,
  \end{align*}
  where the second to last equality uses the symmetry of $M$.
\end{proof}

\begin{lemma} \label{lemma:b_block}
  Suppose an $n \times q$ matrix $W$ satisfies $\lVert W B(x_k) \rVert_2 \leq \gamma$ for all $k=1,\ldots,m$, where $B(x)$ and $\bm{B}$ are defined as in Section~\ref{sec:theory}. Then under Assumption~\ref{assump:bs}, the $q \times nm$ block matrix
  \begin{equation*}
    M = \begin{pmatrix}
      B(x_1)B(x_1)^{\tp}W^{\tp} & \cdots & B(x_m)B(x_m)^{\tp}W^{\tp}
    \end{pmatrix}
  \end{equation*}
  satisfies $\lVert M \rVert_2 \leq \gamma ( C_B m / q )^{1/2}$.
\end{lemma}

\begin{proof}[Proof of Lemma~\ref{lemma:b_block}]
  Rewrite
  \begin{equation*}
    M = \bm{B}^{\tp} \begin{pmatrix}
      B(x_1)^{\tp}W^{\tp} & 0 & \\
      0 & \ddots & 0 \\
      & 0 & B(x_m)^{\tp}W^{\tp}
  \end{pmatrix},
  \end{equation*}
  where the second factor is an $m \times nm$ block diagonal matrix. Then by Assumption~\ref{assump:bs}, part (B), $\lVert \bm{B} \rVert_2 \leq (C_B m / q)^{1/2}$, and $\lVert WB(x_k) \rVert_2 \leq \gamma$ by assumption.
\end{proof}

\begin{lemma} \label{lemma:proc_alignment}
  Define vectors $x,y,z \in \mathbb{R}^n$. Suppose (without loss of generality) $x^{\tp}z > 0$.
  If
  \begin{equation*}
    \min_{s \in \{-1,1\}}\lVert sx - y \rVert_2 + \lVert y - z \rVert_2 < \lVert x \rVert_2,
  \end{equation*}
  then $x^{\tp}y > 0$.
\end{lemma}

\begin{proof}
  Define $s^* = \operatorname{argmin}_{s \in \{-1,1\}}\lVert sx - y \rVert_2$.

  Write $c_1 = \lVert s^*x - y \rVert_2 / \lVert x \rVert_2$ and $c_2 = \lVert y - z \rVert_2 / \lVert x \rVert_2$. $c_1,c_2 \in (0,1)$ and satisfy $c_1 + c_2 < 1$.
  Then
  \begin{equation*}
    \lvert x^{\tp} y \rvert = \lvert x^{\tp}(y - s^*x + s^*x) \rvert \geq (1 - c_1) \lVert x \rVert_2^2.
  \end{equation*}
  Suppose, with the goal of obtaining a contradiction, that $x^{\tp}y \leq 0$. Since $x^{\tp}z > 0$ this implies
  \begin{equation*}
    x^{\tp}z - x^{\tp}y > (1 - c_1)\lVert x \rVert_2^2.
  \end{equation*}
  On the other hand, by assumption
  \begin{equation*}
    x^{\tp}z - x^{\tp}y \leq c_2 \lVert x \rVert_2^2 < (1-c_1)\lVert x \rVert_2^2,
  \end{equation*}
 which is a contradiction.
\end{proof}

\begin{lemma} \label{lemma:probability}
  Suppose $\{A_k\}_{k=1}^m$ are generated from a latent process network model, with independent sub-Gaussian edges with parameter at most $\sigma$, $B(x)$ satisfies Assumption~\ref{assump:bs}, and $n$, $q$ are such that $nq\log q \geq nq \log 5 + (n+q) \log 9$.
  Define the set
  \begin{equation*}
    \mathcal{B} = \{ W : \lVert W B(x_k) \rVert_2 \leq \gamma \tabby \forall k=1,\ldots,m \} \subseteq \mathbb{R}^{n \times q}
  \end{equation*}
  for some $\gamma > 0$.
  Then there is a constant $c_{\mathrm{prob}}$ such that the event
  \begin{equation*}
    \underset{W \in \mathcal{B}}{\bigcap} \left\{ \left\lVert \frac{1}{m} \sum_{k=1}^m \left\{ A_k - \sum_{r=1}^d Z_r(x_k)Z_r(x_k)^{\tp} \right\} W B(x_k)B(x_k)^{\tp} \right\rVert_2
    \leq c_{\mathrm{prob}} \gamma \left( \frac{\sigma^2 q^2 n \log q}{m} \right)^{1/2} \right\}
  \end{equation*}
  denoted by $\mathcal{E}$, satisfies $\prob( \mathcal{E} ) \geq 1 - 2\exp(-n/2)$.
\end{lemma}

\begin{proof}[Proof of Lemma~\ref{lemma:probability}]
  We will prove a high probability bound for
  \begin{equation*}
    \sup_{W \in \mathcal{B}} \lVert M(W) \rVert_2,
  \end{equation*}
  where $M(W) = \sum_{k=1}^m \{ A_k - \sum_{r=1}^d Z_r(x_k)Z_r(x_k)^{\tp} \} W B(x_k)B(x_k)^{\tp}$. Define
  \begin{equation*}
    \mathcal{B}^+ = \{ W : \lVert W \rVert_F \leq \frac{\sqrt{C_B q}}{c_B} \gamma \} 
  \end{equation*}
  Note that $\mathcal{B}^+$ is equivalent to a closed Euclidean ball in $\mathbb{R}^{nq}$ of radius $\sqrt{C_B q} \gamma / c_B$.
  Moreover, for any $W \in \mathcal{B}$,
  \begin{align*}
    \lVert W \rVert_F^2 &= \lVert W (\bm{B}^{\tp}\bm{B}) (\bm{B}^{\tp}\bm{B})^{-1} \rVert_F^2 \\
    &\leq \lVert W\bm{B}^{\tp} \rVert_F^2 \lVert \bm{B} (\bm{B}^{\tp}\bm{B})^{-1} \rVert_2^2 \\
    &\leq \left( \sum_{k=1}^m \lVert W B(x_k) \rVert_2^2 \right) \left( \frac{C_B m}{q} \right) \left( \frac{c_B m}{q} \right)^{-2} \\
    &\leq \frac{C_B}{c_B^2} \gamma^2 q. \\
  \end{align*}
  so that $\mathcal{B} \subseteq \mathcal{B}^+$.

  By standard covering results \citep[Proposition 4.2.12]{vershynin18highdimensional}, we can find a $(C_B^{1/2}\gamma / 2c_B)$-net for $\mathcal{B}^+$ (under Frobenius metric), denoted by $\mathcal{L}$,  satisfying
  \begin{equation*}
    \left\lvert \mathcal{L} \right\rvert \leq (4 q^{1/2} + 1)^{nq} \leq \left( 5 q^{1/2} \right)^{nq}.
  \end{equation*}
  Every element $W$ of $\mathcal{B}$ can be written as $W' + E$, where $W' \in \mathcal{L}$, and $\lVert E \rVert_2 \leq C_B^{1/2}\gamma / 2c_B$. Fix $W \in \mathcal{B}$. Then
  \begin{align*}
    \lVert M(W) \rVert_2 &\leq \lVert M(W') \rVert_2 + \lVert M(E) \rVert_2 \\
    &\leq \max_{W' \in \mathcal{L}} \lVert M(W') \rVert_2 + \frac{1}{2} \sup_{W \in \mathcal{B}} \lVert M(W) \rVert_2. \\
    &\leq \max_{W' \in \mathcal{L}} \left\{ 2 \max_{x \in \mathcal{N}, y \in \mathcal{M}} x^{\tp} M(W') y \right\} + \frac{1}{2} \sup_{W \in \mathcal{B}} \lVert M(W) \rVert_2 ,
  \end{align*}
  where $\mathcal{N}$ and $\mathcal{M}$ are $(1/4)$-nets for $\mathcal{S}^{n-1}$ and $\mathcal{S}^{q-1}$ of cardinalities $9^n$ and $9^q$, respectively \citep[Theorem 4.4.5]{vershynin18highdimensional}. Taking a supremum on the left hand side and rearranging,
  \begin{equation} \label{net_approx}
    \sup_{W \in \mathcal{B}} \lVert M(W) \rVert_2 \leq 4 \max_{W' \in \mathcal{L}, x \in \mathcal{N}, y \in \mathcal{M}} x^{\tp} M(W')y.
  \end{equation}

  For fixed $x$, $y$, and $W'$, concentration follows directly from Theorem 4.4.5 in \citet{vershynin18highdimensional}, and the fact that $\lVert B(x_k) \rVert_2 \leq 1$ for all $k$, resulting in the sub-Gaussian tail bound
  \begin{equation*}
    \prob \left\{ x^{\tp} M(W') y \geq t \right\} \leq 2 \exp \left( \frac{-c'_{\mathrm{prob}}t^2}{\sigma^2 m q \gamma^2} \right)
  \end{equation*}
  for a constant $c'_{\mathrm{prob}}$ (which depends on $C_B$ and $c_B$). We now take a union bound over the elements in the net, obtaining
  \begin{equation*}
    \prob \left[ \max_{W' \in \mathcal{L},x \in \mathcal{N},y \in \mathcal{M}} \left\{ x^{\tp} M(W') y \right\} \geq t \right] \leq 2 \cdot \left( 5 q^{1/2} \right)^{nq} 9^{n + q} \exp \left( \frac{-c'_{\mathrm{prob}} t^2}{\sigma^2 m q \gamma^2} \right).
  \end{equation*}
  By assumption, $n$ and $q$ satisfy $nq \log q \geq nq \log 5 + (n + q) \log 9$. Set
  \begin{equation*}
    c_{\mathrm{prob}} = 4 \left( \frac{2}{c'_{\mathrm{prob}}} \right)^{1/2},
  \end{equation*}
  and $t^* = c_{\mathrm{prob}} \sigma \gamma (nm q^2 \log q)^{1/2}$.
  Then
  \begin{equation*}
        \prob \left[ \max_{W' \in \mathcal{L},x \in \mathcal{N},y \in \mathcal{M}} \left\{ x^{\tp} M(W') y \right\} \geq \frac{t^*}{4} \right] \leq 2 \exp \left( -\frac{n}{2} \right).
  \end{equation*}
 In combination with \eqref{net_approx}, this completes the proof.
\end{proof}

  \subsection{Proof of Theorem~\ref{thm:main_concur}} \label{subsec:pf_thm2}

  In this section, we prove Theorem~\ref{thm:main_concur}.
  Recall $c_B$, $C_B$, $\gamma_Z$, $\kappa$ defined in the main body of the paper, and define
  \begin{align*}
    c_{\mathrm{SNR}} &= \frac{\gamma_Z^2}{\sigma} \left( \frac{m}{q^5 n \log q}\right)^{1/2}  , \\
    c_{\mathrm{init}} &= \frac{\lVert \widehat{\mathcal{W}}^0 - \mathcal{W}^{*,0} \rVert_F^2}{\gamma_Z^2}  , \\
    c_{\mathrm{approx},2} &= \frac{q}{\gamma_Z^2} \operatorname{sup}_{h \geq 0} \neweps_{\mathrm{approx},2}^{(h)} \\
    c_{\mathrm{approx},\infty} &= \frac{1}{\gamma_Z^2} \operatorname{sup}_{h \geq 0} \neweps_{\mathrm{approx},\infty}^{(h)}.
  \end{align*}

  We start from two necessary lemmas.

  \begin{lemma} \label{lemma:op_norm2}
    Suppose the assumptions of Theorem~\ref{thm:main_concur} hold. Fix $h \geq 0$ and let
    \begin{equation*}
      c_{\mathrm{prev}} = \frac{1}{\gamma_Z^2} \lVert \widehat{\mathcal{W}}^h - \mathcal{W}^{*,h} \rVert^2_F.
    \end{equation*}
    Then
    \begin{equation*}
      \lVert \widehat{\bm{W}}_r^h B(x_k) \rVert_2 \leq c_W \gamma_Z
    \end{equation*}
    uniformly over $r=1,\ldots,d$ and $k=1,\ldots,m$ for a constant $c_W = (c_{\mathrm{prev}}^{1/2} + c_{\mathrm{approx},\infty}^{1/2} + \kappa^{1/2})$.
  \end{lemma}


  \begin{proof}[Proof of Lemma~\ref{lemma:op_norm2}]
    \begin{align*}
      \lVert \widehat{\bm{W}}_r^h B(x_k) \rVert_2 \rVert_2 &\leq \lVert (\widehat{\bm{W}}_r^h - \bm{W}_r^{*,h} ) B(x_k) \rVert_2 + \lVert \bm{W}_r^{*,h} B(x_k) - Z(x_k)Q_k^{*,h}\rVert_2 + \lVert Z(x_k)Q_k^{*,h} \rVert_2 \\
      &\leq c_{\mathrm{prev}}^{1/2} \gamma_Z + c_{\mathrm{approx},\infty}^{1/2} \gamma_Z + \kappa^{1/2} \gamma_Z.
    \end{align*}
  \end{proof}

  \begin{lemma} \label{lemma:onestep2}
    Suppose the assumptions of Theorem~\ref{thm:main_concur} hold, and that $\mathcal{E}$ occurs.
    Fix $h \geq 0$, define $c_{\mathrm{prev}}$ as in Lemma~\ref{lemma:op_norm2}, and suppose
    \begin{equation} \label{onestep2_prev}
      c_{\mathrm{prev}} \leq \frac{c_B}{16 C_B} .
    \end{equation}
    Then for positive constants
    \begin{equation} \label{onestep2_constants}
      \rho = c_B/8, \tabby c_{\mathrm{step}} = \max \left\{ c_{\mathrm{prob}}^2 c_W \left( \frac{2d}{c_B} + \frac{1}{16} \right), 4c_{\mathrm{approx},2} + 6\kappa \right\},
    \end{equation}
    we have
    \begin{equation*}
      \lVert \widehat{\mathcal{W}}^{h+1} - \mathcal{W}^{*,h+1} \rVert_F^2 \leq \left( 1 - \eta' \rho \right) \lVert \widehat{\mathcal{W}}^{h} - \mathcal{W}^{*,h} \rVert_F^2 + c_{\mathrm{step}} \eta' \left( \frac{\sigma^2 q^5 n \log q}{\gamma_Z^2 m} + q \neweps_{\mathrm{approx},2}^{(h)} \right).
    \end{equation*}
  \end{lemma}

  \begin{proof}[Proof of Lemma~\ref{lemma:onestep2}]
    We define the following terms which we will see later in the proof:
    \begin{align*}
      &T_{\mathrm{mean}} = \sum_{k=1}^m \left\lVert \left\{ \sum_{r=1}^d \widehat{\bm{W}}_r^{h}B(x_k)B(x_k)^{\tp}(\widehat{\bm{W}}_r^h)^{\tp} \right\} - Z(x_k)Z(x_k)^{\tp} \right\rVert_F^2, \\
      &T_{\mathrm{cross}} = \bigg\vert \sum_{r=1}^d \sum_{k=1}^m \bigg\langle \left\{ \sum_{r'=1}^d \widehat{\bm{W}}_{r'}^{h}B(x_k)B(x_k)^{\tp} (\widehat{\bm{W}}_{r'}^h)^{\tp} \right\} - Z(x_k)Z(x_k)^{\tp}, \\
      &\hspace{1in} (\widehat{\bm{W}}_r^h - \bm{W}^{*,h}_r) B(x_k)B(x_k)^{\tp} (\widehat{\bm{W}}_r^h - \bm{W}^{*,h}_r)^{\tp} \bigg\rangle \bigg\vert, \\
      &T_{\mathrm{op}} = 2 \left\lvert \sum_{r=1}^d \iprod{\sum_{k=1}^m \left\{ A_k - Z(x_k)Z(x_k)^{\tp} \right\} \widehat{\bm{W}}_r^h B(x_k)B(x_k)^{\tp} }{\widehat{\bm{W}}_r^h - \bm{W}^{*,h}_r} \right\rvert, \\
      &T_{\mathrm{approx}} = 2 \bigg\vert \sum_{k=1}^m \bigg\langle \left\{ \sum_{r=1}^d \widehat{\bm{W}}_r^{h}B(x_k)B(x_k)^{\tp} (\widehat{\bm{W}}_r^h)^{\tp} \right\} - Z(x_k)Z(x_k)^{\tp}, \\
      &\hspace{1in} \left\{ \sum_{r'=1}^d \bm{W}_{r'}^{*,h} B(x_k)B(x_k)^{\tp}(\bm{W}_{r'}^{*,h})^{\tp} \right\} - Z(x_k)Z(x_k)^{\tp} \bigg\rangle \bigg\vert, \\
      &T_{\mathrm{quad.mean}}^2 = 2 \sum_{r=1}^d \left\lVert \sum_{k=1}^m \left[ \left\{ \sum_{r'=1}^d  \widehat{\bm{W}}_{r'}^{h}B(x_k)B(x_k)^{\tp}(\widehat{\bm{W}}_{r'}^h)^{\tp} \right\} - Z(x_k)Z(x_k)^{\tp} \right] \widehat{\bm{W}}_r^h B(x_k)B(x_k)^{\tp} \right\rVert_F^2, \\
      &T_{\mathrm{quad.op}}^2 = 2 \sum_{r=1}^d \left\lVert \sum_{k=1}^m \left\{ A_k - Z(x_k)Z(x_k)^{\tp} \right\} \widehat{\bm{W}}_r^h B(x_k)B(x_k)^{\tp} \right\rVert_F^2.
    \end{align*}

    Then we have
    \begin{align} \label{summary_terms2}
     &\tabby \tabby \lVert \widehat{\mathcal{W}}^{h+1} -   \mathcal{W}^{*,h+1} \rVert_F^2 \nonumber \\
        = &\sum_{r=1}^d \lVert \widehat{\bm{W}}_r^{h+1} - \bm{W}^{*,h+1}_r \rVert^2_F \nonumber \\
        \leq &\sum_{r=1}^d \lVert \widehat{\bm{W}}_r^{h+1} - \bm{W}^{*,h}_r \rVert^2_F \nonumber \\
       =  &\sum_{r=1}^d \lVert \widehat{\bm{W}}_r^h - \bm{W}^{*,h}_r + \eta_h \sum_{k=1}^m \left( A_k - \sum_{r'=1}^d \left\{ \widehat{\bm{W}}_{r'}^{h}B(x_k)B(x_k)^{\tp}\left[\widehat{\bm{W}}_{r'}^{h}\right]^{\tp} \right\} \right) \widehat{\bm{W}}_r^{h} B(x_k)B(x_k)^{\tp} \rVert^2_F \nonumber \\
       \leq &\sum_{r=1}^d \lVert \widehat{\bm{W}}_r^h - \bm{W}^{*,h}_r \rVert^2_F + \eta_h^2T_{\mathrm{quad.mean}}^2 + \eta_h^2T_{\mathrm{quad.op}}^2 \nonumber \\
       & + 2\eta_h \sum_{r=1}^d \iprod{\sum_{k=1}^m \left( A_k - \sum_{r'=1}^d \left\{ \widehat{\bm{W}}_{r'}^{h}B(x_k)B(x_k)^{\tp}\left[\widehat{\bm{W}}_{r'}^{h}\right]^{\tp} \right\} \right) \widehat{\bm{W}}_r^{h} B(x_k)B(x_k)^{\tp}}{\widehat{\bm{W}}_r^{h} - \bm{W}^{*,h}_r} \nonumber \\
       \leq &\lVert \widehat{\mathcal{W}}^{h+1} - \mathcal{W}^{*,h+1} \rVert_F^2 + \eta_h^2T_{\mathrm{quad.mean}}^2 + \eta_h^2T_{\mathrm{quad.op}}^2 + \eta_h T_{\mathrm{op}} \nonumber \\
       & - 2\eta_h \sum_{r=1}^d \iprod{\sum_{k=1}^m \sum_{r'=1}^d \left\{ \left(\widehat{\bm{W}}_{r'}^{h}B(x_k)B(x_k)^{\tp}\left[\widehat{\bm{W}}_{r'}^{h}\right]^{\tp} \right\} - Z(x_k)Z(x_k)^{\tp} \right) \widehat{\bm{W}}_r^{h} B(x_k)B(x_k)^{\tp}}{\widehat{\bm{W}}_r^{h} - \bm{W}^{*,h}_r} \nonumber \\
      \leq &\lVert \widehat{\mathcal{W}}^{h} - \mathcal{W}^{*,h} \rVert_F^2 +   \eta_h^2T_{\mathrm{quad.mean}}^2 + \eta_h^2T_{\mathrm{quad.op}}^2 +
       +
      \eta_hT_{\mathrm{op}} + \eta_hT_{\mathrm{approx}} + \eta_hT_{\mathrm{cross}} - \eta_hT_{\mathrm{mean}},
    \end{align}
    where the first inequality follows from the choice of $\mathcal{W}^{*,h+1}$, and the final inequality uses Lemma~\ref{lemma:unfold}.
    We will next bound each of these terms.

    For $T_{\mathrm{quad.mean}}$, we fix $r \in \{1,\ldots,d\}$ and bound each term by
    \begin{align*}
      T_{\mathrm{quad.mean}} &= \sqrt{2} \left\lVert
      \begin{pmatrix}
        \widehat{\bm{W}}_r^{h}B(x_1)B(x_1)^{\tp}(\widehat{\bm{W}}_r^{h})^{\tp} - Z_r(x_1)Z_r(x_1)^{\tp} \\ \vdots \\ \widehat{\bm{W}}_r^{h}B(x_m)B(x_m)^{\tp}(\widehat{\bm{W}}_r^{h})^{\tp} - Z_r(x_m)Z_r(x_m)^{\tp}
      \end{pmatrix}^{\tp}
      \begin{pmatrix}
        \widehat{\bm{W}}_r^h B(x_1)B(x_1)^{\tp} \\ \vdots \\ \widehat{\bm{W}}_r^h B(x_m)B(x_m)^{\tp}
      \end{pmatrix} \right\rVert_F \\
      &\leq \sqrt{2} \left\lVert
      \begin{pmatrix}
        \widehat{\bm{W}}_r^{h}B(x_1)B(x_1)^{\tp}(\widehat{\bm{W}}_r^{h})^{\tp} - Z_r(x_1)Z_r(x_1)^{\tp} \\ \vdots \\ \widehat{\bm{W}}_r^{h}B(x_m)B(x_m)^{\tp}(\widehat{\bm{W}}_r^{h})^{\tp} - Z_r(x_m)Z_r(x_m)^{\tp}
      \end{pmatrix}^{\tp} \right\rVert_F \left\lVert
      \begin{pmatrix}
        \widehat{\bm{W}}_r^h B(x_1)B(x_1)^{\tp} \\ \vdots \\ \widehat{\bm{W}}_r^h B(x_m)B(x_m)^{\tp}
      \end{pmatrix} \right\rVert_2 \\
      &\leq T_{\mathrm{mean}}^{1/2} \gamma_Z \left( \frac{2 c_W m}{q} \right)^{1/2} ,
    \end{align*}
    to conclude
    \begin{equation*}
      T_{\mathrm{quad.mean}}^2 \leq 2c_W d \frac{m \gamma_Z^2}{q} T_{\mathrm{mean}} .
    \end{equation*}

    For $T_{\mathrm{quad.op}}$, we fix $r \in \{1,\ldots,d\}$ and bound each term by
    \begin{align*}
      T_{\mathrm{quad.op}} &\leq \sqrt{2} m \lVert \frac{1}{m} \sum_{k=1}^m \left\{ A_k - Z_r(x_k)Z_r(x_k)^{\tp}\right\} \widehat{\bm{W}}_r^h B(x_k)B(x_k)^{\tp} \rVert_F \\
      &\leq 2 m q^{1/2}  \lVert \frac{1}{m} \sum_{k=1}^m \left\{ A_k - Z_r(x_k)Z_r(x_k)^{\tp} \right\} \widehat{\bm{W}}_r^h B(x_k)B(x_k)^{\tp} \rVert_2 \\
      &\leq \sqrt{2} m q^{1/2}  c_{\mathrm{prob}} c_W \gamma_1 \left( \frac{\sigma^2 q^2 n \log q}{m}\right)^{1/2},
    \end{align*}
    where the final inequality uses Lemma~\ref{lemma:probability}, to conclude
    \begin{equation*}
      T_{\mathrm{quad.op}}^2 \leq 2c_{\mathrm{prob}}^2 c_W^2 d \sigma^2 q^3 m \gamma_Z^2 n \log q.
    \end{equation*}

    For $T_{\mathrm{cross}}$, we bound
    \begin{align*}
      T_{\mathrm{cross}} &\leq \sum_{r=1}^d \left\{ T_{\mathrm{mean}}^{1/2} \left\lVert \begin{pmatrix}
        (\widehat{\bm{W}}_r^h - \bm{W}_r^{*,h}) B(x_1)B(x_1)^{\tp} (\widehat{\bm{W}}_r^h - \bm{W}_r^{*,h})^{\tp} \\ \vdots \\ (\widehat{\bm{W}}_r^h - \bm{W}_r^{*,h}) B(x_m)B(x_m)^{\tp} (\widehat{\bm{W}}_r^h - \bm{W}_r^{*,h})^{\tp}
      \end{pmatrix} \right\rVert_F \right\} \\
      &\leq \sum_{r=1}^d \left\{ T_{\mathrm{mean}}^{1/2} \lVert \widehat{\bm{W}}_r^h - \bm{W}_r^{*,h} \rVert_F \left\lVert \begin{pmatrix}
        (\widehat{\bm{W}}_r^h - \bm{W}_r^{*,h}) B(x_1)B(x_1)^{\tp}  \\ \vdots \\ (\widehat{\bm{W}}_r^h - \bm{W}_r^{*,h}) B(x_m)B(x_m)^{\tp}
      \end{pmatrix} \right\rVert_2 \right\} \\
      &\leq \sum_{r=1}^d \left\{ \left( \frac{C_B m T_{\mathrm{mean}}}{q} \right)^{1/2} \lVert \widehat{\bm{W}}_r^h - \bm{W}_r^{*,h} \rVert_F^2 \right\} \\
      &= \left( \frac{C_B m T_{\mathrm{mean}}}{q} \right)^{1/2} \lVert \widehat{\mathcal{W}}^{h} - \mathcal{W}^{*,h} \rVert_F^2 \\
      &\leq c_{\mathrm{cross}} T_{\mathrm{mean}} + \frac{c_{\mathrm{prev}}C_Bm\gamma_Z^2}{4 c_{\mathrm{cross}} q} \lVert \widehat{\mathcal{W}}^{h} - \mathcal{W}^{*,h} \rVert_F^2
    \end{align*}
    for any positive constant $c_{\mathrm{cross}}$.
    Specifying $c_{\mathrm{cross}} = 1/8$, we conclude
    \begin{equation*}
      T_{\mathrm{cross}} \leq \frac{1}{8} T_{\mathrm{mean}} + \frac{2 c_{\mathrm{prev}}C_Bm\gamma_Z^2}{4 q} \lVert \widehat{\mathcal{W}}^{h} - \mathcal{W}^{*,h} \rVert_F^2 .
    \end{equation*}

    For $T_{\mathrm{op}}$, applying Lemma~\ref{lemma:probability},
    \begin{align*}
      T_{\mathrm{op}} &\leq \sum_{r=1}^d 2m \lVert \frac{1}{m} \sum_{k=1}^m (A_k - Z_r(x_k)Z_r(x_k)^{\tp}) \widehat{\bm{W}}_r^h B(x_k)B(x_k)^{\tp} \rVert_2 \lVert \widehat{\bm{W}}_r^h - \bm{W}_r^{*,h} \rVert_* \\
      &\leq \sum_{r=1}^d m q^{1/2} \left\{ c_{\mathrm{prob}} c_W \gamma_1 \left( \frac{\sigma^2 q^2 n \log q}{m} \right)^{1/2} \right\} \lVert \widehat{\bm{W}}_r^h - \bm{W}_r^{*,h} \rVert_F \\
      &\leq \frac{c_{\mathrm{prob}}^2 c_W^2 d \sigma^2 q^4 n \log q}{4 c_{\mathrm{op}}} + \frac{c_{\mathrm{op}} m \gamma_1^2}{q} \lVert \widehat{\mathcal{W}}^h - \mathcal{W}^{*,h} \rVert_F^2,
    \end{align*}
    where $\lVert \cdot \rVert_*$ denotes the matrix nuclear norm.
    Specifying $c_{\mathrm{op}}=c_B/8$, we conclude
    \begin{equation*}
      T_{\mathrm{op}} \leq \frac{2 c_{\mathrm{prob}}^2 c_W^2 d \sigma^2 q^4 n \log q}{c_B} + \frac{c_B m \gamma_1^2}{8 q} \lVert \widehat{\mathcal{W}}^h - \mathcal{W}^{*,h} \rVert_F^2.
    \end{equation*}

    For $T_{\mathrm{approx}}$, we require one auxiliary result:
    \begin{equation} \label{approx_slice2}
      \sum_{k=1}^m    \|  \left\{ \sum_{r=1}^d \bm{W}_r^{*,h} B(x_k)B(x_k)^{\tp}(\bm{W}_r^{*,h})^{\tp} \right\} - Z(x_k)Z(x_k)^{\tp} \|_F^2 \leq (4c_{\mathrm{approx},2} + 6\kappa) m \gamma_Z^2 \neweps_{\mathrm{approx},2}^{(h)}.
    \end{equation}
    Then
    \begin{align*}
      T_{\mathrm{approx}} &= 2 \bigg\vert \sum_{k=1}^m \bigg\langle \left\{ \sum_{r=1}^d \widehat{\bm{W}}_r^{h}B(x_k)B(x_k)^{\tp} (\widehat{\bm{W}}_r^h)^{\tp} \right\} - Z(x_k)Z(x_k)^{\tp}, \\
      &\tabby\tabby \left\{ \sum_{r'=1}^d \bm{W}_{r'}^{*,h} B(x_k)B(x_k)^{\tp}(\bm{W}_{r'}^{*,h})^{\tp} \right\} - Z(x_k)Z(x_k)^{\tp} \bigg\rangle \bigg\vert \\
      &\leq T_{\mathrm{mean}}^{1/2} \left\{ \sum_{k=1}^m \lVert \left\{ \sum_{r'=1}^d \bm{W}_{r'}^{*,h} B(x_k)B(x_k)^{\tp}(\bm{W}_{r'}^{*,h})^{\tp} \right\} - Z(x_k)Z(x_k)^{\tp} \rVert_F^2 \right\}^{1/2} \\
      &\leq \{(4c_{\mathrm{approx},2} + 6\kappa) m\}^{1/2} \gamma_Z T_{\mathrm{mean}}^{1/2} (\neweps_{\mathrm{approx},2}^{(h)})^{1/2} \\
      &\leq c'_{\mathrm{approx}} T_{\mathrm{mean}} + \frac{(4c_{\mathrm{approx},2} + 6\kappa) m \gamma_Z^2}{4c'_{\mathrm{approx}}} \neweps_{\mathrm{approx},2}^{(h)},
    \end{align*}
    for any positive constant $c'_{\mathrm{approx}}$. The second inequality uses \eqref{approx_slice2}.

    Specifying $c'_{\mathrm{approx}} = 1/4$, we conclude
    \begin{equation*}
      T_{\mathrm{approx}} \leq \frac{1}{4} T_{\mathrm{mean}} + (4c_{\mathrm{approx},2} + 6\kappa) m \gamma_Z^2 \neweps_{\mathrm{approx},2}^{(h)}.
    \end{equation*}

    For $T_{\mathrm{mean}}$, define an $m$-tuple of orthogonal transformations, $\mathcal{Q}^{\mathrm{Proc},h}$ such that for each $k=1,\ldots,m$,
    \begin{equation*}
      Q_k^{\mathrm{Proc},h} = \operatorname{argmin}_{Q \in \mathcal{O}_d} \lVert \widehat{\mathcal{W}}^h \times_2 B(x_k) - Z(x_k)Q \rVert_F^2
    \end{equation*}
    which has a closed form expression \citep{cape19twotoinfinity}.
    Define an operator
    \begin{equation*}
      \mathcal{P}_{\bm{B}}(Z) = Z \times_2 (\bm{B}^{\tp}\bm{B})^{-1} \bm{B}^{\tp},
    \end{equation*}
    and suppose that for an arbitrary $n \times m \times d$ tensor $\mathcal{Z}$ and an $m$-tuple of orthogonal transformations $\mathcal{Q}$, $\mathcal{Z}\mathcal{Q}$ gives the $n \times m \times d$ tensor which right multiplies each component of $\mathcal{Q}$ by the corresponding $n \times d$ slice of $\mathcal{Z}$.
    Then,
    \begin{align*}
      \lVert \widehat{\mathcal{W}}^h - \mathcal{W}^{*,h} \rVert_F^2 &\leq \lVert \widehat{\mathcal{W}}^h - \mathcal{P}_{\bm{B}}(\mathcal{Z}\mathcal{Q}^{\mathrm{Proc},h})\rVert_F^2 \\
      &= \lVert \left\{ \widehat{\mathcal{W}}^h - \mathcal{P}_{\bm{B}}(\mathcal{Z}\mathcal{Q}^{\mathrm{Proc},h}) \right\} \times_2 \left\{ (\bm{B}^{\tp}\bm{B})^{-1} \bm{B}^{\tp} \right\} \bm{B} \rVert_F^2 \\
      &\leq \frac{q}{c_B m} \lVert \widehat{\mathcal{W}}^h \times_2 \bm{B} - \mathcal{P}_{\bm{B}}(\mathcal{Z}\mathcal{Q}^{\mathrm{Proc},h}) \times_2 \bm{B} \rVert_F^2 \\
      &\leq \frac{q}{c_B m} \sum_{k=1}^m \lVert \widehat{\mathcal{W}}^h \bar{\times}_2 B(x_k) - Z(x_k)Q_k^{\mathrm{Proc},h} \rVert_2^2 \\
      &\leq \frac{q}{c_B m 2 (2^{1/2} - 1)\gamma_Z^2} T_{\mathrm{mean}}. \numberthis \label{tmean_lb}
    \end{align*}
    The first inequality follows from the choice of $\mathcal{W}^{*,h}$, the second inequality follows from Assumption~\ref{assump:bs}, the third inequality follows from a projection argument, and the final inequality follows from \cite{ma20universal}, Lemma 28.

    This display implies that
    \begin{equation*}
      \frac{1}{2} T_{\mathrm{mean}} \geq \frac{c_B (\sqrt{2} - 1) m \gamma_Z^2}{q} \lVert \widehat{\mathcal{W}}^h - \mathcal{W}^{*,h} \rVert_F^2.
    \end{equation*}

    Substituting all these inequalities into \eqref{summary_terms2}, we have that
    \begin{align} \label{summary_terms22}
      &\lVert \widehat{\mathcal{W}}^{h+1} - \mathcal{W}^{*,h+1} \rVert_F^2 \nonumber \\
      \leq &\lVert \widehat{\mathcal{W}}^{h} - \mathcal{W}^{*,h} \rVert_F^2 + \eta_h^2T_{\mathrm{quad.mean}}^2 + \eta_h^2T_{\mathrm{quad.op}}^2 + \eta_hT_{\mathrm{op}} \nonumber \\
      &\tabby + \eta_hT_{\mathrm{approx}} + \eta_hT_{\mathrm{cross}} - \eta_hT_{\mathrm{mean}} \nonumber \\
      \leq &\left(1 + \frac{\eta_h c_B m \gamma_Z^2}{8 q} + \frac{2 \eta_h c_{\mathrm{prev}} C_B m \gamma_Z^2}{q} - \frac{\eta_h c_B (2^{1/2} - 1) m \gamma_Z^2}{q} \right)  \lVert \widehat{\mathcal{W}}^{h} - \mathcal{W}^{*,h} \rVert^2_F \nonumber \\
      &\tabby + \left( \frac{\eta_h}{8} + \frac{\eta_h}{4} + 2 \eta_h^2 c_W d \frac{m \gamma_Z^2}{q} - \frac{1}{2}\eta_h \right) T_{\mathrm{mean}} \nonumber \\
      &\tabby + \frac{2 \eta_h c_{\mathrm{prob}}^2 c_W^2 d \sigma^2 q^4 n \log q }{c_B} + 2 \eta_h^2 c_{\mathrm{prob}}^2 c_W^2 d \sigma^2 q^3
      m \gamma_Z^2 n \log q \nonumber \\
      &\tabby + \eta_h (4c_{\mathrm{approx},2} + 6\kappa) m \gamma_Z^2 \neweps_{\mathrm{approx},2}^{(h)}.
    \end{align}
    We consider the coefficients of the first two terms of \eqref{summary_terms22} separately.
    For the second term, expanding $\eta_h \equiv \eta'q/m \gamma_Z^2$, for a constant
    \begin{equation*}
      \eta' = \left( 32 d \left\{\left( \frac{c_B}{16 C_B} \right)^{1/2} + c_{\mathrm{approx},\infty}^{1/2} + \kappa^{1/2} \right\} \right)^{-1},
    \end{equation*}
    we have coefficient
    \begin{equation*}
      \left( \frac{\eta' q}{\gamma_Z^2 m} \right) \left( \frac{1}{8} + {1}{4} + 2 \eta' c_W d - \frac{1}{2} \right) < 0.
    \end{equation*}
    For the first term, again expanding the definition of $\eta_h$, we have coefficient
    \begin{equation*}
      \left(1 - \eta' \left[ c_B (2^{1/2} - 1) - \frac{c_B}{8} - 2 c_{\mathrm{prev}}C_B \right] \right).
    \end{equation*}
    Since $c_{\mathrm{prev}} \leq c_B/(16C_B)$ by assumption, we lower bound the quantity inside the square brackets by $\rho$.
    Then, expanding the definition of $\eta_h$ in the final three terms of \eqref{summary_terms22}, and choosing a constant $c_{\mathrm{step}}$ as in \eqref{onestep2_constants}, we complete the proof.
  \end{proof}

    To complete the proof of Theorem~\ref{thm:main_concur}, we begin by showing that \eqref{onestep2_prev} and $\mathcal{E}$ (defined in Lemma~\ref{lemma:probability}) hold with high probability for all $h \geq 0$, and thus we can repeatedly apply Lemma~\ref{lemma:onestep2}.
    Suppose \eqref{onestep2_prev} holds for all $0 \leq h' \leq h$.
    Then by repeated application of Lemma~\ref{lemma:onestep2},
    \begin{align*}
      \lVert \widehat{\mathcal{W}}^{h+1} - \mathcal{W}^{*,h+1} \rVert^2_F &\leq \lVert \widehat{\mathcal{W}}_1^{0} - \mathcal{W}^{*,0} \rVert^2_F + \frac{c_{\mathrm{step}}}{\rho} \left(  \frac{\sigma^2 q^5 n \log q}{\gamma_Z^2 m} + q \cdot \max_{0 \leq h' \leq h} \neweps_{\mathrm{approx},2}^{(h')} \right) \\
      &\leq \left( c_{\mathrm{init}} + \frac{c_{\mathrm{step}}}{\rho c_{\mathrm{SNR}}^2} + \frac{c_{\mathrm{step}}c_{\mathrm{approx},2}}{\rho} \right) \gamma_Z^2,
    \end{align*}
    which implies that both the inductive step and the base case hold as long as $c_{\mathrm{init}}$ and $c_{\mathrm{approx}}$ are sufficiently small, and $c_{\mathrm{SNR}}$ is sufficiently large.

    In particular, we require that
    \begin{equation} \label{constant_condition}
      c_{\mathrm{init}} + \frac{c_{\mathrm{step}}}{\rho c_{\mathrm{SNR}}^2} + \frac{c_{\mathrm{step}}c_{\mathrm{approx},2}}{\rho} \leq \frac{c_B}{16 C_B}.
    \end{equation}
    Recall that $\rho = c_B/8$ and by assumption,
    \begin{equation*}
      c_{\mathrm{step}} \leq \max \left\{ c_{\mathrm{prob}}^2 \left\{\left( \frac{c_B}{16 C_B} \right)^{1/2} + c_{\mathrm{approx},\infty}^{1/2} + \kappa^{1/2} \right\} \left( \frac{2d}{c_B} + \frac{1}{16} \right), 4c_{\mathrm{approx},2} + 6\kappa \right\}.
    \end{equation*}
    Thus, it is easy to see that there exist positive constants $\nu_1$, $\nu_2$, $\nu_3$, and $\nu_4$ (written in terms of $c_B$, $C_B$, $\kappa$, $d$, and $c_{\mathrm{prob}}$ but free of $n$, $m$, $q$, $\sigma$ and $\gamma_Z$) such that
    \begin{equation} \label{constant_condition_simple}
      c_{\mathrm{SNR}} \geq \nu_1, \tabby c_{\mathrm{init}} \leq \nu_2, \tabby c_{\mathrm{approx},2} \leq \nu_3, \tabby c_{\mathrm{approx},\infty} \leq \nu_4
    \end{equation}
    is sufficient for \eqref{constant_condition} to hold.

    By Assumptions~\ref{assump:snr}-\ref{assump:init} and Lemma~\ref{lemma:probability}, choose $n$ sufficiently large so that
    $nq\log q \geq nq \log 5 + (n+q) \log 9$,
    and both $\mathcal{E}$ and \eqref{constant_condition_simple} hold with probability at least $1-\xi$.

    Thus, we can repeatedly apply Lemma~\ref{lemma:onestep2} to conclude that for any $h \geq 0$,
    \begin{align*}
      \lVert \widehat{\mathcal{W}}^{h} - \mathcal{W}^{*,h} \rVert^2_F \leq &\left( 1 - \eta' \rho \right)^h    \lVert \widehat{\mathcal{W}}^{0} - \mathcal{W}^{*,0} \rVert^2_F \\
      &\tabby + c_{\mathrm{step}} \eta' \frac{\sigma^2 q^5 n \log q}{\gamma_Z^2 m} \sum_{j=0}^{h} \left( 1 - \eta' \rho \right)^{j} + c_{\mathrm{step}} \eta' q \sum_{j=0}^{h} \neweps_{\mathrm{approx},2}^{(h-j)} \left( 1 - \eta' \rho \right)^{j}.
    \end{align*}
    The first two terms have limits in $h$, and for even $h$, the final term satisfies
    \begin{align*}
      \sum_{j=0}^{h} \neweps_{\mathrm{approx},2}^{(h-j)} \left( 1 - \eta' \rho \right)^{j} &\leq \frac{c_{\mathrm{approx},2} \gamma_Z^2}{q} (1 - \eta'\rho)^{h/2} \sum_{j=0}^{h/2} (1 - \eta'\rho)^j + \left( \sup_{j'  > h/2} \neweps_{\mathrm{approx},2}^{(j')} \right) \sum_{j=0}^{h/2} (1 - \eta'\rho)^j \\
      &\leq \frac{c_{\mathrm{approx},2} \gamma_Z^2}{q} (1 - \eta'\rho)^{h/2} \sum_{j=0}^{h/2} (1 - \eta'\rho)^j + \left( \sup_{j'  > h/2} \neweps_{\mathrm{approx},2}^{(j')} \right) \sum_{j=0}^{h/2} (1 - \eta'\rho)^j \\
      &\leq \frac{c_{\mathrm{approx},2} \gamma_Z^2}{\eta' \rho q} (1 - \eta'\rho)^{h/2} + \frac{1}{\eta' \rho} \sup_{j'  > h/2}\neweps_{\mathrm{approx},2}^{(j')}. 
    \end{align*}
    Thus for a constant $C_2' = c_{\mathrm{step}}/\rho + 2$,
    \begin{equation*}
      \limsup_{h \rightarrow \infty} \lVert \widehat{\mathcal{W}}^{h} - \mathcal{W}^{*,h} \rVert^2_F \leq C_2' \left(  \frac{\sigma^2 q^5 n \log q}{\gamma_Z^2 m} + q \limsup_{h \rightarrow \infty} \neweps_{\mathrm{approx},2}^{(h)} \right).
    \end{equation*}
    Then to complete the proof of \eqref{main_concur_concl},
    \begin{align*}
      &\limsup_{h \rightarrow \infty} \frac{1}{m} \sum_{k=1}^m \lVert \widehat{Z}^h(x_k) - Z(x_k)Q^{*,h}_k \rVert_F^2 \\
      = &\limsup_{h \rightarrow \infty} \frac{1}{m} \sum_{k=1}^m \lVert \widehat{\mathcal{W}}^h \bar{\times}_2 B(x_k) - \mathcal{W}^{*,h} \bar{\times}_2 B(x_k) + \mathcal{W}^{*,h} \bar{\times}_2 B(x_k) - Z(x_k) Q^{*,h}_k \rVert_F^2 \\
      \leq &\limsup_{h \rightarrow \infty} \left\{ \frac{2}{m} \lVert (\widehat{\mathcal{W}}^h - \mathcal{W}^{*,h}) \times_2 \bm{B} \rVert_F^2 + 2 \neweps_{\mathrm{approx},2}^{(h)} \right\} \\
      \leq &C_2 \left( \frac{\sigma^2 q^4 n \log q}{\gamma_Z^2 m} + \limsup_{h \rightarrow \infty} \neweps_{\mathrm{approx},2}^{(h)} \right)
    \end{align*}
    for a constant $C_2 = 2C_BC_2' + 2$, as desired.

    Recall that for arbitrary $\xi > 0$, this inequality of random variables holds with probability at least $1-\xi$ for any $n \geq N(\xi)$.
    By assumption,
    \begin{equation*}
      \limsup_{h \rightarrow \infty} \neweps_{\mathrm{approx},2}^{(h)} = O_p(\alpha_n),
    \end{equation*}
    and thus, by definition we conclude that
    \begin{equation*}
      \limsup_{h \rightarrow \infty} \frac{1}{m} \sum_{k=1}^m \lVert \widehat{Z}^h(x_k) - Z(x_k)Q^{*,h}_k \rVert_F^2 = O_p\left(  \frac{\sigma^2 q^4 n \log q}{\gamma_Z^2 m} + \alpha_n \right),
    \end{equation*}
    completing the proof.

\subsection{Proof sketch of Corollary~\ref{cor:main_concur1}}

The proof of Corollary~\ref{cor:main_concur1} proceeds almost identically to Theorem~\ref{thm:main_concur}, specifying $d=1$.

The difference is in the derivation of the upper bound on $T_{\mathrm{mean}}$ in Lemma~\ref{lemma:onestep2}, display \eqref{tmean_lb}, where we must prove that the optimal (Procrustes) alignment is given by the identity transformation for each $k=1,\ldots,m$.

Recall that
\begin{equation*}
  c_{\mathrm{approx},\infty} = \frac{1}{\gamma_Z^2} \max_{k=1,\ldots,m} \lVert \bm{W}_1^* B(x_k) - Z_1(x_k) \rVert_2^2.
\end{equation*}

Following \cite{cape19twotoinfinity}, the optimal one-dimensional alignment is given by
\begin{equation*}
  \operatorname{sign}\{Z_1(x_k)^{\tp}\widehat{\bm{W}}_1^h B(x_k) \}.
\end{equation*}
Thus by Lemma~\ref{lemma:proc_alignment}, it is sufficient to show that
\begin{equation} \label{proc_sufficient}
  2 \lVert \widehat{\bm{W}}_1^h B(x_k) - Z_1(x_k) \rVert_2 \leq \lVert Z_1(x_k) \rVert_2.
\end{equation}
By definition, $\lVert Z_1(x_k) \rVert_2 \geq \gamma_Z$, and by assumption,
\begin{equation*}
  2 \lVert \widehat{\bm{W}}_1^h B(x_k) - Z_1(x_k) \rVert_2 \leq 2 \gamma_Z( c_{\mathrm{prev}}^{1/2} + c_{\mathrm{approx},\infty}^{1/2}).
\end{equation*}
Recall that by an assumption of Lemma~\ref{lemma:onestep2}, $c_{\mathrm{prev}} \leq c_B/16C_B \leq 1/16$.
Thus a sufficient condition for \eqref{proc_sufficient} is $c_{\mathrm{approx},\infty} \leq 1/16$.

We can add this condition on $c_{\mathrm{approx},\infty}$ to the statement of Lemma~\ref{lemma:onestep2}, and the remainder of the proof still goes through with possibly stronger condition
\begin{equation*}
  c_{\mathrm{approx},\infty} \leq \nu_4'
\end{equation*}
in \eqref{constant_condition_simple}.
Then since \eqref{constant_condition_simple} still holds asymptotically under Assumption~\ref{assump:approxes1}, we complete the proof of Corollary~\ref{cor:main_concur1}.

\subsection{Proof of Proposition~\ref{prop:init_concur}}

In this section we prove the main theoretical result stated in Section~\ref{subsec:thm_init} regarding the estimation performance of our local averaged initializer (see Section~\ref{subsec:initializers_new}).

Recall that we assumed partition sets of equal integer size for local averaging.
For $\ell=1,\ldots,L$, let $k^*(\ell)$ denote the median index in $T_{\ell}$.

We first state generic perturbation bound on the ASE error up to unknown rotation.
\begin{lemma} \label{lemma:ase_lemma}
  Let $P=VV^{\tp}$, where $V \in \mathbb{R}^{n \times d}$, and $E \in \mathbb{R}^{n \times n}$ a symmetric matrix.
  Then
  \begin{equation*}
    \min_{Q \in \mathcal{O}_d}\lVert \operatorname{ASE}_d(P+E) - VQ \rVert_F^2 \leq \frac{10 d \lVert E \rVert_2^2}{\lambda_d(P)}.
  \end{equation*}
\end{lemma}

\begin{proof}
  By Lemma 5.4 in \cite{tu16low}, we have
  \begin{equation*}
    \min_{Q \in \mathcal{O}_d}\lVert \operatorname{ASE}_d(P+E) - VQ \rVert_F^2  \leq \frac{1}{2(\sqrt{2}-1)\lambda_d(P)} \lVert [P+E]_{(d)} - P \rVert_F^2
  \end{equation*}
  where $[\cdot]_{(d)}$ denotes rank truncation.
  \begin{align*}
    \lVert [P+E]_{(d)} - P \rVert_F &\leq \sqrt{2d} \lVert [P+E]_{(d)} - P \rVert_2 \\
    &\leq \sqrt{2d} \{ \lVert [P+E]_{(d)} - (P+E) \rVert_2 + \lVert E \rVert_2 \} \\
    &\leq \sqrt{2d} \{ \lambda_{d+1}(P+E) + \lVert E \rVert_2 \} \\
    &\leq 2\sqrt{2d} \lVert E \rVert_2
  \end{align*}
  where the first step uses the relationship between the operator and Frobenius norms for low rank matrices, and the final step uses Weyl's inequality.
  Thus,
  \begin{equation*}
    \min_{Q \in \mathcal{O}_d}\lVert \operatorname{ASE}_d(P+E) - VQ \rVert_F^2  \leq \frac{8d}{2(\sqrt{2}-1)\lambda_d(P)} \lVert E \rVert_2^2
  \end{equation*}
  as desired, since $8/\{2(\sqrt{2}-1)\} \leq 10$.
\end{proof}

The following lemma provides a high probability bound on the errors of the local averages used for initialization.
\begin{lemma} \label{lemma:init_prob}
  Under the setting of Proposition~\ref{prop:init_concur}, with probability at least $1 - 4L\exp(-n)$,
  \begin{equation} \label{init_prob}
    \max_{1 \leq \ell \leq L} \left\lVert \frac{1}{M} \sum_{k \in \tilde{T}_{\ell}} (A_k - \mathbb{E}A_k) \right\rVert_2 \leq c_{\mathrm{prob}} \sigma \sqrt{\frac{n}{M}}
  \end{equation}
\end{lemma}

\begin{proof}
  For $\ell=1,\ldots,L$, let
  \begin{equation*}
    \tilde{E}_{\ell} = \frac{1}{M} \sum_{k \in \tilde{T}_{\ell}} (A_k - \mathbb{E}A_k).
  \end{equation*}
  these matrices are mutually independent over $\ell$, and have independent subgaussian edges with parameter at most $\sigma / \sqrt{M}$.

  Fix $\ell$. By \cite{vershynin18highdimensional}, Corollary 4.4.8, we have with probability at least $1 - 4 \exp(-n)$,
  \begin{equation*}
    \lVert \tilde{E}_{\ell} \rVert_2 \leq c_{\mathrm{prob}} \sigma \sqrt{\frac{n}{M}}
  \end{equation*}
  for some universal constant $c_{\mathrm{prob}}$.

  Then by a union bound, we have that \eqref{init_prob} holds with probability at least $1 - 4L \exp(-n)$, as desired.
\end{proof}

\begin{proof}[Proof of Proposition~\ref{prop:init_concur}]

  First, we want to bound the error of $\widehat{Z}^0_{\ell}$ as an estimator of $Z(x_{k^*(\ell)})$ up to an unknown rotation.
  By the Lemma~\ref{lemma:ase_lemma}, this requires control of the operator norm error
  \begin{equation} \label{local_opnorm}
    \left\lVert \frac{1}{M} \sum_{k \in \tilde{T}_{\ell}} A_k - \Theta(x_{k^*(\ell)}) \right\rVert_2.
  \end{equation}
  for $\ell=1,\ldots,L$.

  The random part of the operator norm error can be controlled with high probability by Lemma~\ref{lemma:init_prob}.
  By triangle inequality and the relationship between norms, the deterministic part is bounded above by
  \begin{equation*}
    \leq \max_{k' \in T_{\ell}} \lVert \Theta(x_{k'}) - \Theta(x_{k^*(\ell)}) \rVert_2.
  \end{equation*}
  \begin{align*}
    &\lVert \Theta(x_{k'}) - \Theta(x_{k^*(\ell)}) \rVert_2 \\
    = &\lVert Z(x_{k'})Z(x_{k'})^{\tp} - Z(x_{k^*(\ell)})Z(x_{k^*(\ell)})^{\tp}\rVert_2 \\
    = &\lVert Z(x_{k'})Z(x_{k'})^{\tp} - Z(x_{k'})Z(x_{k^*(\ell)})^{\tp} + Z(x_{k'})Z(x_{k^*(\ell)})^{\tp} - Z(x_{k^*(\ell)})Z(x_{k^*(\ell)})^{\tp}\rVert_2 \\
    \leq &2 \gamma_Z \left( \lVert Z(x_{k'}) - Z(x_k^*(\ell)) \rVert_F^2 \right)^{1/2} \\
    \leq &2 \gamma_Z \left( \sum_{i,r} \left\{ z_{i,r}(x_{k'}) - z_{i,r}(x_k^*(\ell))\right\}^2 \right)^{1/2} \\
    \leq &\frac{2 K_1 \gamma_Z \sqrt{nd}}{L}
  \end{align*}
  uniformly over $k' \in T_{\ell}$.

  Combining the random and deterministic parts, we get \eqref{local_opnorm} is bounded above by
  \begin{equation*}
    \frac{2 K_1 \gamma_Z \sqrt{nd}}{L} + \frac{c_{\mathrm{prob}} \sigma \sqrt{nL}}{\sqrt{m}}
  \end{equation*}
  with high probability over all $\ell = 1,\ldots,L$.

  Applying Lemma~\ref{lemma:ase_lemma}, we get that for each $\ell=1,\ldots,L$, there exists $\tilde{Q}_{\ell}$ such that
  \begin{equation*}
    \lVert \widehat{Z}_{\ell}^0 - Z(x_{k^*(\ell)})\tilde{Q}_{\ell} \rVert_F 
    \leq \frac{2 K_1 d \sqrt{10 n}}{L} + \frac{c_{\mathrm{prob}} \sigma \sqrt{10 d L n}}{\gamma_Z \sqrt{m}} .
  \end{equation*}

  %
  %
  %

  Now fix an arbitrary $k' \in T_{\ell}$ and note that by the Lipschitz condition,
  \begin{equation*}
    \lVert  Z(x_{k'})\tilde{Q}_{\ell} - Z(x_{k^*(\ell)})\tilde{Q}_{\ell} \rVert_F \leq \frac{K_1 \sqrt{nd}}{L},
  \end{equation*}
  and thus
  \begin{equation*}
    \max_{k' \in \tilde{T}_{\ell}} \lVert \hat{Z}^0_{\ell} - Z(x_{k'})\tilde{Q}_{\ell} \rVert_F \leq \frac{(2\sqrt{10d} + 1)K_1 \sqrt{d n}}{L} + \frac{c_{\mathrm{prob}} \sigma \sqrt{10 d L n}}{\gamma_Z \sqrt{m}}.
  \end{equation*}

  Recall that for arbitrary $k \in \{1,\ldots,m\}$, we defined an initializer for the unknown process snapshots as
  \begin{equation*}
    \widehat{Z}^0(x_k) = \{\widehat{Z}^0_{\ell} : k \in T_{\ell}\},
  \end{equation*}
  which suggests a natural set of initial coordinates
  \begin{equation*}
    \widehat{\mathcal{W}}^0 = \widehat{\mathcal{Z}}^0 \times_2 (\bm{B}^{\tp}\bm{B})^{-1} \bm{B}^{\tp} \in \mathbb{R}^{n \times q \times d}.
  \end{equation*}
  Analogously define an initial target process which is aligned to the initializer,
  \begin{equation*}
    \widetilde{\mathcal{Z}}^0 = \{ Z(x_k)\tilde{Q}_{\ell} : k \in T_{\ell}\}_{k=1}^m.
  \end{equation*}

  \begin{align*}
    \lVert \widehat{\mathcal{W}}^0 - \mathcal{W}^{*,0} \rVert_F^2 &\leq \lVert (\widehat{\mathcal{Z}}^0 - \widetilde{\mathcal{Z}}^{0}) \times_2 (\bm{B}^{\tp}\bm{B})^{-1} \bm{B}^{\tp} \rVert_F^2 \\
    &\leq \frac{C_B q}{c_B^2 m} \sum_{\ell=1}^q \sum_{k' \in \tilde{T}_{\ell}} \lVert \hat{Z}^0_{\ell} - Z(x_{k'})\tilde{Q}_{\ell} \rVert_F^2 \\
    &\leq \left\{ \frac{C_B q}{c_B^2} \left( \frac{(2\sqrt{10d} + 1)K_1 \sqrt{d n}}{\gamma_Z L} + \frac{c_{\mathrm{prob}} \sigma \sqrt{10 d L n}}{\gamma_Z^2 \sqrt{m}} \right)^2 \right\} \gamma_Z^2,
  \end{align*}
  where the first inequality follows by definition of the target coordinates, since it minimizes the error over rotations of the true processes.

  This nonasymptotic bound which holds with probability at least $1 - 4
  L \exp(-n)$, as desired.
\end{proof}

\ifjasa
\section{Initialization for estimation algorithms} \label{subsec:initializers}

To initialize gradient descent, we propose a family of kernel smoothed embedding algorithms.
The basic idea is to get local estimates of $Z(x)$ at a user-specified grid of indices.
Then each component function is projected into the $\operatorname{span}(B)$ to recover initial estimates for the basis coordinates. As in Section~\ref{subsec:basis_approx}, we will work with prespecified $d$, and $q$-dimensional basis $B$.

We start with the formal definition of adjacency spectral embedding (ASE) from \citep{tang13universally}. For an $n \times n$ symmetric matrix $P$ with eigendecomposition $Y \Lambda Y^{\tp}$, define
\ifjasa
  $\operatorname{ASE}_d(P) = Y_d \Lambda_d^{1/2}$,
\else
\begin{equation*}
  \operatorname{ASE}_d(P) = Y_d \Lambda_d^{1/2} \in \mathbb{R}^{n \times d},
\end{equation*}
\fi
where $Y_d \in \mathbb{R}^{n \times d}$ and $\Lambda_d \in \mathbb{R}^{d \times d}$ correspond to the $d$ largest eigenvectors and eigenvalues of $P$. If the diagonal entries of $\Lambda$ are distinct, then the ASE is uniquely defined up to sign flips of each column.

Our initialization algorithm takes the network snapshots $\{A_k\}_{k=1}^m$ and snapshot indices $\{x_k\}_{k=1}^m$ as input, and further depends on a user specified kernel $K(y)$, and grid of indices $\{\tilde{x}_1,\ldots,\tilde{x}_{q'}\} \subset \mathcal{X}$ for some $q'$ satisfying $q \leq q' \leq m$.
We assume that $K(y)$ is nonnegative and normalized with respect to the snapshot indices, so that, for any $y \in \mathcal{X}$,
\ifjasa
  $\sum_{k=1}^m K(y - x_k) = 1$.
\else
\begin{equation*}
  \sum_{k=1}^m K(y - x_k) = 1 \, .
\end{equation*}
\fi
Moreover, define
\ifjasa
  a $q \times q'$ matrix $\bm{B}_{\mathrm{init}} = [B(\tilde{x}_1) \cdots B(\tilde{x}_{q'}) ]^{\tp}$,
\else
\begin{equation*}
  \bm{B}_{\mathrm{init}} = \begin{pmatrix} B(\tilde{x}_1) & \cdots & B(\tilde{x}_{q'}) \end{pmatrix}^{\tp},
\end{equation*}
\fi
and assume that the indices in the grid are well spaced such that $\bm{B}_{\mathrm{init}}^{\tp}\bm{B}_{\mathrm{init}}$ is invertible.
With this notation in hand, we may state our initialization algorithm.
\begin{algorithm} \label{alg:init_alg}
  \ifjasa
  \onehalfspacing
  \else
  \doublespacing
  \fi
  \caption{Kernel smoothed embedding initialization algorithm.}
  \begin{tabbing}
  \enspace For $\ell=1$ to $\ell=q'$ \\
  \qquad Set $\widehat{Z}_{\mathrm{init}}(\tilde{x}_{\ell}) = \operatorname{ASE}_d \left( \sum_{k=1}^m K(\tilde{x}_{\ell} - x_k) A_k \right)$ \\
  \qquad If $\ell>1$ then \\
  \qquad\qquad Set $O_{\ell} = \operatorname{argmin}_{O \in \mathcal{O}_d} \lVert \widehat{Z}_{\mathrm{init}}(\tilde{x}_{\ell})O - \widehat{Z}_{\mathrm{init}}(\tilde{x}_{\ell-1}) \rVert^2_F$ \\
  \qquad\qquad $\widehat{Z}_{\mathrm{init}}(\tilde{x}_{\ell}) \leftarrow \widehat{Z}_{\mathrm{init}}(\tilde{x}_{\ell})O_{\ell}$ \\
  \enspace Set $\mathcal{Z}_{\mathrm{init}} = \left\{\widehat{Z}_{\mathrm{init}}(\tilde{x}_{\ell}) \right\}_{\ell=1}^{q'} \in \mathbb{R}^{n \times q' \times d}$ \\
  \enspace Set $\widehat{\mathcal{W}}^0 = \mathcal{Z}_{\mathrm{init}} \times_2 (\bm{B}_{\mathrm{init}}^{\tp}\bm{B}_{\mathrm{init}})^{-1} \bm{B}_{\mathrm{init}}^{\tp}$ \\
  \enspace Output $\widehat{\mathcal{W}}^0$
\end{tabbing}
\end{algorithm}

Algorithm~\ref{alg:init_alg} includes an alignment step which orthogonally transforms the columns of each embedding to minimize the discrepancies between consecutive embeddings.
This helps to produce $\mathcal{Z}_{\mathrm{init}}$ which is close to a relatively smooth in $x$ representative of $\mathcal{T}(Z)$, increasing the potential to share information across network snapshots.
The optimal transformations each solve a so-called Procrustes problem, which has a closed form solution \citep{cape19twotoinfinity}.

In our implementation of Algorithm~\ref{alg:init_alg} in Sections~\ref{sec:simulations} and \ref{sec:real_data}, we choose $q'=q$ and use an equally spaced grid $\{\tilde{x}_1,\ldots,\tilde{x}_q\} \subset \mathcal{X}$, and a normalized rectangular kernel.
With equally spaced snapshot indices on $\mathcal{X}$, this implies that each local embedding is the $d$-dimensional ASE of the local average of the closest snapshots to a given grid point.
\else
\fi

\section{Derivation of NGCV criterion} \label{app:tuning}

In this appendix we will derive the specific form of the NGCV criterion \eqref{ngcv}.
First, recall the standard generalized cross validation (GCV) criterion for linear regression, which is derived based on leave one out cross validation \citep{golub79generalized}.
Suppose we have univariate responses $y_i$ and $p$-dimensional predictors $\bm{x}_i$ for $i=1,\ldots,n$, with $n \times p$ design matrix $\bm{X}$.
The least squares regression coefficients for this problem are $\hat{\beta} = (\bm{X}^{\tp}\bm{X})^{-1}\bm{X}^{\tp}\bm{y}$ and the {\em hat matrix} is given by $\bm{H} = \bm{X}(\bm{X}^{\tp}\bm{X})^{-1}\bm{X}^{\tp}$.
Then the generalized cross validation criterion \citep{golub79generalized} is
\begin{equation*}
  \sum_{i=1}^n \left( \frac{y_i - \bm{x}_i^{\tp}\hat{\beta}}{1 - [\bm{H}]_{ii}}\right)^2.
\end{equation*}
A common approximation is to replace each of the diagonal elements of $\bm{H}$ with their mean, resulting in the criterion
\begin{equation*}
  \left\{ 1 - \frac{\operatorname{tr}(\bm{H})}{n} \right\}^{-2} \left\{ \frac{1}{n} \sum_{i=1}^n (y_i - \bm{x}_i^{\tp}\hat{\beta})^2 \right\}.
\end{equation*}
In least squares, $\operatorname{tr}(\bm{H}) = p$, so we get the further simplification
\begin{equation*}
  \left( 1 - \frac{p}{n} \right)^{-2} \left\{ \frac{1}{n} \sum_{i=1}^n (y_i - \bm{x}_i^{\tp}\hat{\beta})^2 \right\}.
\end{equation*}
which can be calculated directly from the mean of squared residuals and the dimensions of the linear regression problem.

This simplified GCV criterion will be the building block of the NGCV criterion.
Recall from \eqref{snap_opt} the form of $\ell(\mathcal{W})$, which we will expand in the following way in terms of individual basis coordinates $\bm{w}_{i,r}$ for $i=1,\ldots,n$ and $r=1,\ldots,d$:
\begin{equation*}
  \ell(\mathcal{W}) = \sum_{i=1}^n \sum_{j=1}^n \sum_{k=1}^m \left\{ [A_k]_{ij} - \sum_{r=1}^d \bm{w}_{i,r}^{\tp} B(x_k) B(x_k)^{\tp} \bm{w}_{j,r} \right\}^2.
\end{equation*}
Note that we can view this as a summation of $2n$ least squares objectives, where the first $n$ fix $i$, the snapshot row, and the second $n$ fix $j$, the snapshot column, and sum over the other two indices
However, to avoid double counting matrix entries, we suppose that each $[A_k]_{ij}$ is assigned in a balanced way to either the row $i$ problem or the column $j$ problem.
As a result, each problem is a least squares problem with $nm/2$ observations.

In truth, the basis coordinates are shared between these problems and each problem involves the coordinates for all the nodes.
Without loss of generality, consider the row 1 problem.
For simplicity, suppose the problem only optimizes over the $qd$ total coordinates $\{\bm{w}_{1,r}\}_{r=1}^d$, treating the others as fixed. Thus, plugging in the fitted coordinates $\widehat{W}$ and summing over the $nm/2$ observations for this problem, we can calculate the simplified GCV criterion as
\begin{equation*}
  \left(1 - \frac{2qd}{nm} \right)^{-2} \left[ \frac{2}{nm} \sum_j \sum_k \left\{ [A_k]_{1j} - \sum_{r=1}^d \hat{\bm{w}}_{1,r}^{\tp} B(x_k) B(x_k)^{\tp} \hat{\bm{w}}_{j,r} \right\}^2 \right].
\end{equation*}
Each of the component squared errors in $\ell(\mathcal{W})$ appears in exactly one of the $2n$ problems.
Thus taking a mean of GCV's over these problems we get the overall criterion
\begin{equation*}
  \left(1 - \frac{2qd}{nm} \right)^{-2} \left\{ \frac{1}{mn^2} \ell(\widehat{W}) \right\},
\end{equation*}
which is equivalent to the NGCV criterion \eqref{ngcv}.

\section{Additional evaluation on synthetic networks} \label{app:one_rotation}

\subsection{Recovery up to a single unknown orthogonal transformation}

In this appendix we evaluate FASE on synthetic functional network data in terms of latent process recovery up to a single unknown orthogonal transformation.
As additional post-processing, we perform a sequential Procrustes alignment for a collection of snapshot process estimates, similar to the alignment procedure used in Algorithm~\ref{alg:init_alg}, and used as post-processing for the FASE estimate in Section~\ref{sec:real_data}.
We will use this procedure to unambiguously select representatives of the unidentified classes $\mathcal{T}(Z)$ for the ground truth latent processes, and $\mathcal{T}(\widehat{Z})$ for the FASE estimator.
Formally, suppose we have latent processes $\widetilde{Z}$, evaluated at indices $\{y_1,\ldots,y_{m'}\} \subset \mathcal{X}$ and stored in an $n \times m' \times d$ tensor.
Then the sequential Procrustes alignment procedure $\operatorname{Proc}_{m'}$ sets $\widetilde{O}_1 = I_d$, then for $k=2,\ldots,m'$, replaces the $k$th $n\ \times d$ slice $\widetilde{Z}(y_k)$ by $\widetilde{Z}(y_k)\widetilde{O}_k$, where
\begin{equation*}
  \widetilde{O}_k = \operatorname{argmin}_{O \in \mathcal{O}_d} \lVert \widetilde{Z}(y_k)O - \widetilde{Z}(y_{k-1})\widetilde{O}_{k-1} \rVert_F^2.
\end{equation*}
For simplicity, we will set $m'=m$ and compute $\operatorname{Proc}_{m}$ using the same snapshot times used to generate the data.
As the FASE estimator is well-defined for any $x \in \mathcal{X}$, we can evaluate this same sequential Procrustes alignment for arbitrarily fine grids.
We stress that this alignment procedure is completely internal to its argument $\widetilde{Z}$, and does not require oracle knowledge of any ground truth $Z$.

We will compare FASE to the same baseline approaches for the same scenarios and settings as in Section~\ref{sec:simulations},
\ifjasa
  but with $\sigma \in \{2,4,6,8\}$, and
\else
  but
\fi
with a new error metric given by
\begin{equation*}
  \mathrm{Err}^*_Z(\widehat{Z}) = \min_{Q_0 \in \mathcal{O}_d} \left\{
  \frac{1}{ndm} \sum_{k=1}^m \left\lVert \operatorname{Proc}_m\{\widehat{Z}\}(x_k) -  \operatorname{Proc}_m\{Z\}(x_k)Q_0 \right\rVert_F^2 \right\}^{1/2}.
\end{equation*}
In contrast to $\mathrm{Err}_Z$ defined in Section~\ref{sec:simulations}, this metric only requires optimization over a single orthogonal transformation-valued argument.
We also plot $\mathrm{Err}_Z$ for the oracle version of FASE (FASE (ORC, ErrZ)) as an achieveable lower bound for $\mathrm{Err}^*_Z$ for the same oracle FASE estimator.

Many of the overall conclusions from these plots are the same as in Section~\ref{sec:simulations}, although we will highlight some key differences in performance seen as a result of switching error metrics.

In Figures~\ref{fig:i_varym_onerot} and \ref{fig:i_varyn_onerot}, we report results for scenario (i).
In all settings, FASE shows a modest difference between $\mathrm{Err}^*_Z$ and the lower bound $\mathrm{Err}_Z$.
The difference is most pronounced for small $m$ and $n$, where the larger estimation error is magnified by the sequential Procrustes alignment, and for large $m$, as the domain of the optimization in the definition of $\mathrm{Err}_Z$ grows relative to the analogous domain in the definition of $\mathrm{Err}_Z^*$.

In Figure~\ref{fig:error_scatter}, we show the relationship between the two error metrics for scenario (i) and the setting with $\sigma=2$, $n=100$ and $m=20$.
We can see that in the majority of cases for FASE (left panel), both metrics perform well, with comparable error that exceeds the best results of ASE (right panel).
Moreover, in terms of $\mathrm{Err}_Z$, FASE outperforms ASE in about 90\% of iterations.
However, the corresponding $\mathrm{Err}_Z^*$ can be much larger, a phenomenon that occurs only once for ASE.
In these cases, the FASE estimate also tends to have slighly higher $\mathrm{Err}_Z$, suggesting it has attempted to ``smooth out'' a discontinuity in the sequence of aligning orthogonal transformations used to evaluate $\mathrm{Err}_Z$, and converged to a local minimum of the objective which cannot take full advantage of the parametric form of the true processes.
This phenomenon appears to occur infrequently for sufficiently large values of $m$ or $n$.

\begin{figure}
  \centering
  \includegraphics[width=\textwidth]{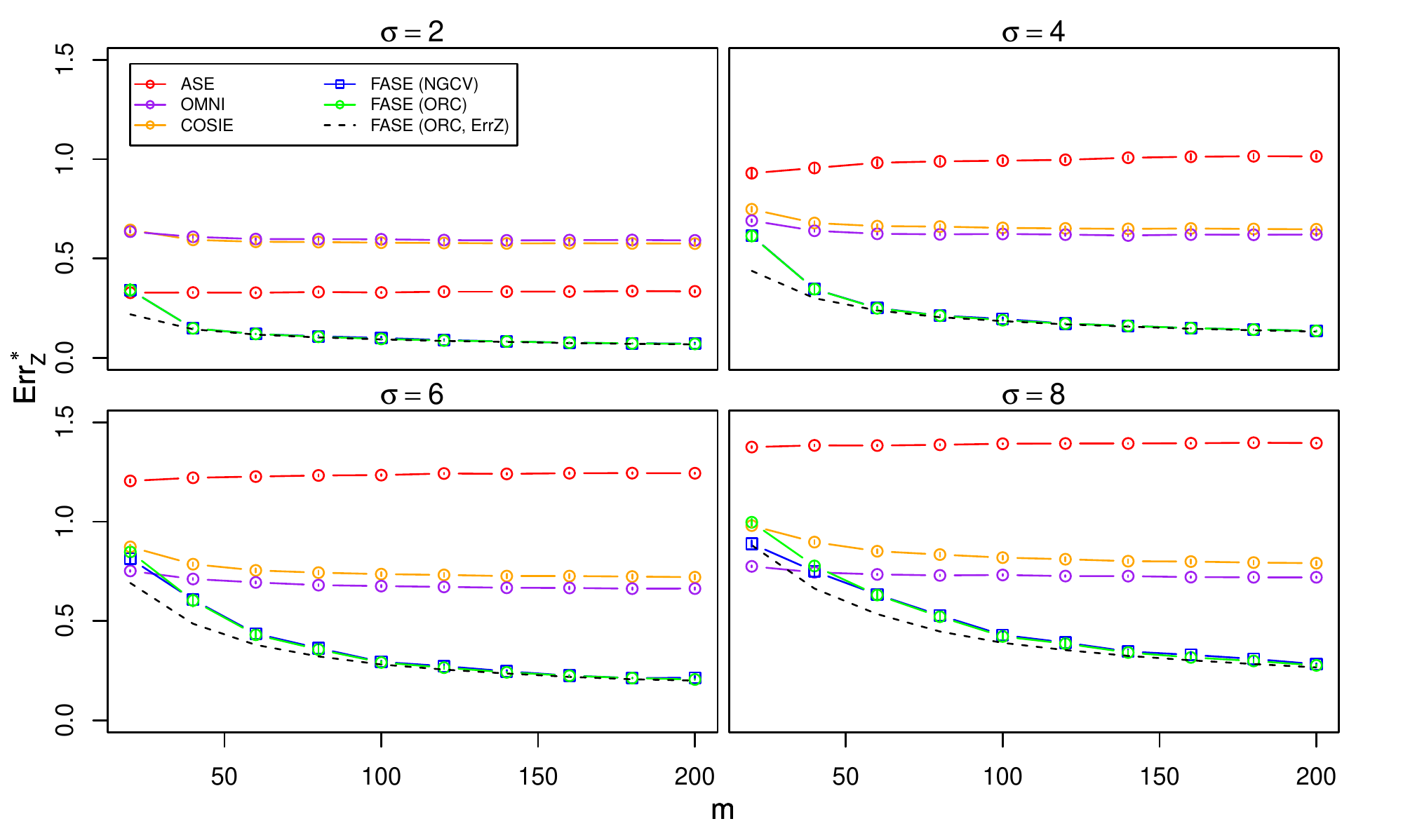}
  \caption{Mean of $\mathrm{Err}_Z^*$, varying $m$, the number of snapshots. Scenario (i), parametric Gaussian networks. Plots are labeled by edge standard deviation $\sigma$.}
  \label{fig:i_varym_onerot}
\end{figure}

\begin{figure}
  \centering
  \includegraphics[width=\textwidth]{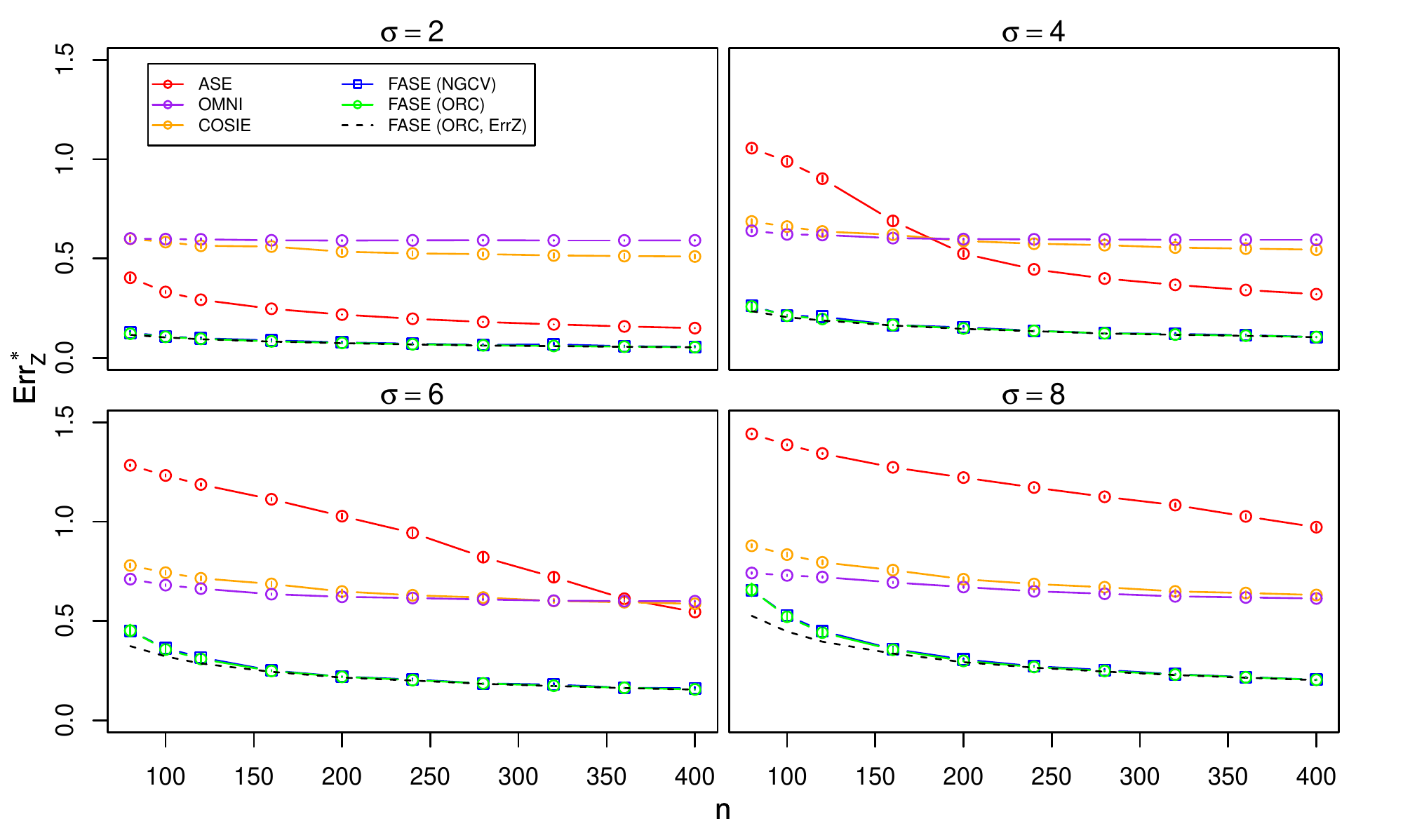}
  \caption{Mean of $\mathrm{Err}_Z^*$, varying $n$, the number of nodes. Scenario (i), parametric Gaussian networks. Plots are labeled by edge standard deviation $\sigma$.}
  \label{fig:i_varyn_onerot}
\end{figure}

\begin{figure}
  \centering
  \includegraphics[width=\textwidth]{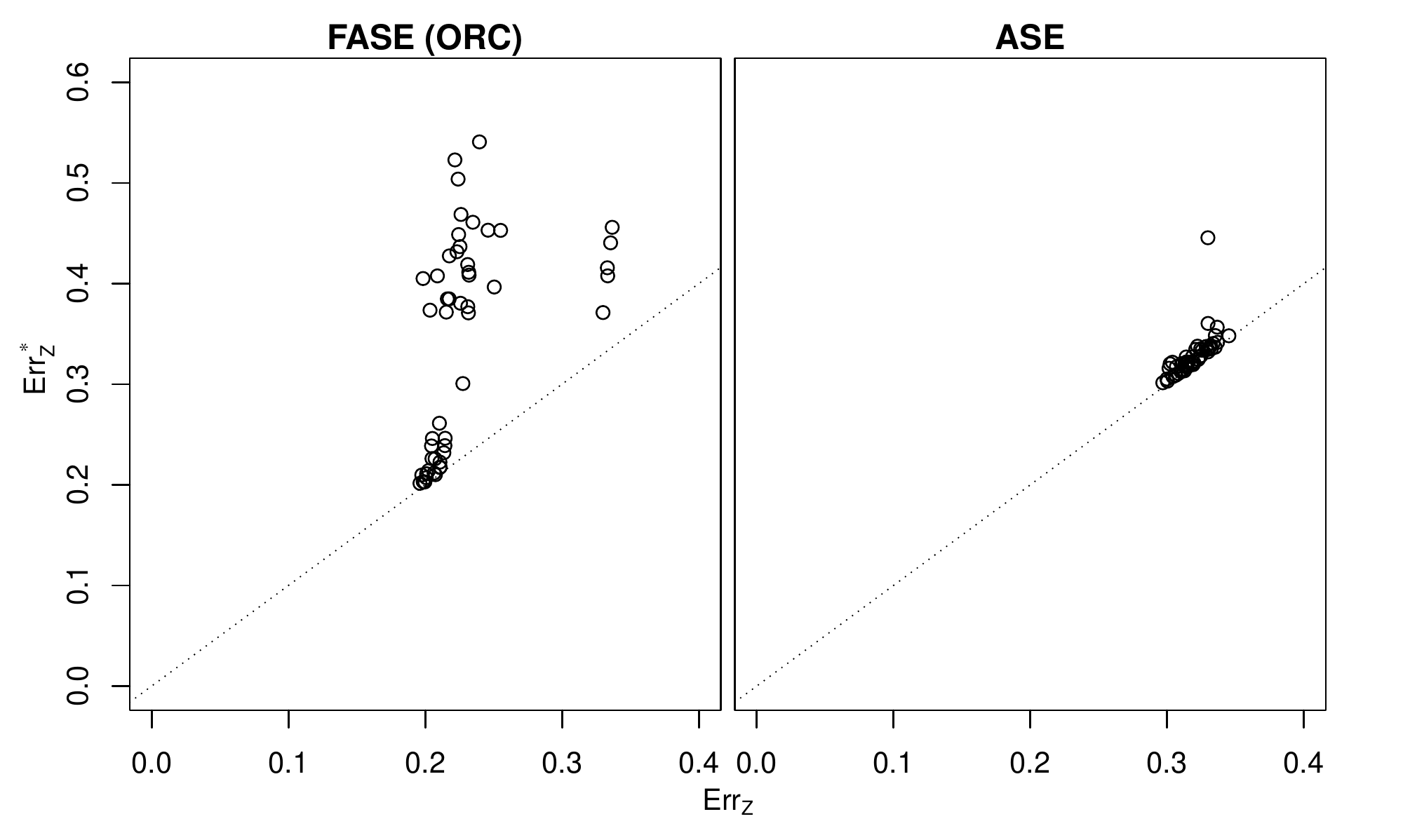}
  \caption{Scatter plot of $\mathrm{Err}_Z$ against $\mathrm{Err}_Z^*$ for FASE (ORC) (left panel) and ASE (right panel). Scenario (i), $\sigma=2$, $n=100$, $m=20$. Dotted lines denote $x=y$.}
  \label{fig:error_scatter}
\end{figure}

In Figures~\ref{fig:ii_varym_onerot} and \ref{fig:ii_varyn_onerot}, we report results for scenario (ii).
In this nonparametric scenario, switching error metrics has a more substantial effect on the performance of FASE.
While we still see a decrease in $\mathrm{Err}_Z^*$ when increasing either $m$ or $n$, there are now high signal settings for $\sigma \leq 4$ in which FASE does not clearly dominate ASE, even for relatively large values of $m$ and $n$.

\begin{figure}
  \centering
  \includegraphics[width=\textwidth]{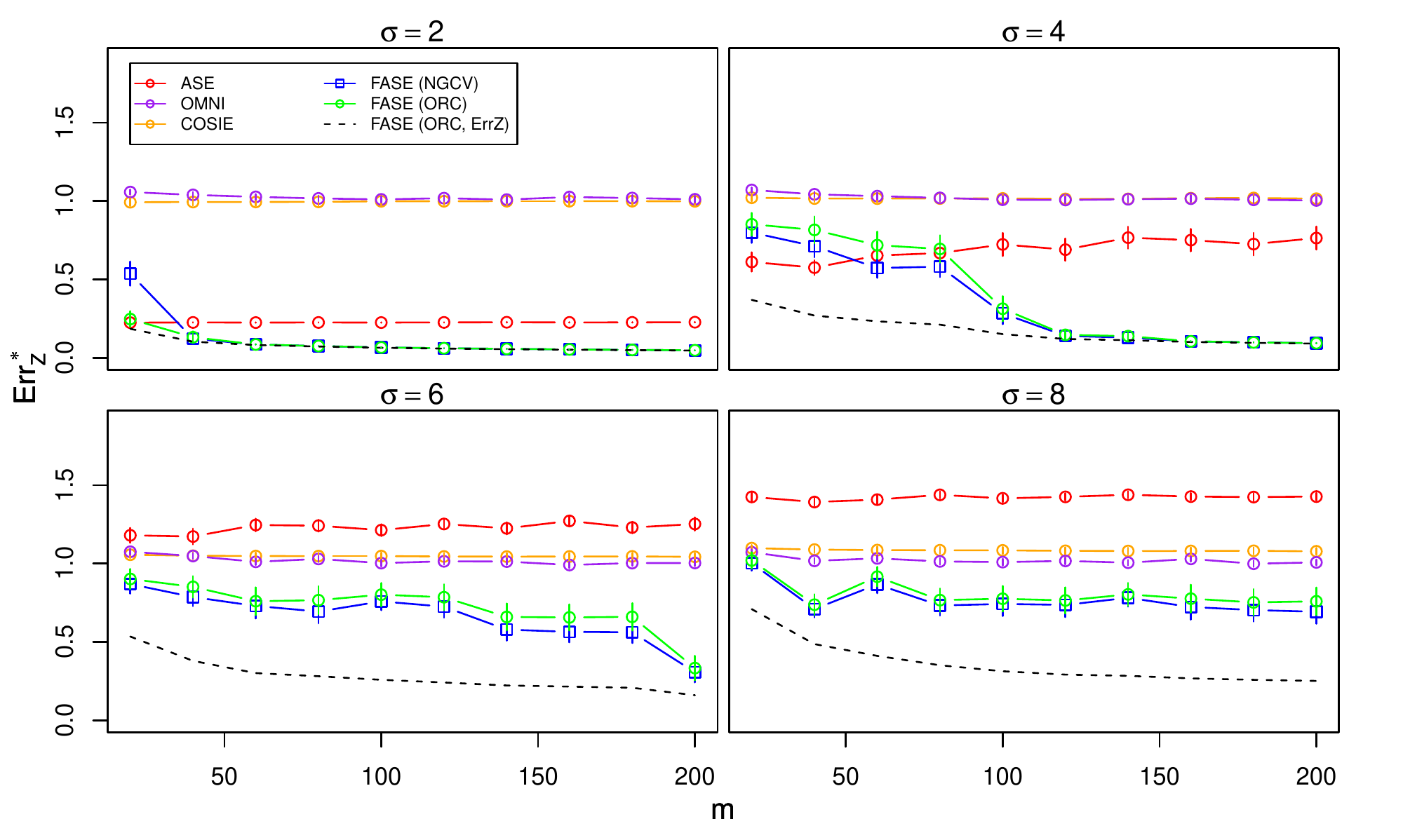}
  \caption{Mean of $\mathrm{Err}_Z^*$, varying $m$, the number of snapshots. Scenario (ii), nonparametric Gaussian networks. Plots are labeled by edge standard deviation $\sigma$.}
  \label{fig:ii_varym_onerot}
\end{figure}

\begin{figure}
  \centering
  \includegraphics[width=\textwidth]{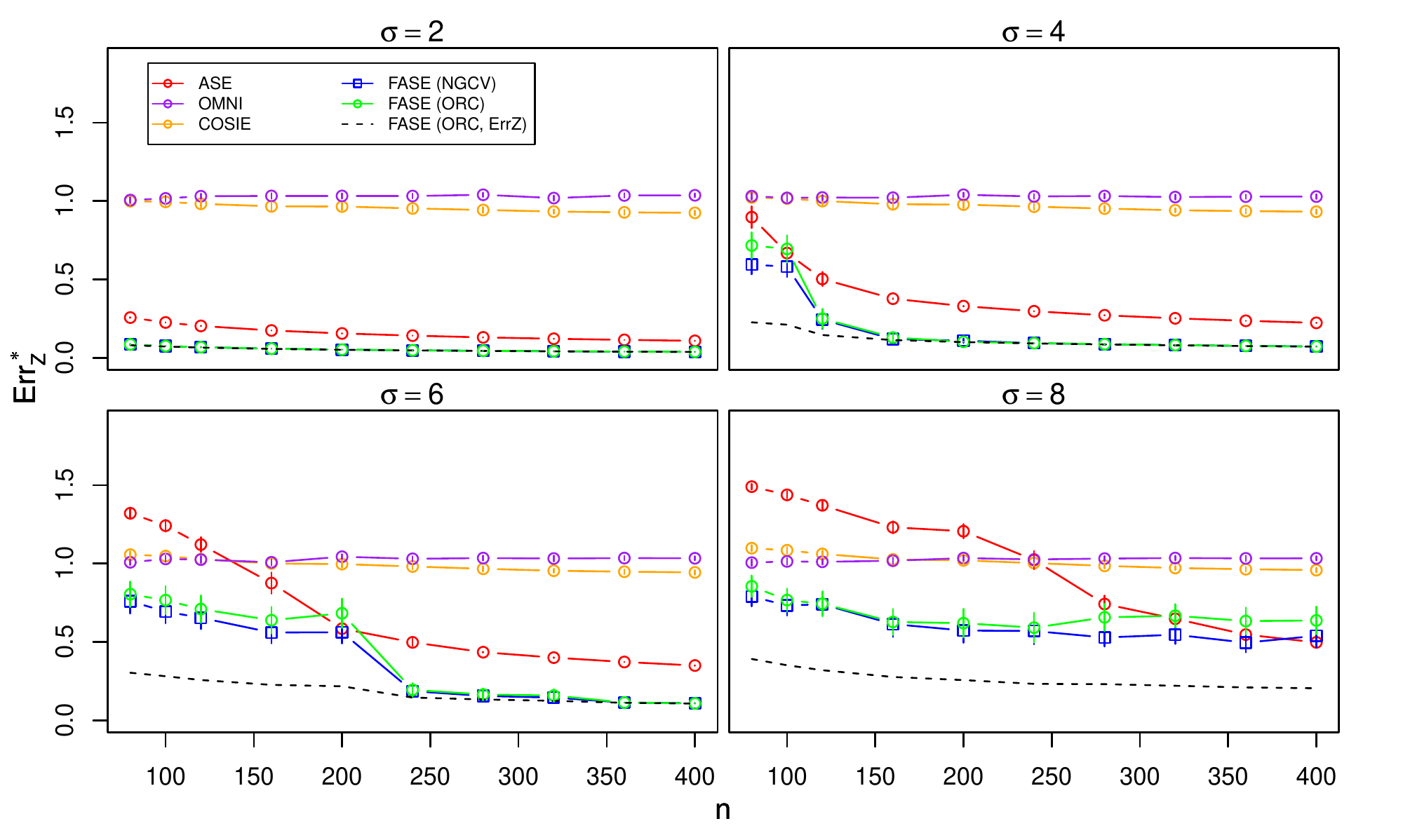}
  \caption{Mean of $\mathrm{Err}_Z^*$, varying $n$, the number of nodes. Scenario (ii), nonparametric Gaussian networks. Plots are labeled by edge standard deviation $\sigma$.}
  \label{fig:ii_varyn_onerot}
\end{figure}

In Figure~\ref{fig:iii_varyall_onerot}, we report results for scenario (iii).
In this scenario, for sufficiently large $n$, and most values of $m$, there appears to be very little difference in the two error metrics.
Especially for functional networks with edge density $1/2$, $\mathrm{Err}_Z^*$ can increase as $m$ increases.
This phenomenon is considered above for scenario (i), see Figure~\ref{fig:error_scatter} and the related discussion.

\begin{figure}
  \centering
  \includegraphics[width=\textwidth]{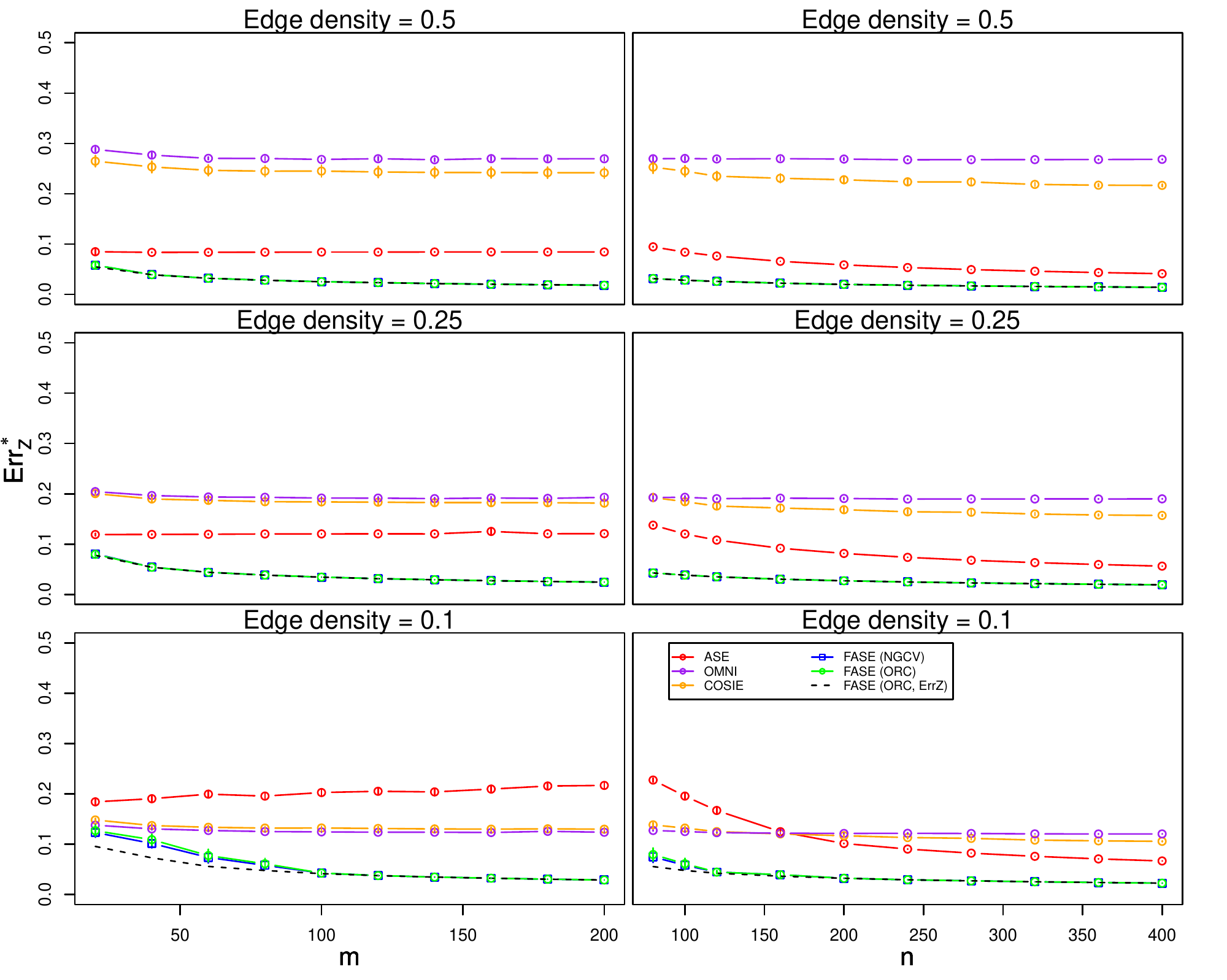}
  \caption{Mean of $\mathrm{Err}_Z^*$, varying $m$, the number of snapshots (left column), and $n$, the number of nodes (right column). Scenario (iii), parametric RDPG networks. Plots are labeled by edge density.}
  \label{fig:iii_varyall_onerot}
\end{figure}


\subsection{Interpolation for missing snapshots} \label{subsec:impute}

In this section we evaluate our FASE estimator against other baselines in the literature as a tool for interpolation to predict edges that were not observe.   We generate data similar to scenario (ii) in the manuscript as follows.
We set $n=100$, $d=2$, and generate latent  processes according to
\begin{equation*}
  z_{i,r}(x) = \frac{3 \sin [C\pi (2x - U_{i,r})]}{1 + 5[x + B_{i,r}(1-2x)]} + G_{i,r}
\end{equation*}
for a constant $C$, where $U_{i,r} \iid \mathrm{Unif}[0,1]$, $B_{i,r} \iid \mathrm{Bernoulli}(1/2)$, and $G_{i,r} \iid \mathcal{N}(0,1/4)$.
Setting $C=2$, the processes go through two complete cycles in the index space $\mathcal{X}=[0,1]$, as in Section~\ref{sec:simulations}.
We also consider smoother processes which only go through one complete cycle by setting $C=1$.

We initially generate $m=100$ equally spaced snapshots on index space $\mathcal{X}=[0,1]$.
We then uniformly select a random snapshot index $x_{k}^* \in [0.25,0.5]$ (to avoid boundary effects) and remove it along with $M$ snapshots immediately before and after $x_{k}^*$, for $M=0,1,...,10$.
That is, we treat the $2M+1$ network snapshots closest to the selected snapshot in the index space as missing.

For all of the dynamic network embedding methods we consider, our goal will be estimating the expected adjacency matrix of the central unobserved snapshot, evaluated in terms of the RMSE
\begin{equation*}
  \mathrm{Err}_{\Theta-\mathrm{mid}}(\widehat{Z}) = \frac{1}{n} \lVert \widehat{Z}(x_k^*)\widehat{Z}(x_k^*)^{\tp} - \Theta(x_k^*) \rVert_F.
\end{equation*}

For FASE, estimation of $Z$ for an unobserved snapshot index is simple given the basis design and estimated coordinates.
To compare to ASE, we consider estimation based on an embedding of the closest observed snapshot, either the next smallest (ASE (below)) or the next largest (ASE (above)).
To compare to OMNI and COSIE, which produce embeddings nearly constant in the index space, we average the estimated embeddings for the next smallest and next largest observed snapshots.

Note that in contrast to the examples in the main paper, the snapshot indices for this partially missing data are not equally spaced.
Thus to choose the basis for the  FASE estimator, rather than specifying equally spaced knots in index space, we place them at equally spaced quantiles of the observed snapshot times.
We also consider both oracle and NGCV-selected parameters for FASE, where $d$ is selected form $\{1,2,3,4\}$, and $q$ from $\{6,8,10,12,14,16\}$.

The results are presented in Figure~\ref{fig:impute_zf1} for processes which complete one cycle, and Figure~\ref{fig:impute_zf2} for processes which complete two cycles.

\begin{figure}
  \centering
  \includegraphics[width=\textwidth]{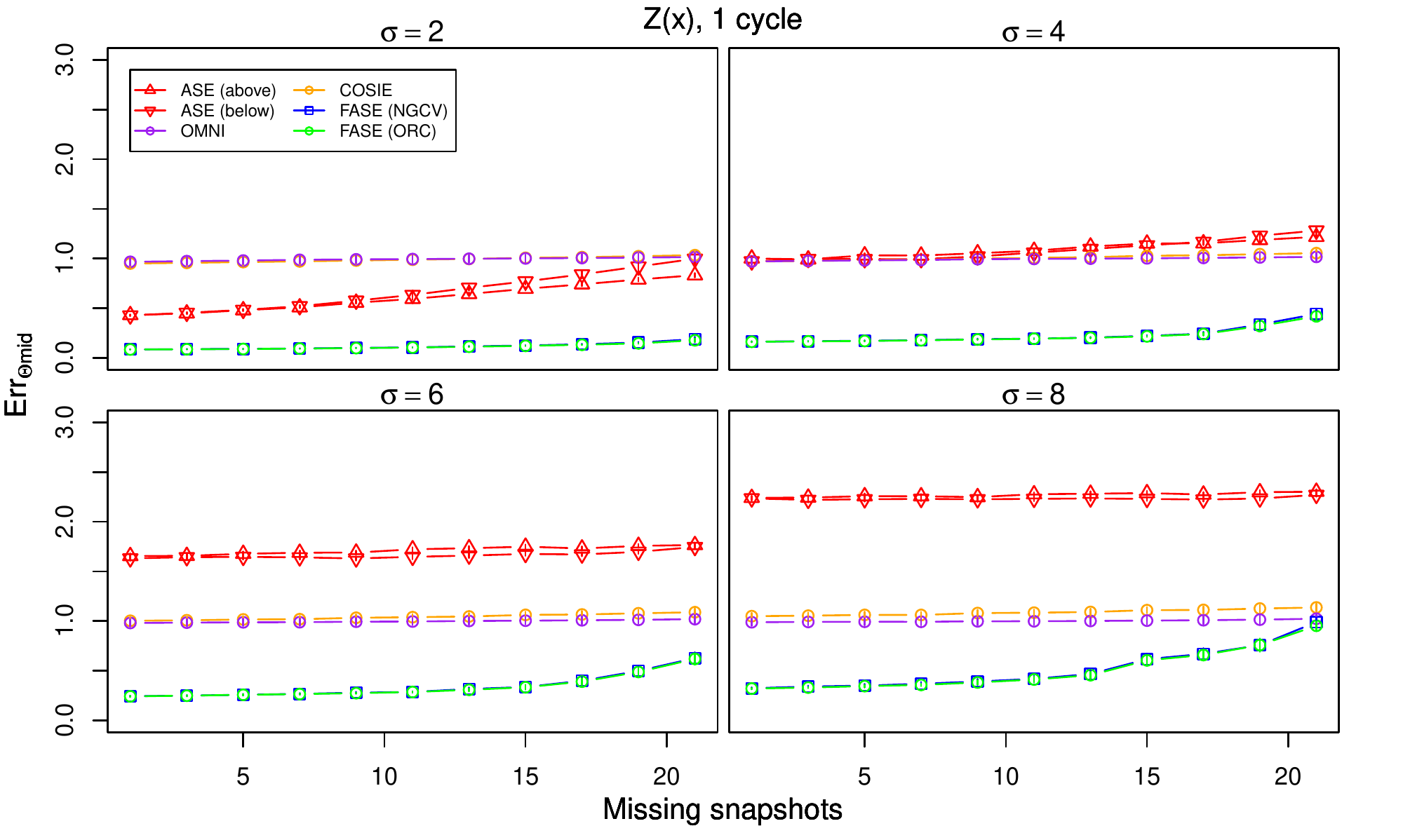}
  \caption{Prediction performance for central unobserved snapshot, latent processes complete one cycle in $[0,1]$.}
  \label{fig:impute_zf1}
\end{figure}

\begin{figure}
  \centering
  \includegraphics[width=\textwidth]{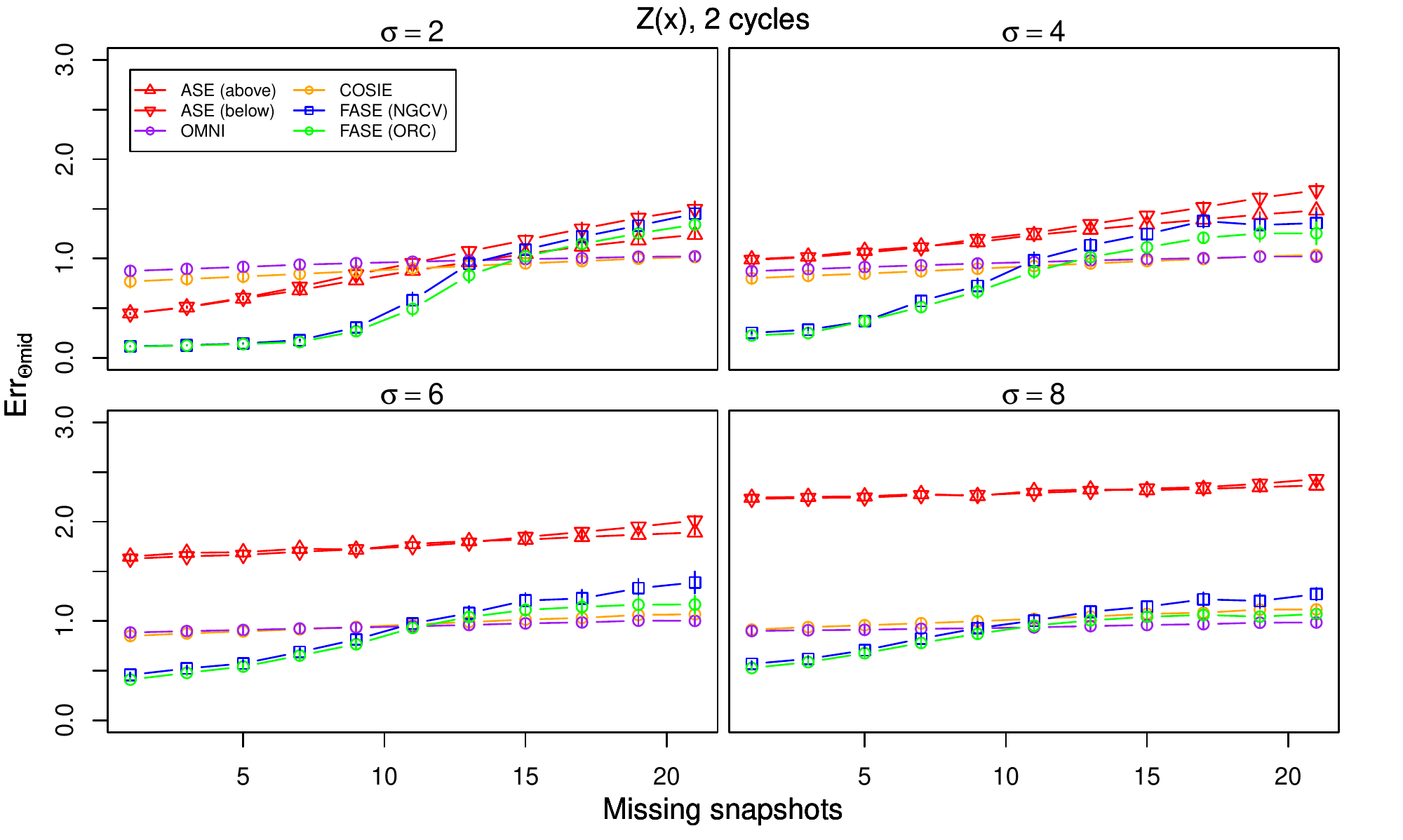}
  \caption{Prediction prediction performance for central unobserved snapshot, latent processes complete two cycles in $[0,1]$.}
  \label{fig:impute_zf2}
\end{figure}

We see that in this nonparametric setting, FASE performs well for interpolation of unobserved snapshots relative to competing dynamic network embeddings.
The relative performance of the different approaches is similar to the network recovery study in Section~\ref{sec:simulations}.
Comparing the results in Figures~\ref{fig:impute_zf1} and \ref{fig:impute_zf2}, FASE performs better when the underlying processes are smoother in the index space.
In particular, FASE gives the best performance of all methods for all settings in Figure~\ref{fig:impute_zf2}.
In Figure~\ref{fig:impute_zf2}, in almost all cases, FASE outperforms the ASE approaches, especially as $\sigma$ increases.
FASE outperforms or is competitive with COSIE and OMNI in most settings.
Both OMNI and COSIE produce nearly constant embeddings in this setting, meaning their bias is quite insensitive to the number of missing snapshots, while the bias of FASE is sensitive to the number of missing snapshots, and the smoothness of the underlying latent processes.

The performance of the NGCV-tuned FASE and the oracle FASE are generally comparable, except in Figure~\ref{fig:impute_zf2} with many missing snapshots and $\sigma \geq 6$. In this case, the oracle approach can choose a smaller $q$ to optimize for imputation, in contrast to the NGCV tuning, which prioritizes model fit on the observed data.

\section{Additional analyses of international relations} \label{app:real_data}

In this appendix we include some additional details of the analysis of international political interactions described in Section~\ref{sec:real_data}.
As described briefly in Section~\ref{sec:real_data}, we tune the model parameters $d$ and $q$ by finding a FASE estimate for each pair in a grid, and evaluating NGCV.
In particular, we vary $d$ between $1$ and $10$, incrementing by $1$, and vary $q$ between $4$ and $12$, incrementing by $1$.
The results are shown in Figure~\ref{fig:ngcv_heatmap}.
Importantly, we note that the NGCV criterion reaches a minimum on the interior of the grid, supporting the use of a functional embedding on this data, rather than one which is constant over time.

\begin{figure}
  \centering
  \includegraphics[width=\textwidth]{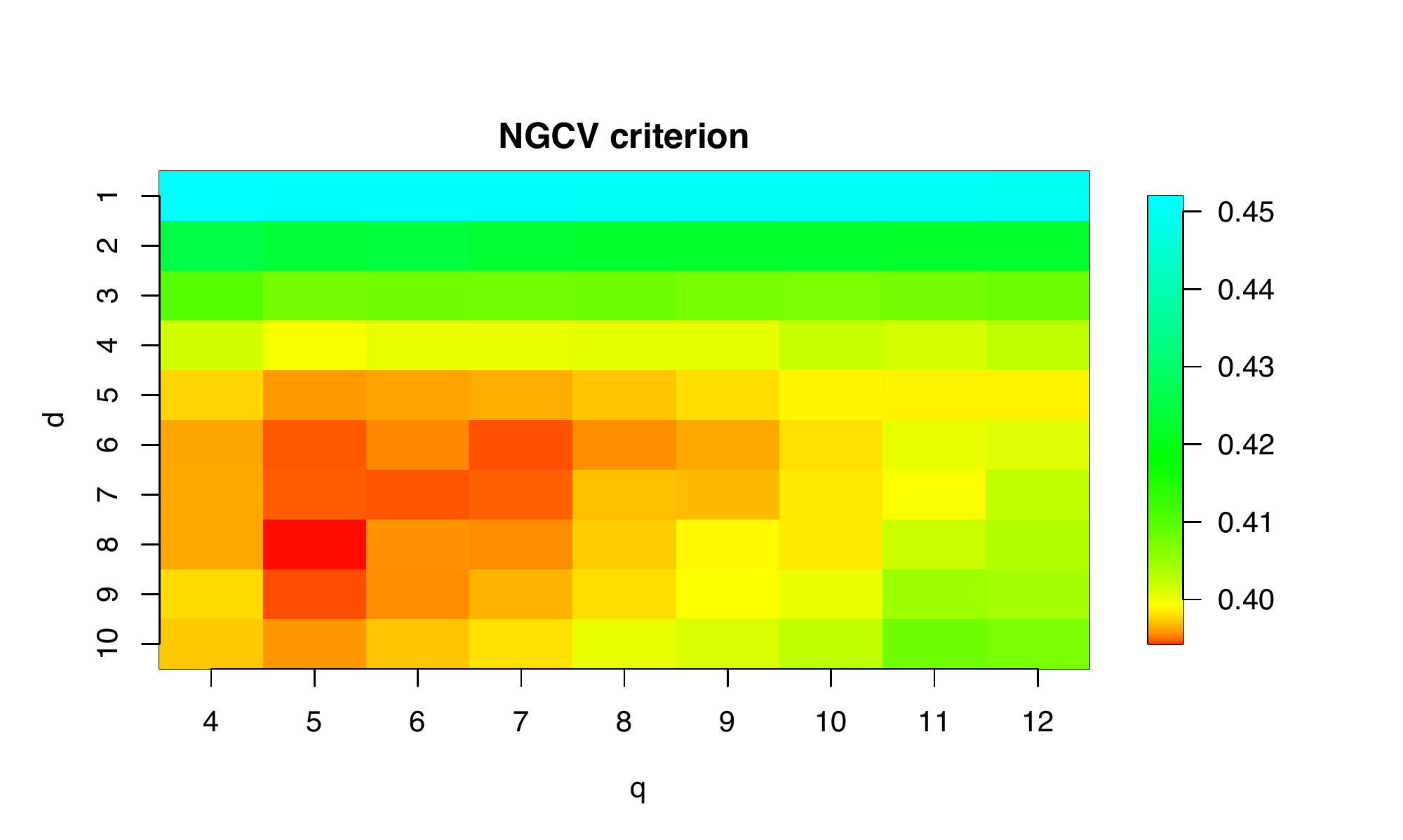}
  \caption{NGCV for FASE estimate, evaluated for each $(d,q)$ pair.}
  \label{fig:ngcv_heatmap}
\end{figure}

After selecting $\hat{d}=8$ and $\hat{q}=5$, we calculate our final estimator $\widehat{Z}$ and apply the sequential Procrustes alignment procedure described in Appendix~\ref{app:one_rotation}.
To unambiguously label the latent dimensions from $1$ to $8$, we evaluate an average magnitude
\begin{equation*}
  \frac{1}{m} \sum_{k=1}^m \sum_{i=1}^n \left\{ \hat{z}_{i,r}(x_k) \right\}^2
\end{equation*}
for each $r=1,\ldots,8$.
The largest average magnitude (dimension $1$) is about $46.3$, dimensions $2-4$ have smaller average magnitude between $12.9$ and $14.8$, and the remaining dimensions $5-8$ have average magnitudes between $6.7$ and $11.1$.

The first four estimated latent dimensions are plotted in Section~\ref{sec:real_data}, and we plot the remaining four here, and make some brief remarks on the embeddings.
Figure~\ref{fig:embed_d56} plots the fifth latent dimension against the sixth latent dimension, and Figure~\ref{fig:embed_d78} plots the seventh latent dimension against the eighth latent dimension.   The full evolution of the estimated latent processes can be seen in videos available online at \url{github.com/peterwmacd/fase/tree/main/videos}.

\begin{figure}
  \centering
  \includegraphics[width=\textwidth]{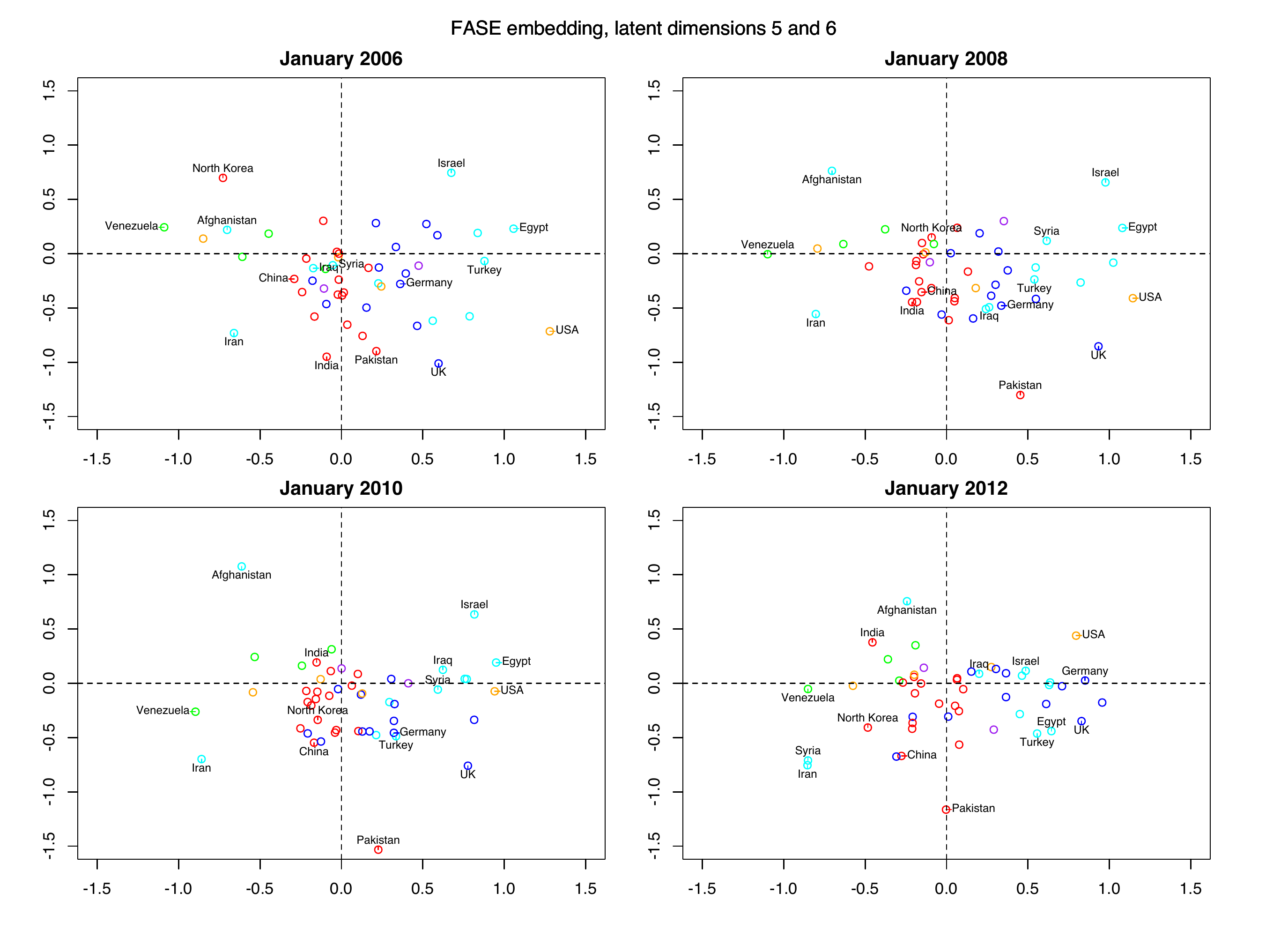}
  \caption{Fifth (horizontal axis) and sixth (vertical axis) dimensions of FASE evaluated at four times: January 2006, January 2008, January 2010, and January 2012. Points are colored by geographical region. Purple: Africa, Red: Asia-Pacific, Blue: Europe, Cyan: Middle East, Orange: North America, Green: South America.}
  \label{fig:embed_d56}
\end{figure}

\begin{figure}
  \centering
  \includegraphics[width=\textwidth]{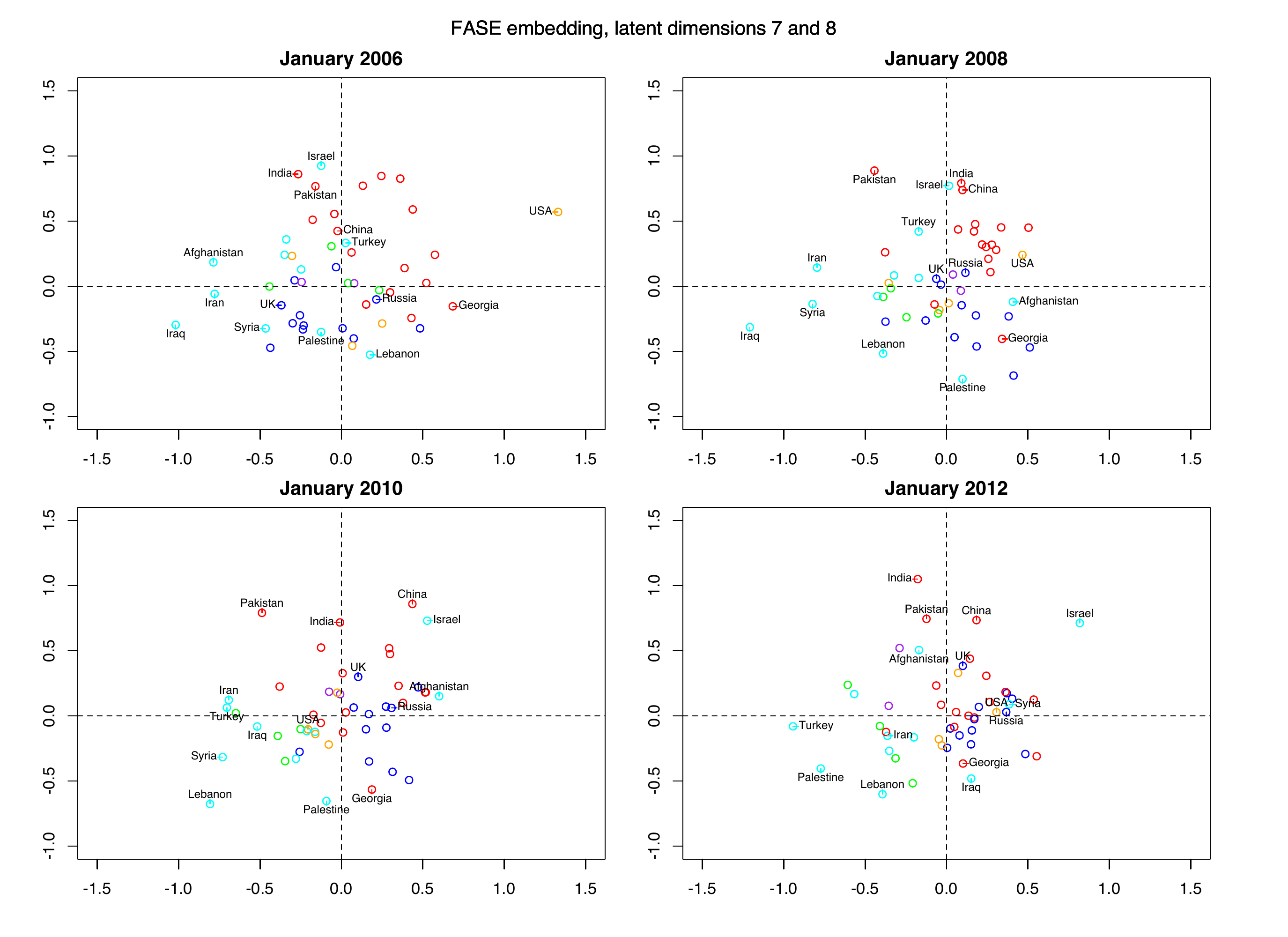}
  \caption{Seventh (horizontal axis) and eighth (vertical axis) dimensions of FASE evaluated at four times: January 2006, January 2008, January 2010, and January 2012. Points are colored by geographical region. Purple: Africa, Red: Asia-Pacific, Blue: Europe, Cyan: Middle East, Orange: North America, Green: South America.}
  \label{fig:embed_d78}
\end{figure}

In Figure~\ref{fig:embed_d56}, the USA and Venezuela are separated at extremes in the fifth dimension, while for much of the time period, the top right and bottom left quadrants separate countries with respect to a conflict between Israel and Iran.
As noted in Section~\ref{sec:real_data}, the fifth latent coordinate for Syria moves substantially, from the postive to negative half-plane between January 2010 and January 2012
\ifjasa
  \citeyearpar[BBC News,][]{11us}.
\else
\ifjmlr
  \citeyearpar[BBC News,][]{11us}.
\else
  \citep{11us}.
\fi
\fi
We also see movement, mostly in the sixth latent dimension of Afghanistan and India, reflecting worsening relations with Pakistan during this period
\ifjasa
  \citeyearpar[Reuters,][]{13tensions}.
\else
\ifjmlr
  \citeyearpar[Reuters,][]{13tensions}.
\else
  \citep{13tensions}.
\fi
\fi

In Figure~\ref{fig:embed_d78}, we again see a regional clusters formed by countries from Europe and Asia, although these begin to merge by the end of the time period.
The seventh dimension separates the USA and Iraq in January 2006 and January 2008, but similar to the conclusion from Figure~\ref{fig:embed_d34}, this conflict appears to have fully dissipated by January 2012, as the two countries have similar seventh latent coordinates, with the same sign.

\section{FASE with smoothing splines} \label{app:smoothing_splines}

Here we report preliminary results using an extension of the FASE methodology to {\em smoothing splines}, in which we select a maximal natural spline basis and optimize a penalized objective function.

Briefly, recall that the optimization problem introduced in Section~\ref{sec:estimation} minimizes a loss function $\ell(\mathcal{W})$ over coordinate tensors $\mathcal{W} \in \mathbb{R}^{n \times q \times d}$.
The vector-valued function $B(x) \in \mathbb{R}^q$ contains the $B$-spline basis for a $q$-dimensional cubic spline space.
We can rewrite this in a functional way if we let $\mathbb{S}_q^{n \times d}$  denote the space of functions from $\mathcal{X}$ to $\mathbb{R}^{n \times d}$ with components in $\operatorname{span}(B)$.
Then we can rewrite \eqref{snap_opt} as
\begin{equation*}
  \min_{Z \in \mathbb{S}_q^{n \times d}} \left\{ \sum_{k=1}^m \lVert A_k - Z(x_k)Z(x_k)^{\tp} \rVert_F^2 \right\}.
\end{equation*}

Following the usual development for smoothing splines, suppose we instead optimize over $Z$'s with components in the Sobolev space $\operatorname{Sob}_{2,2}^{n \times d}$ of twice-differentiable functions and add a penalty term
\begin{equation*}
  \operatorname{Pen}(Z) = \sum_{i=1}^n \sum_{r=1}^d \int_{\mathcal{X}} \left\{ z_{ir}^{''}(x) \right\}^2 dt
\end{equation*}
scaled by a penalty parameter $\lambda \geq 0$.
That is we solve the smoothing spline optimization problem
\begin{equation*}
  \min_{Z \in \operatorname{Sob}_{2,2}^{n \times d}} \left\{ \sum_{k=1}^m \lVert A_k - Z(x_k)Z(x_k)^{\tp} \rVert_F^2 + \lambda \operatorname{Pen}(Z) \right\}.
\end{equation*}

Classical results on smoothing splines \citep{hastie01elements} can be used to justify that this is equivalent to solving
\begin{equation} \label{optim_ss}
  \min_{Z \in \mathbb{N}_m^{n \times d}} \left\{ \sum_{k=1}^m \lVert A_k - Z(x_k)Z(x_k)^{\tp} \rVert_F^2 + \lambda \operatorname{Pen}(Z) \right\}
\end{equation}
where $\mathbb{N}_m^{n \times d}$ is the natural cubic spline with knots at the $x_k$'s for $k=1,\ldots,m$.
Hence, we can easily adapt FASE to these settings, solving \eqref{optim_ss} with gradient descent.
The penalty term can be evaluated in terms of integrals of the second derivatives of the natural spline basis functions, and rewritten as a quadratic function of the basis coordinates, hence its gradient is easy to calculate.

We compare this smoothing spline version of FASE (FASE (SS)) to the $B$-spline version developed in the body of the paper (FASE (BS)) through a small simulation study.
In particular, we generate functional networks as in scenario (ii) in Section~\ref{sec:simulations}, fixing $n=100$ and varying $m$ from $20$ to $200$ and $\sigma \in \{2,4,6,8\}$.
We perform $50$ replications for each setting, and evaluate the mean of $\mathrm{Err}_Z$ (see Section~\ref{sec:simulations}). 
We select the nonparametric scenario (ii) as it should favor the fully nonparametric smoothing spline version of FASE.

For each of FASE (SS) and FASE (BS), we fit an oracle version with $d$ fixed at the ground truth value and $\lambda$ and $q$ respectively selected from a grid to minimize $\mathrm{Err}_Z$.
In about 98\% of replications, the grids contain a local minimum of $\mathrm{Err}_Z$.
We try two initialization routines.
First, the usual initializer introduced in Section~\ref{subsec:initializers_new}. Second, an oracle initialization from the ground truth processes $Z$, or the closest approximation to each component in the $B$-spline space (ORC-init).
The results are shown in Figure~\ref{fig:ii_varym_smooth}. 

\begin{figure}
  \centering
  \includegraphics[width=\textwidth]{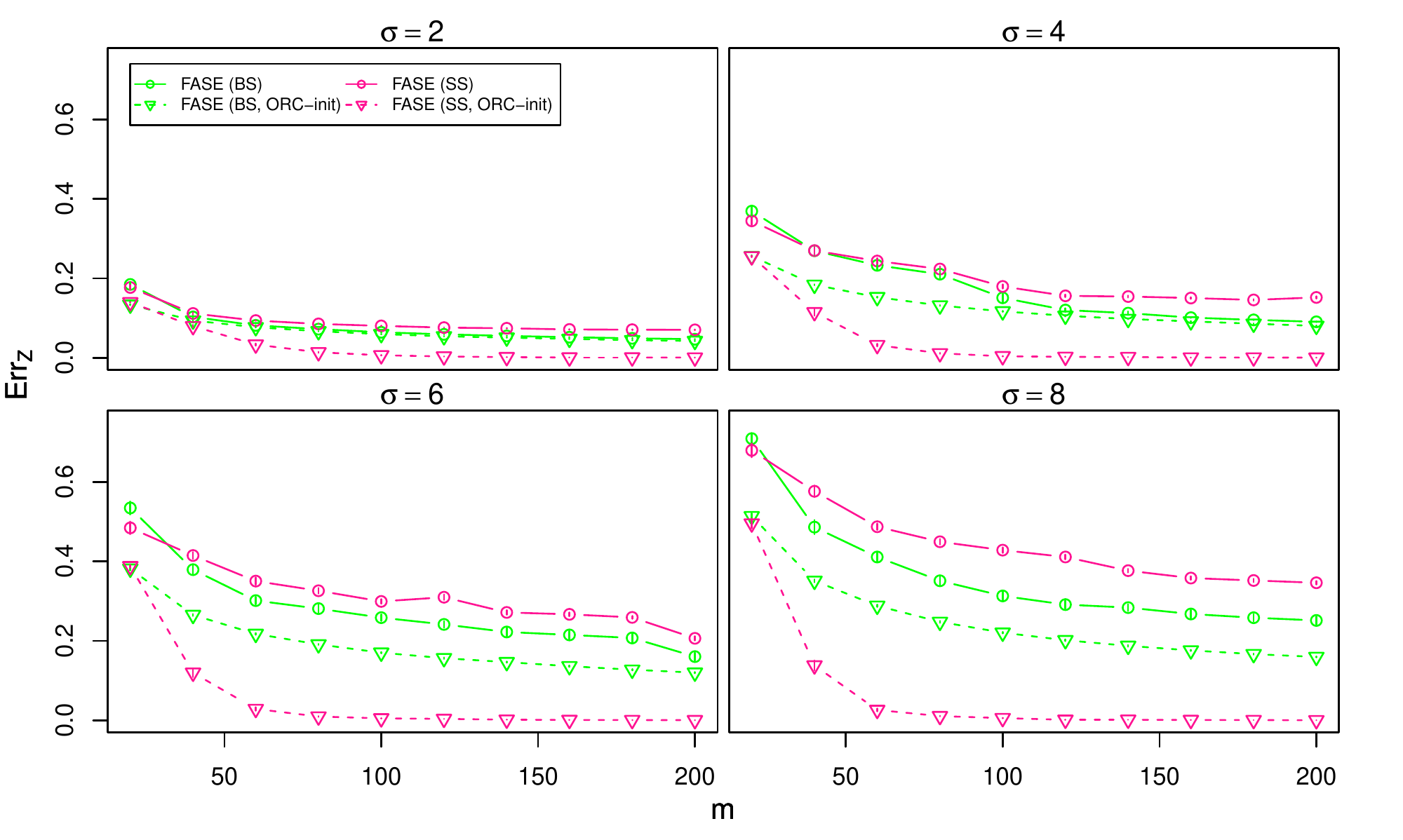}
  \caption{Mean of $\mathrm{Err}_Z$, varying $m$, the number of snapshots. Scenario (ii), nonparametric Gaussian networks. Plots are labeled by edge standard deviation $\sigma$.}
  \label{fig:ii_varym_smooth}
\end{figure}


In Figure~\ref{fig:ii_varym_smooth}, we see that with data-driven initialization, FASE (BS) outperforms FASE (SS) in terms of $\mathrm{Err}_Z$ except in the setting where $m=20$.
This ordering is reversed with oracle initialization.
In fact, the performance of FASE (SS) is insensitive to $\sigma$ for large $m$, implying that gradient descent is converging to a local minimum very close to the starting point.
This provides evidence that for large $m$ and $n$, as FASE (SS) must optimize far more parameters than FASE (BS), gradient descent becomes unreliable and highly dependent on the starting value.
While this does not preclude the existence of an efficient implementation of FASE (SS) which can overcome these optimization issues, it is not clear that such an implementation would substantially outperform FASE (BS).

\section{Sequential FASE} \label{app:stash_sequential}

In this section, we develop a sequential version of FASE that estimates one latent dimension at a time to overcome some of the identifiability issues.   As inputs, the sequential gradient descent algorithm takes a set of initial coordinates
\ifjasa
  $\widehat{\mathcal{W}}^0 = \{\widehat{\bm{W}}_r^0\}_{r=1}^d$;
\else
\begin{equation*}
  \widehat{\mathcal{W}}^0 = \{\widehat{\bm{W}}_r^0\}_{r=1}^d;
\end{equation*}
\fi
step sizes $\eta_{h,r} > 0$, which may depend on the iteration number $h \geq 0$; the latent dimension $d$; and a maximum iteration number $H$.
In practice, for both gradient descent schemes we will choose $H$ based on the convergence of the value of $\ell$ to a local minimum.

\begin{algorithm} \label{alg:gd_seq}
  \caption{Sequential gradient descent algorithm.}
  \begin{tabbing}
  \enspace Set $\widetilde{\mathcal{W}}^0 = \bm{0}_{n \times q \times d}$ \\
  \enspace For $r=1$ to $r=d$ \\
    \qquad $\widetilde{\bm{W}}_{r}^0 \leftarrow \widehat{\bm{W}}_{r}^{0}$ \\
    \qquad For $h=1$ to $h=H$ \\
      \qquad\qquad $\widetilde{\bm{W}}_r^{h} \leftarrow \widetilde{\bm{W}}_r^{h-1} - \eta_{h-1,r} \frac{\partial \ell}{\partial \bm{W}_r}(\widetilde{\mathcal{W}}^{h-1})$ \\
    \qquad $\widetilde{\bm{W}}_r^{0} \leftarrow \widetilde{\bm{W}}_{r}^{H}$ \\
  \enspace Output $\widehat{\mathcal{W}}^{H} = \{\widetilde{\bm{W}}_r^H\}_{r=1}^d$
  \end{tabbing}
\end{algorithm}

The output of Algorithm~\ref{alg:gd_seq} is an $n \times q \times d$ tensor-valued coordinate estimator $\widehat{\mathcal{W}}^{H}$.
Algorithm~\ref{alg:gd_seq} computes the coordinate estimator $\widehat{\mathcal{W}}^{H}$ one $n \times q$ slice at a time by estimating $d$ one-dimensional latent process models  sequentially.
Once a slice is estimated, it remains fixed, and its contribution is subtracted from the network structure.   Slices which have not yet been estimated have all entries fixed at $0$.

Empirically, we have found that when the singular values of the true latent processes are separated uniformly (in $x$), this sequential approach can achieve better estimation performance, which we attribute to it essentially reducing the space of unknown orthogonal transformations.
Theoretical results can be proven, analogous to Corollary~\ref{cor:main_concur1}, that recover each dimension in sequence with stronger guarantees on the alignment of the estimated and true latent processes.
However, extending a result like Corollary~\ref{cor:main_concur1} to higher dimensions with this sequential estimation scheme requires strong assumptions on the separation of the latent dimensions, in particular that the singular values of each $Z(x)$ are separated uniformly (in $x$), so that the previously estimated or yet to be estimated dimensions do not interfere with the estimation of the current slice.
We omit these results, as in general we recommend the concurrent gradient descent estimator for FASE presented in Section~\ref{sec:estimation}.

\fi
\fi

\end{document}